\documentclass[12pt]{article}
\usepackage{balance}
\usepackage{graphics}
\usepackage{graphicx}
\usepackage{amssymb}
\usepackage{verbatim}     
\usepackage{amsxtra}
\usepackage{epsfig}
\usepackage{subfigure}
\usepackage{amsmath}
\usepackage{latexsym}
\usepackage{ifthen}
\usepackage{psfrag}
\usepackage{subfigure}
\usepackage{color}
\usepackage{hyperref} 
\usepackage{float}

\usepackage[crop=pdfcrop]{pstool}

\usepackage{rgsEnvironments}
\usepackage{rgsMacros}

\graphicspath{{Figures/}}
\newcommand{\Z}{{\cal{Z}}}

\newboolean{Extended}
\setboolean{Extended}{false}
\newcommand{\ExtendedVersion}[1]{\ifthenelse{\boolean{Extended}}{}{#1}}
\newcommand{\NotForExtendedVersion}[1]{\ifthenelse{\boolean{Extended}}{#1}{}}

\usepackage{hdclabTR}
\newenvironment{proof}      
{\par\noindent\textbf{Proof:}}{\eop\smallskip\vskip 3 pt}  
\usepackage[left=1in,top=1in,right=1in,bottom=1in]{geometry}

\begin{document}

%%---------------------------------
%% COMMENT OUT UP TO \makeititle FOR REPORT FORMAT [REP]
\ititle{\bf Dynamical Properties of a Two-gene Network with Hysteresis}
\iauthor{
  Qin Shu \\
  {\normalsize shuq@email.arizona.edu} \\
    Ricardo G. Sanfelice \\
  {\normalsize sricardo@u.arizona.edu}}
\idate{\today{}} 
\iyear{2014}
\irefnr{001. {\bf Status}: NOT PUBLISHED. {\bf Readers of this material have the responsibility to inform all of the authors promptly if they wish to reuse, modify, correct, publish, or distribute any portion of this report.}}
\makeititle
%%---------------------------------

\title{\IfConf{
Modeling and Analysis of a Genetic Network\\ with Binary Hysteresis using Hybrid System Tools
}
{
\bf Dynamical Properties of a Two-gene Network with Hysteresis
}
}

\author{Qin Shu and Ricardo G. Sanfelice\thanks{Q. Shu and R. G. Sanfelice are with
the Department of Aerospace and Mechanical Engineering, University of Arizona
1130 N. Mountain Ave, AZ 85721.
      Email: {\tt\small shuq@email.arizona.edu, sricardo@u.arizona.edu.}
      This research has been partially supported by the National Science Foundation under CAREER Grant no. ECS-1150306 and by the Air Force Office of Scientific Research under Grant no. FA9550-12-1-0366.
}}

\maketitle

\begin{abstract}
A mathematical model for a two-gene regulatory network is 
derived and several of their properties analyzed. 
Due to the presence of mixed continuous/discrete dynamics and hysteresis,
we employ a hybrid systems model to capture the dynamics of the system. 
The proposed model incorporates binary hysteresis with different thresholds capturing the interaction between the genes. 
We analyze properties of the solutions and asymptotic stability of equilibria in the system as a function of its parameters. 
Our analysis reveals the presence of limit cycles for a certain range of parameters, behavior that is associated with hysteresis. The set of points defining the limit cycle is characterized
and its asymptotic stability properties are studied. 
Furthermore, the stability property of the limit cycle is robust to 
small perturbations.
Numerical simulations are presented to illustrate the results.
\end{abstract}

\NotForConf{
\newpage

\tableofcontents

\newpage
}

%% main text
\section{Introduction}

\IfConf{
%%%% INTRODUCTION FOR CONFERENCE
Several mathematical models have been proposed in the literature for the study of genetic regulatory networks; see \cite{de2002modeling} for a survey.
In particular, boolean models are typically used to capture the dynamics of discrete switches in such networks.  
As introduced by Glass and Kauffman in \cite{98}, Boolean regulation functions, typically modeled as sigmoidal or step functions, can be combined with linear system models to enforce certain logic rules. The properties of such class of piecewise linear models have been studied in the mathematical biology literature, e.g.,  \cite{97,96,99,95}. Snoussi presented a discrete mapping approach in \cite{97} to study the qualitative properties of the dynamics of genetic regulatory networks. In his work, the properties of the discrete mapping were studied to determine stable isolated steady states as well as limit cycles. In  \cite{96}, Gouz$\acute{e}$ and Sari employ the concept of Filippov solution to study piecewise linear models of genetic regulatory networks with discontinuities occurring on hyperplanes defined by thresholds on the variables. 
 Chaves and coauthors \cite{99} studied the robustness of Boolean models of gene control networks. 
In \cite{95}, de Jong and coauthors presented a method for qualitative simulation of genetic regulatory networks based on the piecewise linear model of \cite{98}. Genetic regulatory networks with continuous dynamics coupled with switching can be written as a hybrid system. In \cite{94} and \cite{93}, the authors apply hybrid systems tools to model a variety of cell biology problems. 

Although it is an important phenomenon present
in genetic regulatory networks, 
hysteresis behavior is not usually included in models of such networks.
Hysteresis is characterized by behavior in which, for instance,
once a gene has been inhibited due to the concentration of cellular protein reaching a particularly low value, a higher value of cellular protein concentration is required to express it. 
In his survey paper on the impact of genetic modeling on tumorigenesis and drug discovery \cite{92}, Huang stated that {\em ``hysteresis is a feature that a synthetic model has to capture.''} Through experiments, Das and coauthors \cite{91} demonstrated the existence of hysteresis in lymphoid cells and the interaction of continuous evolution of some cellular proteins.  Hysteresis was also found to be present in mammalian genetic regulatory networks; see, e.g., \cite{90,89}. More importantly, hysteresis is a key mechanism contributing to oscillatory behavior in biological models \cite{88}, \cite{87}. 

In this paper, we propose a hybrid system model that captures both continuous and discrete dynamics of a genetic regulatory network with binary hysteresis. We combine the methodology of piecewise linear modeling of genetic regulatory networks with the framework of hybrid dynamical systems in \cite{84}, and construct a hybrid system model for a genetic network with two genes; see Section~\ref{sec:2}. 
Unlike piecewise linear models, our model incorporates hysteresis explicitly. 
We prove existence of solutions to the genetic network, a property that is typically overlooked or difficult to prove due to the discontinuity in the dynamics introduced by boolean variables.
In Section \ref{sec:3}, we compute the equilibria of the system in terms of its parameters. We analyze the asymptotic stability of the isolated equilibrium points and determine conditions under which a limit cycle exists.  It is found that, for a particular set of parameters, hysteresis is the key mechanism leading to oscillations, as without hysteresis, the limit cycle converges
to an isolated equilibrium point (cf. \cite{97}).  The stability of the limit cycle is established using a novel approach consisting of measuring the distance between solutions
of hybrid systems (rather than the distance to the limit cycle as in classical continuous-time systems).  The asymptotic stability of the limit cycle is found to be robust to small perturbations.
In Section \ref{sec:4}, simulations validating some of our results are presented.
%An extended version of this submission, including proofs and other results, is available \cite{Shu.Sanfelice.13.TR}.
}
{
%%%% INTRODUCTION FOR JOURNAL
\subsection{Mathematical modeling of genetic regulatory networks}

In recent years, the development of advanced experimental techniques in molecular biology has led to a growing interest in mathematical modeling methods for the study of genetic regulatory networks; see \cite{de2002modeling} for a literature review. A number of gene regulatory network models have been proposed to capture their main properties \cite{99}, \cite{98}, \cite{97}, \cite{96}, \cite{95}, \cite{94}, \cite{93}. Boolean models capture the dynamics of the discrete switch in genetic networks.  
As introduced by Glass and Kauffman in \cite{98}, Boolean regulation functions, typically modeled as sigmoidal or step functions, can be combined with linear system models to enforce certain logic rules. The properties of such a class of piecewise linear models have been studied in the mathematical biology literature, e.g.,  \cite{97,96,99,95}. Snoussi presented a discrete mapping approach in \cite{97} to study the qualitative properties of the dynamics of genetic regulatory networks. In this work, the properties of the discrete mapping were studied to determine stable isolated steady states as well as limit cycles. In  \cite{96}, Gouz$\acute{e}$ and Sari employ the concept of Filippov solution to study piecewise linear models of genetic regulatory networks with discontinuities occurring on hyperplanes defined by thresholds on the variables. 
 Chaves and coauthors \cite{99} studied the robustness of Boolean models of gene control networks. 
de Jong and coauthors \cite{95} presented a method for qualitative simulation of genetic regulatory networks based on the piecewise linear model of \cite{98}. Genetic regulatory networks with continuous dynamics coupled with switching can be written as a hybrid system. In \cite{94} and \cite{93}, the authors apply hybrid systems tools to model a variety of cell biology problems. More recently, hybrid models have been used in \cite{NoelVincent.2012} for the study of molecular interactions.
It is important to note that
hysteresis behavior, which is typically present in genetic regulatory networks, has not been considered
in the models mentioned above.
   
\subsection{The role of hysteresis in genetic regulatory networks}

Hysteresis is an important phenomenon in genetic regulatory networks. It is characterized by behavior in which, for instance,
once a gene has been inhibited due to the concentration of cellular protein reaching a particularly low value,
a higher value of cellular protein concentration is required to express it. In his survey paper on the impact of genetic modeling on tumorigenesis and drug discovery \cite{92}, Huang stated that {\em ``hysteresis is a feature that a synthetic model has to capture.''} Through experiments, Das and coauthors \cite{91} demonstrated the existence of hysteresis in lymphoid cells and the interaction of continuous evolution of some cellular proteins.  Hysteresis was also found to be present in mammalian genetic regulatory networks; see, e.g., \cite{90,89}. More importantly, it has been observed that hysteresis is a key mechanism contributing to oscillatory behavior in computational biological models \cite{88}, \cite{87}. On the other hand, it is well known that hysteresis is one of the key factors that makes a system robust to noise and parametric uncertainties \cite{86}, \cite{85}.

\subsection{Contributions and organization of the paper}

Our work is motivated by the following facts:
\begin{enumerate}
\item {\em Piecewise linear models do not incorporate hysteresis, although
it plays a key role in the dynamics of genetic regulatory networks.  In fact, as we establish in this paper, hysteresis leads to oscillatory, robust behavior in two-gene networks}.
\item {\em The discontinuities introduced by the Boolean regulation functions yield a non-smooth dynamical system, for which classical analysis tools cannot be applied to study existence of solutions, stability, robustness, etc.}
\end{enumerate}
Motivated by these two limitations, we propose a hybrid system model that captures both continuous and discrete dynamics of genetic regulatory networks with hysteresis behavior. We combine the methodology of piecewise linear modeling of genetic regulatory networks with the framework of hybrid dynamical systems in \cite{84}, and construct a hybrid system model for a genetic network with two genes. Our model incorporates hysteresis explicitly, which we found leads to limit cycles. We prove existence of solutions and compute the equilibrium points in terms of parameters for the system. We analyze the stability of the isolated equilibrium points and determined conditions under 
which a limit cycle exists.  It is found that hysteresis is the key mechanism leading to hysteresis, as without hysteresis, the limit cycle converges
to an isolated equilibrium point (cf. \cite{97}).  The stability of the limit cycle is established using a novel approach consisting of measuring distance between solutions
of hybrid systems (rather than the distance to the limit cycle as in classical continuous-time systems).  Moreover, we show that the asymptotic stability of the limit cycle is robust to small perturbations.

The remainder of this paper is organized as follows. In Section~\ref{sec:2}, a mathematical framework of hybrid dynamical system is introduced 
and then applied to model a two-gene network. The analysis of existence of solutions, stability, and robustness are presented in Section \ref{sec:3}. 
Section \ref{sec:4} presents simulations validating our results. 
%\NotForExtendedVersion{
%An extended version of this submission, including proofs of the results that are not included here due to space constraints, 
%is available at \href{http://www.u.arizona.edu/~sricardo/index.php?n=Main.TechnicalReports}{http://www.u.arizona.edu/$\sim$sricardo/index.php?n=Main.TechnicalReports}}
%
}
\section[A Hybrid Model for Genetic Networks w/Hysteresis]{A Hybrid Systems Model for Genetic Regulatory Networks with Hysteresis}
\label{sec:2}

Models of genetic regulatory networks given by piecewise-linear differential equations have been proposed in \cite{93}, \cite{82}.  Such models take the form 
\footnote{The notation $x \geq0$ is equivalent to $x_i \geq 0$ for each $i$.}
\begin{equation}\label{eqn:1}
\dot{x}=f(x)-\gamma x, \qquad x\geq 0,
\end{equation}
where $x=[x_1, x_2, \ldots, x_n]^\top$ and $x_i$ represents the concentration of the protein in the $i$-$th$ cell,  $f=[f_1, f_2, \ldots, f_n]^\top$ is a function, $\gamma=[\gamma_1,\gamma_2, \ldots, \gamma_n]^\top$ is a vector of constants, and $1\leq i \leq n$. For each $i$, $f_i$ is a function representing the rate of synthesis, while $\gamma_i$ represents the degradation rate constant of the protein. The function $f_i $ is typically defined as the linear combination
$f_i(x)=\sum_{\ell\in L}k_{i\ell}b_{i\ell}(x)$
where $k_{i\ell}$ is the nonzero and nonnegative growth rate constants, $b_{i\ell}$ is a Boolean regulation function that describes the gene regulation logic, and $L=\{1,2 \dots, n\}$ is the set of indices of regulation functions.

The modeling strategy for the Boolean regulation functions $b_{il}$ is a key element that captures the behavior of a particular genetic regulatory network. A major feature of a genetic regulatory network is the presence of threshold-like relationships between the system variables, i.e., if a variable $x_i$ is above (or below) a certain level, it could cause little or no effect on another variable $x_j,$ whereas if $x_i$ is below (or above) this certain value, the effect on $x_j$ would become more significant (for example, it may increase the value of $x_j$ or inhibit the growth of the value of $x_j$). Boolean regulation functions can be modeled by sigmoidal or step functions, an approach that was first proposed by Glass and Kauffmann \cite{98}. When modeling as a step function, the functions $b_{i\ell}$ are given by the combination (linear or nonlinear) of
\begin{equation}
\label{eqn:4}
s^+(x_i,\theta)=\left\{\begin{array}{ll}
1&\textrm{if $x_i \geq\theta$}\\
0&\textrm{if $x_i< \theta$}
\end{array}\right., \quad
s^{-}(x_i,\theta)=1-s^{+}(x_i, \theta),
\end{equation}
where $s^{+}(x_i, \theta)$ represents the logic for gene expression when the protein concentration exceeds a threshold $\theta$, while $s^{-}(x_i, \theta)$ represents the logic for gene inhibition.

To illustrate this modeling approach, let us consider the genetic regulatory network  shown in Figure \ref{fig:1}. Genes ${\emph a}$ and {\emph b} encode proteins {\emph A} and {\emph B}, respectively. When the concentration of protein {\emph A} is below certain threshold, it will inhibit gene {\emph b}. Similarly,  protein {\emph B} inhibits gene {\emph a} when the concentration of protein {\emph B} is above certain threshold. In this way, a set of piecewise-linear differential equations representing the behavior in Figure \ref{fig:1} is given by
\begin{equation}\label{eqn:PWsystem}
\begin{array}{lll}
\dot{x}_1=k_1s^-(x_2, \theta_2)-\gamma_1x_1, \IfConf{\ \ }{\qquad}
\dot{x}_2=k_2s^+(x_1, \theta_1)-\gamma_2x_2,
\end{array}
\end{equation}
where $x_1$ is representing the concentration of protein $\emph A$, while $x_2$ is the concentration of protein {\emph B}. The constants $\theta_1,$ $\theta_2$ are the thresholds associated with concentrations of protein $\emph A$ and $\emph B,$ respectively.

%\begin{figure}[h]
%\begin{center}
%\psfrag{a}[][][0.9]{$a$}
%\psfrag{b}[][][0.9]{$b$}
%\psfrag{A}[][][0.9]{$A$}
%\psfrag{B}[][][0.9]{$B$}
%\includegraphics[width=0.7\columnwidth]{00}
%\end{center}
%\caption{\emph{A genetic regulatory network of two genes (a and b), each encoding for a protein (A and B). Lines ending in arrows represent genetic expression triggers, while lines ending in flatheads refer to genetic inhibition triggers.}}
%\label{fig:1}
%\end{figure}
\begin{figure}[h]
\begin{center}
  \psfragfig*[width=0.7\columnwidth]{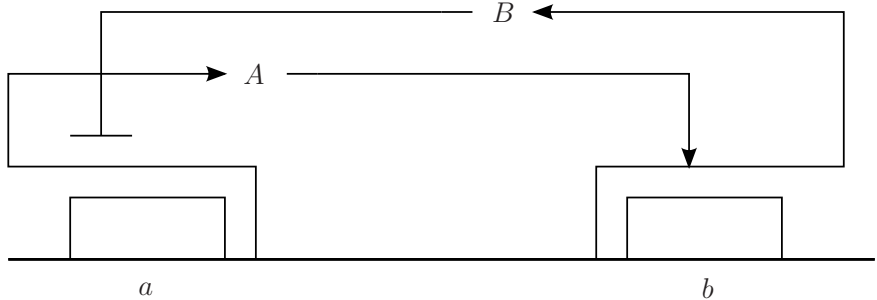}
   {
\psfrag{a}[][][0.9]{$a$}
\psfrag{b}[][][0.9]{$b$}
\psfrag{A}[][][0.9]{$A$}
\psfrag{B}[][][0.9]{$B$}
  }
\end{center}
\caption{\emph{A genetic regulatory network of two genes (a and b), each encoding for a protein (A and B). Lines ending in arrows represent genetic expression triggers, while lines ending in flatheads refer to genetic inhibition triggers.}}
\label{fig:1}
\end{figure}

In this model, gene $\emph a$ is expressed at a rate $k_1$ when $x_2$ is below the threshold $\theta_2$. Similarly, gene $\emph b$ is expressed at a rate $k_2$ when $x_1$ is above the threshold $\theta_1$. 
Degradations of both proteins are assumed to be proportional to their own concentrations, a mechanism that is captured by $-\gamma_1x_1$  and $-\gamma_2x_2,$ respectively. 

Note that the model in \eqref{eqn:PWsystem} capturing the interaction between gene {\em a} and gene {\em b} does not incorporate binary hysteresis. Furthermore, due to the discontinuities introduced by the Boolean regulation functions, it is not straightforward to argue that solutions to \eqref{eqn:PWsystem} exist from every initial value of $x$. 
In order to overcome such limitations, we propose a hybrid system with hysteresis for this two gene genetic regulatory network, to which hybrid systems tools for analysis of existence of solutions and asymptotic stability can be applied.

\makeatletter
\newcommand{\rmnum}[1]{\romannumeral #1}
\newcommand{\Rmnum}[1]{\expandafter\@slowromancap\romannumeral #1@}
\makeatother

\NotForConf{
\subsection{Introduction to Hybrid System Modeling}
\label{sec:2.1}

Following \cite{84} and \cite{83}, a hybrid system in this paper is defined by four objects:
\begin{itemize}
\item A set $C\subset\mathbb{R}^{n}$, called the {\em flow set}.
\item A set $D\subset\mathbb{R}^{n}$, called the {\em jump set}.
\item A single-valued mapping $F$: $\mathbb{R}^{n}\rightarrow\mathbb{R}^{n}$, called the {\em flow map}. 
\item A set-valued mapping $G$: $\mathbb{R}^{n}\rightrightarrows\mathbb{R}^{n}$, called the {\em jump map}.
\end{itemize}
The flow map $F$ defines the continuous dynamics on the flow set $C$, while the jump map $G$ defines the discrete dynamics or jumps on the jump set $D$. These objects are referred to as the data of the hybrid system $\cal{H}$. Then, defining $z\in\mathbb{R}^n$ to be the state of the system, $\cal{H}$ can be written in the compact form
$$
{\cal H}: \left\{\begin{array}{ll}
\dot{z}= F(z)&\textrm{$z\in C$}\\
z^+\in G(z) &\textrm{$z\in D$}\end{array}\right.
$$

Solutions to hybrid systems are given by hybrid arcs which are trajectories defined on hybrid time domains.
\begin{definition}[hybrid time domain]
\label{hybrid time domain definition} 
 A set $E$ is a {\em hybrid time domain} if for all $(T,J)\in E, E \cap ([0,T] \times \{0, 1, ... , J\} )$ is a compact hybrid time domain; that is, it can be written as $\cup_{j=0}^{j-1}([t_j, t_{j+1}], j)$ for some finite sequence of times $0\leq t_0 \leq t_1\leq \ldots \leq t_j$. 
\end{definition}

\begin{definition}[hybrid arc]
\label{def:hybrid arc}
A hybrid arc $\phi$ is a function that takes values from $\mathbb{R}^n$, is defined on a hybrid time domain $\dom \phi$, and is such that $t\mapsto\phi (t, j)$ is locally absolutely continuous for every $j$, $(t, j)\in \dom \phi$. 
\end{definition}

Hybrid time domains impose a specific structure on the domains of solutions to hybrid systems.
In simple words, solutions to $\HS$ are defined on intervals of flow $[t_j,t_{j+1}]$ 
indexed by the jump time $j$ when $t_{j+1} > t_j$.  Hybrid arcs specify the functions 
that define solutions to hybrid systems when the following conditions are satisfied.
We refer the reader to \cite{83,84} for more details on the definition of solutions to hybrid systems.
 
\begin{definition}[solution]
\label{def:solution}
A hybrid arc $\phi$ is a solution to the hybrid system ${\cal H}$ if $\phi (0, 0)\in\overline{C}\cup D$ and \\
 $(S1)$ For all $j\in\mathbb{N}:= \{0,1,2,\ldots\}$ and almost all $t$ such that $(t, j)\in  \dom \phi$,$$\phi (t, j)\in C, \quad\phi (t, j)= F(\phi (t, j))$$\\
$(S2)$ For all $(t, j)\in  \dom \phi$ such that $(t, j+1)\in  \dom \phi$, $$\phi (t, j)\in D,\quad\phi (t, j+1)\in G(\phi (t, j))$$
\end{definition}

Solutions to hybrid systems are classified as follows:
\begin{itemize}
\item A solution $\phi$ to ${\cal H}$ is said to be nontrivial if $\dom \phi$ contains at least two points.
\item A solution $\phi$ to ${\cal H}$ is said to be complete if $\dom \phi$ is unbounded.
\item A solution $\phi$ to ${\cal H}$ is said to be Zeno if it is complete and the projection of $ \dom \phi$ onto $\mathbb{R}^n_{\geq0}$ is bounded.
\item A solution $\phi$ to ${\cal H}$ is said to be maximal if there does not exist another solution $\varphi$ to ${\cal H}$ such that $\dom \varphi$ is a proper subset of $\dom \phi$, and $\varphi (t, j)=\phi (t, j)$ for all $(t, j)\in \dom \phi$.
\end{itemize}
%}
The reader is referred to \cite{84} and \cite{83} for more details on this hybrid system framework.

\subsection{Modeling of a Two-Gene Network}
\label{sec:2.2}
}

\IfConf{
We model the genetic network in \eqref{eqn:PWsystem} as a hybrid system $\cal{H}$
within the formalism of \cite{83,84}, where hybrid systems are given in terms of 
a flow map $F$, a flow set $C$, a jump map $G$, and a jump set $D$, and solutions
are parameterized by flow time $t$ and jump time $j$. To this end,
%Solutions to hybrid systems $\HS$ will be given in terms 
%of hybrid arcs on hybrid time domains.
}
{
To model the genetic network in \eqref{eqn:PWsystem} as a hybrid system $\cal{H}$, 
}
two discrete logic variables, $q_1$ and $q_2$, are introduced. The dynamics of these variables depend on the thresholds, $\theta_1$ and $\theta_2$, respectively. As one of our goals is to introduce binary hysteresis in the model in \eqref{eqn:PWsystem}, we define hysteresis level constants $h_1$ and $h_2$  associated with gene $\emph a$ and gene $\emph b$, respectively. In this way, $q_i$ is governed by dynamics such that the evolution in Figure~\ref{fig:update} holds.

\begin{figure}[htbp]
\begin{center}
\psfragfig*[width=0.55\columnwidth]{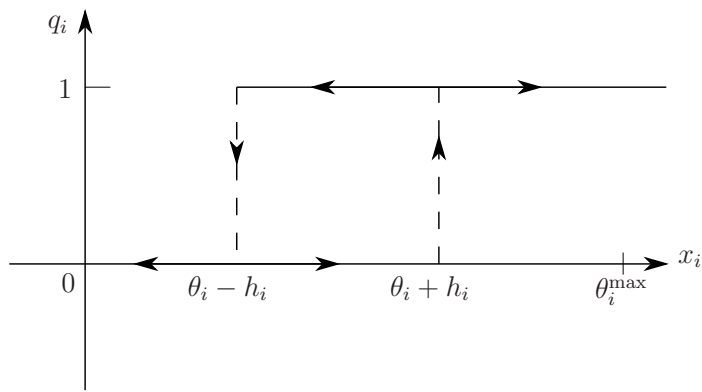}
  {
\psfrag{q}[][][0.9]{$q_i$}
\psfrag{0}[][][0.9]{$0$}
\psfrag{xi}[][][0.9]{$x_i$}
\psfrag{1}[][][0.9]{$1$}
\psfrag{t1}[][][0.9]{$\theta_i-h_i$}
\psfrag{t2}[][][0.9]{$\theta_i+h_i$}
\psfrag{t3}[][][0.9]{$\theta_i^{\max}$}
  }
\end{center}
\caption{\emph{The update mechanism of $q_i$ as a function of $x_i$ and previous values of $q_i.$}}
\label{fig:update}
\end{figure}

The state of the hybrid system is defined as
\begin{equation}\non
z=[x_1, x_2, q_1, q_2]^\top,
\end{equation}
where $z\in \Z:=\mathbb{R}^{2}_{\geq0}\times\{0, 1\}^2$;  $x_1$, $x_2$ are (nonnegative) continuous states representing protein concentrations; and $q_1$, $q_2$ are discrete variables. Here, $\mathbb{R}_{\geq0} := [0,+\infty)$.
We specify constants $\theta_1$ and $\theta_2$, usually inferred from biological data, satisfying $0<\theta_1<\theta_1^{\max}, 0<\theta_2<\theta_2^{\max}$, where $\theta^{\max}_1$ and $\theta_2^{\max}$ are the maximal value of the concentration of protein {\emph A} and of the protein {\emph B}, respectively.

To define the continuous dynamics of the hybrid system capturing the evolution of \eqref{eqn:PWsystem}, we rewrite the piecewise-linear differential equation \eqref{eqn:PWsystem} by replacing the $s^+$ term with the logic variables $q_i$, and the $s^-$ term with the complement of the logic variable $q_i$, i.e., $1-q_i$. Note that the discrete logic variables $q_i$ only change at jumps, i.e., they are constants during flows. Then, $\dot{q}_i=0.$ In this way, the continuous dynamics are governed by the differential equation
\begin{equation}\non
\begin{array}{lll}
\dot{x}_1=k_1(1-q_2)-\gamma_1x_1, \qquad
\dot{x}_2=k_2q_1-\gamma_2x_2, \qquad \IfConf{\\}{}
\dot{q}_1=\dot{q}_2=0,
\end{array}
\end{equation}
from where we obtain the flow map
\begin{equation}
F(z)=\matt{
k_1(1-q_2)-\gamma_1x_1\\
k_2q_1-\gamma_2x_2\\
0\\ 0}.
\end{equation}
Now, we describe the discrete update of the state vector $z$, i.e., we define $G$ and $D$. To illustrate this construction, we explain how to model  the mechanism in Figure~\ref{fig:update} for $q_1$. When
$$q_1=0\quad \mbox{ and }\quad  x_1 = \theta_1+h_1$$
the state $q_1$ is updated to 1. We write this update law as
$$q^+_1=1.$$
When
$$q_1=1\quad \mbox{ and }\quad x_1 = \theta_1-h_1,$$
then the state $q_1$ is updated to 0, i.e.,
$$q_1^+=0.$$
It follows that the mechanism of $q_1$ in Figure \ref{fig:update} can be captured by triggering jumps when the components of $z$ satisfy
$$ q_1=0, \ \  x_1 = \theta_1+h_1\qquad \mbox{ or } \qquad q_1=1,\ \  x_1 = \theta_1-h_1$$
Note that the update mechanism for $q_2$ is similar to that of $q_1$ just discussed.

We can define the flow and jump sets in a compact form by defining functions
$$\eta_1(x_1, q_1):=(2q_1-1)(-x_1+\theta_1+(1-2q_1)h_1)$$
$$\eta_2(x_2, q_2):=(2q_2-1)(-x_2+\theta_2+(1-2q_2)h_2).$$
In this way, 
the flow set is given by
\begin{equation}\label{eqn:C}
C:=\{z \in \Z:\eta_1(x_1, q_1)\leq 0, \eta_2(x_2, q_2)\leq0\}
\end{equation}
and 
the jump set is given by

\begin{eqnarray}\label{eqn:D}
D=\{z \in C: 
\eta_1(x_1, q_1) =  0\} \cup\{z \in C: 
\eta_2(x_2, q_2)  = 0\} 
\end{eqnarray}

To define the jump map, first note that at jumps, the continuous states $x_1$ and $x_2$ do not change. Then, we conveniently define 
$$g_1(z):=\matt{
x_1\\x_2\\1-q_1\\q_2
},\quad g_2(z):=\matt{
x_1\\x_2\\q_1\\1-q_2
},$$
so that the jump map $G$ is given by 
\begin{equation}\label{eqn:G}
\
G(z):=
\begin{cases}
g_1(z)&   \mbox{ if } 
\eta_1(x_1,q_1) =  0, \eta_2(x_2,q_2) < 0\\
g_2(z)&   \mbox{ if } 
\eta_1(x_1,q_1) <  0,
\eta_2(x_2,q_2) =  0 \\
\{g_1(z), g_{2}(z)\}& \mbox{ if } \eta_1(x_1, q_1) =  0,\eta_2(x_2, q_2) =  0.
\end{cases}
\end{equation}

The above definitions determine a hybrid system for \eqref{eqn:PWsystem},
which is given by
\begin{equation}\label{eqn:H2}
{\cal H}: z\in \Z \left\{
\begin{array}{ll}
\dot{z}= F(z) = \matt{
k_1(1-q_2)-\gamma_1x_1\\
k_2q_1-\gamma_2x_2\\
0\\
0}
& z\in C\\
z^+\in G(z)&  z\in D,
\end{array}\right.
\end{equation}
where $C$ is in \eqref{eqn:C}, $G$ is in \eqref{eqn:G}, and $D$ is in \eqref{eqn:D}.
Its parameters are given by the positive constants
$k_1$, $k_2$, $\gamma_1$, $\gamma_2$, 
$\theta_1$, $\theta_2$, 
$h_1$, $h_2$, 
which satisfy
$\theta_1 + h_1 < \theta_1^{\max}$, $\theta_2 +h_2  < \theta_2^{\max}$, 
$\theta_1 - h_1 > 0$, $\theta_2 - h_2 > 0$.
\NotForConf{
Figure~\ref{fig:Balloon} depicts a hybrid automaton representation of this system
when sequentially transitioning between $(q_1,q_2) = (0,0), (1,0), (1,1), (0,1)$.

\begin{figure}[h]
\begin{center}
\psfragfig*[width=0.7\columnwidth]{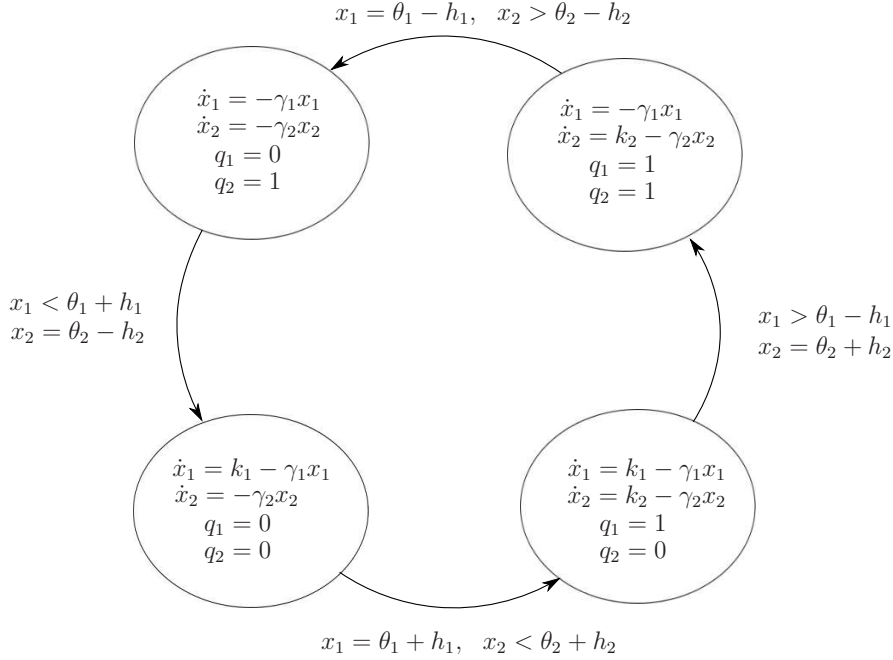}
{
\psfrag{x1dot=k1-r1x1}[][][0.8]{$\dot{x}_1=k_1-\gamma_1x_1$}
\psfrag{x2dot=k2-r2x2}[][][0.8]{$\dot{x}_2=k_2-\gamma_2x_2$}
\psfrag{x1dot=-r1x1}[][][0.8]{$\dot{x}_1=-\gamma_1x_1$}
\psfrag{x2dot=-r2x2}[][][0.8]{$\dot{x}_2=-\gamma_2x_2$}
\psfrag{q1=0}[][][0.8]{\qquad\!\!\!$q_1=0$}
\psfrag{q2=0}[][][0.8]{\qquad\!\!\!$q_2=0$}
\psfrag{q1=1}[][][0.8]{\qquad\!\!\!$q_1=1$}
\psfrag{q2=1}[][][0.8]{\qquad\!\!\!$q_2=1$}
\psfrag{x1leqt1+h1}[][][0.8]{\!\!\!\!$x_1<\theta_1+h_1$}
\psfrag{x1geqt1-h1}[][][0.8]{$x_1>\theta_1-h_1$}
\psfrag{x2leqt2+h2}[][][0.8]{$x_2=\theta_2+h_2$}
\psfrag{x2leqt2-h2}[][][0.8]{\!\!\!\!$x_2=\theta_2-h_2$}
\psfrag{x2geqt2-h2}[][][0.8]{$x_2=\theta_2-h_2$}
\psfrag{x2geqt2+h2}[][][0.8]{$x_2=\theta_2+h_2$}
\psfrag{x1leqt1-h1x2geqt2-h2}[][][0.8]{$x_1 = \theta_1 - h_1, \ \ x_2>\theta_2-h_2$}
\psfrag{x1geqt1+h1x2leqt2+h2}[][][0.8]{$x_1 = \theta_1 + h_1, \ \ x_2<\theta_2+h_2$}
}
\caption{\emph{A hybrid automaton representation of the two-gene genetic regulatory network for sequential transitions of $(q_1,q_2)$.}\label{fig:Balloon}}
\end{center}
\end{figure}
}

\NotForConf{
\begin{lemma}
\label{lemma:HBC}
The data $(C,F,D,G)$ satisfies the following conditions:
\begin{list}{}{}   
\item[(A1)] The sets $C$ and $D$ are closed.   
\item[(A2)] The map $z \mapsto F(z)$ is continuous on $C$.
\item[(A3)] The set-valued mapping $z\mapsto G(z)$ is outer semicontinuous\footnote{A set-valued mapping $G:S\rightrightarrows\reals^n$
with $S\subset \reals^n$ is outer semicontinuous relative to $S$  
if for any $z\in S$ and any sequence   
$\{z_i\}_{i=1}^\infty$ with $z_i\in S$, $\lim_{i\to\infty} z_i=z$,   
and any sequence $\{w_i\}_{i=1}^\infty$   
with $w_i\in G(z_i)$ and $\lim_{i\to\infty} w_i=w$ we have $w\in G(z)$.   
}
relative to $\reals^4$
and locally bounded, and, for all $z\in D$,  
$G(z)$ is nonempty.
\end{list}   
\end{lemma}
\begin{proof}
Properties (A1) and (A2) are obvious. 
Property (A3) holds since the graph of $G$, which is given by
$
\defset{(x,y)}{y \in G(z)},
$
is closed.
\end{proof}
}

\section{Dynamical Properties of the Two-Gene Hybrid System Model}
\label{sec:3}
\subsection{Existence of solutions}

\IfConf{
A solution $z$ to ${\cal H}$ is said to be nontrivial if $\dom z$ contains at least two points,
 complete if $\dom z$ is unbounded, Zeno if it is complete and the projection of $ \dom z$ onto $\mathbb{R}^n_{\geq0}$ is bounded, and maximal if there does not exist another solution $z'$ to ${\cal H}$ such that $\dom z'$ is a proper subset of $\dom z$, and $z' (t, j)=z (t, j)$ for all $(t, j)\in \dom z$.
}

\begin{proposition}
\label{eqn:Existence}
 From every point in $C\cup D,$ there exists a nontrivial solution for the hybrid system $\cal H$ in \eqref{eqn:H2}. Furthermore, every maximal solution is complete
and the projection of its hybrid time domain on $\realsgeq$ is unbounded, i.e., every solution is not Zeno.
 \end{proposition}
 
The proof of this result uses the conditions for the existence of solutions to $\cal H$ in \cite{84}
for general hybrid systems.
\NotForConf{ More precisely, consider the hybrid system $\cal{H}$ and let $z(0, 0)\in C\cup D.$ If $z(0, 0)\in D$ or\\

\begin{itemize}
\item[] {\bf (VC)} there exists a neighborhood $U$ of $z(0, 0)$ such that\footnote{$T_C(z)$ denotes the tangent cone of $C$ at $z$, i.e., 
it is the set      
of all $v$ for which there exists        
a sequence of real numbers $\alpha_i\searrow 0$ and a sequence        
$v_i\rightarrow v$ such that for every $i=1,2,...$, $x+\alpha_i v_i\in C$.} for every $z\in U\cap C$,
\IfConf{
$F(z)\cap T_C(z)\neq \emptyset,$
}
{
$$F(z)\cap T_C(z)\neq \emptyset,$$
}
\end{itemize}
then there exists a nontrivial solution to $\cal{H}$ from $z(0, 0)$. If {\bf (VC)} holds for every $z(0, 0)\in C\setminus D,$ then there exists a nontrivial solution to ${\cal H}$ from every initial point in $C\cup D,$ and every maximal solution $z$ satisfies exactly one of the following conditions:
\begin{enumerate}
\item $z$ is complete;
\item $\dom z$ is bounded and the last interval is of the form $[t_J,t_{J+1})$, where $J=\sup_{(t,j) \in \dom z} j$ has nonempty interior and $t\mapsto\phi(t, J)$ is a maximal solution to $\dot{z}=F(z),$ in fact $\lim_{t\to T}$$|z(t, J)|$=$\infty,$ where $T = \sup_{(t,j) \in \dom z} t$;
\item $z(T, J)\notin C\cup D,$ where $(T, J)= \sup \dom z$.
\end{enumerate}
Furthermore, if $G(D)\subset C\cup D$, then 3) above does not occur.
}

\NotForConf{
\ExtendedVersion{
The proof of Proposition~\ref{eqn:Existence} uses these conditions
and is given in \ref{app:ExistenceProof}.
}
}
\subsection{Characterization of equilibria}

We compute the set of isolated equilibrium points $z^*$ 
as well as (nonisolated, dense) sets of equilibria for the hybrid system ${\cal H}$ in \eqref{eqn:H2}.
For general hybrid systems, isolated equilibrium points are points 
that are an isolated equilibrium point of
$\dot{z} \in F(z), z \in C$ or of
$z^+ \in G(z), z \in D$.
%\begin{definition}[isolated equilibrium points]
%A point \IfConf{\\}{} $z^* =[ x^*_1, x^*_2, q^*_1, q^*_2]\in C\cup D$ 
%is an isolated equilibrium point of ${\cal H}$ 
%if, for some open neighborhood $U$ around $z^*$, 
%either one of the following properties hold:
%
%\begin{equation}\begin{array}{ll}
%\mbox{ if } \ z^*\in C \quad \Rightarrow \quad F(z^*)=0,\ F(z) \not = 0\ \quad\  \forall z \in U\setminus \{z^*\}\\
%\mbox{ if } \ z^*\in D \quad \Rightarrow \ \ z^* \in G(z^*),\ z^* \not \in G(z)\ \quad  \forall z \in U\setminus \{z^*\}
%\end{array}
%\end{equation}
%\end{definition}
On the other hand, an equilibrium set (not necessarily an isolated equilibrium point) for a hybrid system $\HS$ is defined 
as a set that is (strongly) forward invariant.

\begin{definition}[Equilibrium set]
A set $S \subset C\cup D$ 
is an equilibrium set of ${\cal H}$ 
if for every initial condition $z(0,0) \in S$, 
every solution $z$ to $\HS$ satisfies $z(t,j) \in S$ for all $(t,j) \in S$.
\end{definition}

The following results determine the equilibria of \eqref{eqn:H2} for a range of parameters of the system.

\begin{proposition}
\label{eqn:EP}
The equilibria of the hybrid system ${\cal H}$ in \eqref{eqn:H2} is given in Table~\ref{tab:EqPoints}
in terms of the positive constants $k_1$, $k_2$, $\gamma_1$, $\gamma_2$, 
$\theta_1$, $\theta_1^{\max}$, $\theta_2$, $\theta_2^{\max}$, 
$h_1$, and $h_2$ satisfying the conditions therein.
The set $S\subset C\cup D$ in case 5 is 
an equilibrium set and is given by 
\begin{equation}\label{eqn:S}
S=\bigcup^4_{i=1}S_i,
\end{equation}
where\footnote{$p_i(j)$ is the $j$-th component of $p_i$.}  
\begin{eqnarray*}
S_1 & := & \IfConf{\left\{x \in \reals^2: x = \matt{\frac{k_1}{\gamma_1} - \left(\frac{k_1}{\gamma_1} - p_0(1) \right) \exp(-\gamma_1 s) \\ p_0(2) \exp(-\gamma_2 s)}, \right. \\
& & \left.
s \in [0,t_1']\right\}\times \{(0,0)\}}
{\defset{x \in \reals^2}{x = \matt{\frac{k_1}{\gamma_1} - \left(\frac{k_1}{\gamma_1} - p_0(1) \right) \exp(-\gamma_1 s) \\ p_0(2) \exp(-\gamma_2 s)}, s \in [0,t_1']}\times \{(0,0)\}}
\\
S_2 & := & \IfConf{\left\{x \in \reals^2: x = \matt{\frac{k_1}{\gamma_1} - \left(\frac{k_1}{\gamma_1} - p_1(1) \right) \exp(-\gamma_1 s)\\
\frac{k_2}{\gamma_2} - \left(\frac{k_2}{\gamma_2} - p_1(2) \right) \exp(-\gamma_2 s)}, \right. \\
& & \left. s \in [0,t_2']\right\}\times \{(1,0)\}}
{\defset{x \in \reals^2}{x = \matt{\frac{k_1}{\gamma_1} - \left(\frac{k_1}{\gamma_1} - p_1(1) \right) \exp(-\gamma_1 s) \\ 
\frac{k_2}{\gamma_2} - \left(\frac{k_2}{\gamma_2} - p_1(2) \right) \exp(-\gamma_2 s)}, s \in [0,t_2']}\times \{(1,0)\}}
\\
S_3 & := & \IfConf{
\left\{x \in \reals^2: x = \matt{p_2(1) \exp(-\gamma_1 s) \\ \frac{k_2}{\gamma_2} - \left(\frac{k_2}{\gamma_2} - p_2(2) \right) \exp(-\gamma_2 s)}, \right. \\
& & \left. s \in [0,t_3']\right\}\times \{(1,1)\} 
}
{\defset{x \in \reals^2}{x = \matt{p_2(1) \exp(-\gamma_1 s) \\ \frac{k_2}{\gamma_2} - \left(\frac{k_2}{\gamma_2} - p_2(2) \right) \exp(-\gamma_2 s)}, s \in [0,t_3']}\times \{(1,1)\}}
\\
S_4 & := & \IfConf{\defset{x \in \reals^2}{x = \matt{p_3(1) \exp(-\gamma_1 s) \\ 
p_3(2) \exp(-\gamma_2 s)}, s \in [0,t_4']} \\ & & \times \{(0,1)\}}
{\defset{x \in \reals^2}{x = \matt{p_3(1) \exp(-\gamma_1 s) \\ 
p_3(2) \exp(-\gamma_2 s)}, s \in [0,t_4']}\times \{(0,1)\}
}
\end{eqnarray*}
and 
$p_0, p_1, p_2, p_3\in\mathbb{R}^2$ 
are the vertices of the set $S$\NotForConf{(see Figure~\ref{fig:LimitCycle1})},
where
\begin{eqnarray*}
& & t_1' \, = \, \ln\left[{\frac{\frac{k_1}{\gamma_1}-p_0(1)}{\frac{k_1}{\gamma_1}-(\theta_1+h_1)}}\right]^{\frac{1}{\gamma_1}}, 
t_2' \, = \, \ln\left[{\frac{\frac{k_2}{\gamma_2}-p_1(2)}{\frac{k_2}{\gamma_2}-(\theta_2+h_2)}}\right]^{\frac{1}{\gamma_2}},\\
&  & t_3' \, = \, \ln\left[{\frac{p_2(1)}{\theta_1-h_1}}\right]^{\frac{1}{\gamma_1}}, 
t_4' \, = \, \ln\left[{\frac{p_3(2)}{\theta_2-h_2}}\right]^{\frac{1}{\gamma_2}},
\end{eqnarray*}
and
\begin{eqnarray}
\label{eqn:p0}
p_0 &\!\!\!\! =\!\!\!\! & \left[\begin{array}{cc}(\theta_1-h_1)\left(\frac{\theta_2-h_2}{p_3(2)}\right)^{\frac{\gamma_1}{\gamma_2}}\\ \theta_2-h_2\end{array}\right], \ 
\end{eqnarray}
\begin{eqnarray}
\label{eqn:p1}
p_1 \! = \! \left[\begin{array}{cc}\theta_1+h_1\\
(\theta_2-h_2)\left(\frac{\frac{k_1}{\gamma_1}-(\theta_1+h_1)}{\frac{k_1}{\gamma_1}-p_0(1)}\right)^{\frac{\gamma_2}{\gamma_1}}\end{array}\right],
\end{eqnarray}
\begin{eqnarray}
\label{eqn:p2}
p_2 &\!\!\!\! = \!\!\!\! & \left[\begin{array}{cc}\frac{k_1}{\gamma_1}-\left(\frac{k_1}{\gamma_1}-(\theta_1+h_1)\right)\left(\frac{\frac{k_2}{\gamma_2}-(\theta_2+h_2)}{\frac{k_2}{\gamma_2}-p_1(2)}\right)^{\frac{\gamma_1}{\gamma_2}}\\\theta_2+h_2\end{array}\right],
\end{eqnarray}
\begin{eqnarray}
\label{eqn:p3}
p_3 &\!\!\!\! = \!\!\!\!& \left[\begin{array}{cc}\theta_1-h_1\\
\frac{k_2}{\gamma_2}-\left(\frac{k_2}{\gamma_2}-(\theta_2+h_2)\right)\left(\frac{\theta_1-h_1}{p_2(1)}\right)^{\frac{\gamma_2}{\gamma_1}}\end{array}\right].
\end{eqnarray}
Moreover, the period of the limit cycle is given by
\begin{equation}\label{eqn:Period}
T=t_1'+t_2'+t_3'+t_4'.
\end{equation}
\end{proposition}

\begin{table}[htbp]
\caption{Equilibria of the hybrid system \eqref{eqn:H2}.\label{tab:EqPoints}}
\begin{center}
\newcommand{\tabincell}[2]{\begin{tabular}{@{}#1@{}}#2\end{tabular}}
 \begin{tabular}{|c|c|c|}
 \hline
&Conditions on constants&\tabincell{c}{Equilibria}\\\hline
1&\tabincell{c}{$\theta_1+h_1<\frac{k_1}{\gamma_1}<\theta_1^{\max}$\\ $0<\frac{k_2}{\gamma_2}<\theta_2+h_2$} & $z^*_1:=\left[\begin{array}{cccc}\frac{k_1}{\gamma_1}& 
\frac{k_2}{\gamma_2}& 1& 0\end{array}\right]^\top$\\\hline
2&$0<\frac{k_1}{\gamma_1}<\theta_1-h_1$ & $z^*_2:=\left[\begin{array}{cccc}\frac{k_1}{\gamma_1}& 0& 0& 
0\end{array}\right]^\top$\\\hline
3&\tabincell{c}{$\theta_1-h_1<\frac{k_1}{\gamma_1}<\theta_1+h_1$\\$0<\frac{k_2}{\gamma_2}<\theta_2+h_2$}&$z^*_1$ or $z^*_2$\\\hline
4&\tabincell{c}{$\theta_1-h_1<\frac{k_1}{\gamma_1}<\theta_1+h_1$\\$\theta_2+h_2<\frac{k_2}{\gamma_2}<\theta_2^{max}$}&$z^*_2$\\\hline
5&\tabincell{c}{$\theta_1+h_1<\frac{k_1}{\gamma_1}<\theta_1^{\max}$\\$\theta_2+h_2<\frac{k_2}{\gamma_2}<\theta_2^{\max}$ }& \tabincell{c}{equilibrium set $S$  defined in \eqref{eqn:S}}\\
\hline
 \end{tabular}
\end{center}
\end{table}

\NotForConf{
The following result provides a more constructive characterization of $S$.

\begin{corollary}
\label{coro:LimitCycleLineCase}
Under the conditions of Proposition ~\ref{eqn:EP}, if furthermore, $\gamma_1=\gamma_2 = \gamma$, then
\begin{equation}
\label{eqn:p01}
p_0(1) =\frac{-d_6+d_8+d_7-d_5-d_4-d_3+d_2}{d_1},
\end{equation}
where
\begin{eqnarray*}
d_1 & = & 2h_1k_2^2\gamma+h_2k_1k_2\gamma+k_1k_2\gamma\theta_2-2h_1h_2k_2\gamma^2\\
& & -2h_1k_2\gamma^2\theta_2\\
\IfConf{d_2 & = & k_1k_2\gamma\theta_1\theta_2, \quad
d_3 \ = \ h_2k_1k_2\gamma\theta_1, \\
d_4 & = & h_1k_1k_2\gamma\theta_2,\quad
d_5 \ = \ h_1h_2k_1k_2\gamma, \\
d_7 & = & h_2k_1^2k_2, \quad
d_8 \ = \ h_1k_1k_2^2,}
{d_2 & = & k_1k_2\gamma\theta_1\theta_2, \quad
d_3 \ = \ h_2k_1k_2\gamma\theta_1, \quad
d_4 \ = \ h_1k_1k_2\gamma\theta_2,\\
d_5 & = & h_1h_2k_1k_2\gamma, \quad
d_7 \ = \ h_2k_1^2k_2, \quad
d_8 \ = \ h_1k_1k_2^2}
\end{eqnarray*}
\begin{eqnarray*}
d_6 & = & 
k_1^{\frac{1}{2}} k_2^{\frac{1}{2}} (h_1k_2+h_2k_1-2h_1h_2\gamma)^{\frac{1}{2}}\\
& & 
(2h_1^2h_2^2\gamma^3-2h_1^2h_2k_2\gamma^2+2h_1^2k_2\gamma^2\theta_2-2h_1^2\gamma^3\theta_2^2\\
& & -2h_1h_2^2k_1\gamma^2+d_8-2h_1k_1k_2\gamma\theta_2+2h_1k_1\gamma^2\theta_2^2\\
& & +2h_2^2k_1\gamma^2\theta_1-2h_2^2\gamma^3\theta_1^2+d_7-2h_2k_1k_2\gamma\theta_1\\
& & +2h_2k_2\gamma^2\theta_1^2+2k_1k_2\gamma\theta_1\theta_2-2k_1\gamma^2\theta_1\theta_2^2\\
& & -2k_2\gamma^2\theta_1^2\theta_2+2\gamma^3\theta_1^2\theta_2^2)^{\frac{1}{2}}
\end{eqnarray*}

\begin{figure}[h]
\begin{center}
\psfragfig*[width=0.7\columnwidth]{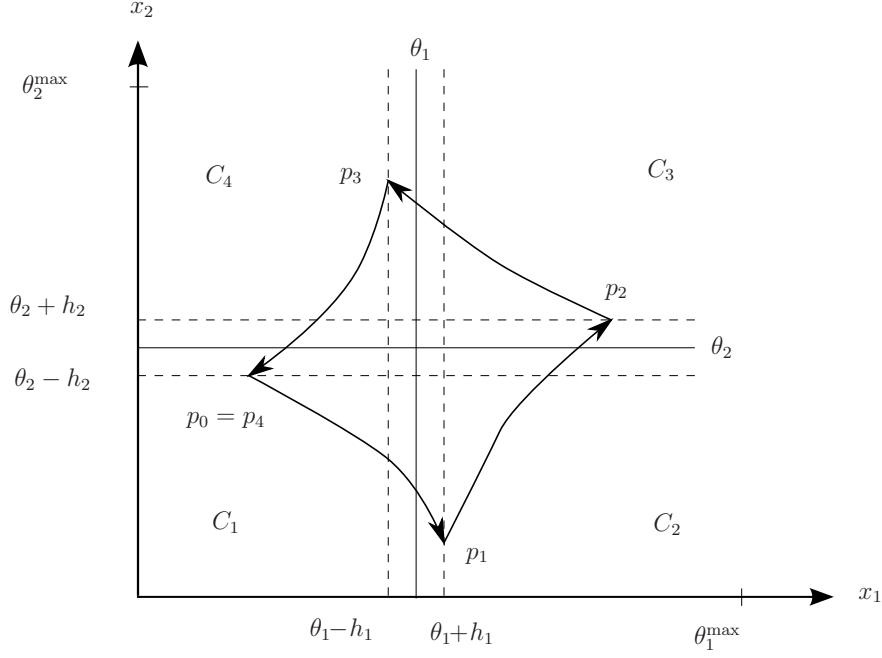}
{
\psfrag{C1}[][][0.8]{$C_1$}
\psfrag{x1}[][][0.8]{$x_1$}
\psfrag{x2}[][][0.8]{$x_2$}
\psfrag{C2}[][][0.8]{$C_2$}
\psfrag{C3}[][][0.8]{$C_3$}
\psfrag{C4}[][][0.8]{$C_4$}
\psfrag{P1}[][][0.8]{$p_1$}
\psfrag{P2}[][][0.8]{$p_2$}
\psfrag{P3}[][][0.8]{$\!\!\!p_3$}
\psfrag{P0}[][][0.8]{$p_0=p_4$}
\psfrag{t1}[][][0.8]{$\theta_1$}
\psfrag{t1m}[][][0.8]{$\theta_1\!\!-\!h_1$}
\psfrag{t1p}[][][0.8]{$\theta_1\!\!+\!h_1$}
\psfrag{t1max}[][][0.8]{$\theta_1^{\max}$}
\psfrag{t2}[][][0.8]{$\theta_2$}
\psfrag{t2m}[][][0.8]{$\theta_2-h_2$}
\psfrag{t2p}[][][0.8]{$\!\!\!\!\theta_2+h_2$}
\psfrag{t2max}[][][0.8]{$\theta_2^{\max}$}
\psfrag{tm}[][][0.8]{}
\psfrag{tp}[][][0.8]{}
}
\caption{\emph{Set $S$ and its vertices corresponding to case 5 of Table~\ref{tab:EqPoints}.}\label{fig:LimitCycle1}}
\end{center}
\end{figure}

Moreover, the sets $S_i$ are given by

\IfConf
{\begin{eqnarray*}
S_1 & =& \{x \in \reals^2: x_2=m_1x_1-m_1p_1(1)+p_1(2),\\
 & & p_0(1)\leq x_1 \leq \theta_1+h_1, p_1(2)\leq x_2 \leq \theta_2-h_2\}\\
& & \times \{(0,0)\},\\
S_2 & =& \{x \in \reals^2: x_2=m_2x_1-m_2p_1(1)+p_1(2), \\
& & \theta_1+h_1\leq x_1\leq p_2(1), p_1(2) \leq x_2\leq\theta_2+h_2\}\\
& & \times \{(1,0)\},\\
S_3 & = & \{x \in \reals^2: x_2=m_3x_1-m_3p_3(1)+p_3(2), \\
& & \theta_1-h_1 \leq x_1\leq p_2(1), \theta_2+h_2 \leq x_2\leq p_3(2)\}\\
& & \times \{(1,1)\},\\
S_4 & = & \{x \in \reals^2: x_2=m_4x_1-m_4p_3(1)+p_3(2), \\
& &  p_0(1) \leq x_1\leq\theta_1-h_1, \theta_2-h_2\leq x_2 \leq p_3(2)\}\\
& & \times \{(0,1)\},
\end{eqnarray*}}
{\begin{eqnarray*}
S_1 & =& \{x \in \reals^2: x_2=m_1x_1-m_1p_1(1)+p_1(2),\\
     & & \hspace{0.2in} p_0(1)\leq x_1<\theta_1+h_1, p_1(2)\leq x_2<\theta_2-h_2\}\times \{(0,0)\},\\
S_2 & =& \{x \in \reals^2: x_2=m_2x_1-m_2p_1(1)+p_1(2), \\
   & & \quad\theta_1+h_1\leq x_1< p_2(1), p_1(2)< x_2\leq\theta_2+h_2\}\times \{(1,0)\},\\
S_3 & = & \{x \in \reals^2: x_2=m_3x_1-m_3p_3(1)+p_3(2), \\
 & & \quad\theta_1-h_1< x_1\leq p_2(1), \theta_2+h_2< x_2\leq p_3(2)\}\times \{(1,1)\},\\
S_4 & = & \{x \in \reals^2: x_2=m_4x_1-m_4p_3(1)+p_3(2), \\
& & \quad p_0(1)< x_1\leq\theta_1-h_1, \theta_2-h_2\leq x_2< p_3(2)\}\times \{(0,1)\},
 \end{eqnarray*}}
where
\begin{equation}\label{eqn:m}
\begin{array}{l} \displaystyle
m_1 \ = \ \frac{p_0(2)-p_1(2)}{p_0(1)-p_1(1)},\ 
m_2 \ = \ \frac{p_2(2)-p_1(2)}{p_2(1)-p_1(1)},\\ \displaystyle
m_3 \ = \ \frac{p_2(2)-p_3(2)}{p_2(1)-p_3(1)},\
m_4 \ = \ \frac{p_0(2)-p_3(2)}{p_0(1)-p_3(1)}
\end{array}
\end{equation}
and the points 
$p_0, p_1, p_2, p_3\in\mathbb{R}^2$ 
are given in \eqref{eqn:p0}-\eqref{eqn:p3}.
\end{corollary}
}
\IfConf{}{\begin{proof}
When $\gamma_1 = \gamma_2 = \gamma$, 
the definitions in \eqref{eqn:p0}-\eqref{eqn:p3}
lead to
\begin{eqnarray*}
p_0(1) & = & 
(\theta_1-h_1)\left(\frac{\theta_2-h_2}{p_3(2)}\right), \\
p_1(2) & = &
(\theta_2-h_2)\left(\frac{\frac{k_1}{\gamma}-(\theta_1+h_1)}{\frac{k_1}{\gamma}-p_0(1)}
\right)\\
p_2(1) & = & 
\frac{k_1}{\gamma}-\left(\frac{k_1}{\gamma}-(\theta_1+h_1)\right)\left(\frac{\frac{k_2}{\gamma}-(\theta_2+h_2)}{\frac{k_2}{\gamma}-p_1(2)}\right)\\
p_3(2) & = & 
\frac{k_2}{\gamma}-\left(\frac{k_2}{\gamma}-(\theta_2+h_2)\right)\left(\frac{\theta_1-h_1}{p_2(1)}\right)
\end{eqnarray*}
Letting $\lambda = p_0(1)$, we obtain
\begin{eqnarray*}
\frac{(\theta_1-h_1)(\theta_2-h_2)}{\lambda} &=& \frac{k_2}{\gamma}-\left(\frac{k_2}{\gamma}-(\theta_2+h_2)\right)\left(\frac{\theta_1-h_1}{p_2(1)}\right)
\\
p_2(1) & = & 
\frac{k_1}{\gamma}-\left(\frac{k_1}{\gamma}-(\theta_1+h_1)\right)\left(\frac{\frac{k_2}{\gamma}-(\theta_2+h_2)}{\frac{k_2}{\gamma}-(\theta_2-h_2)\left(\frac{\frac{k_1}{\gamma}-(\theta_1+h_1)}{\frac{k_1}{\gamma}-\lambda}
\right)
}\right)
\end{eqnarray*}
Replacing the second equation into the first one, after elementary but tedious manipulations,
we obtain that $\lambda = p_0(1)$ as in \eqref{eqn:p01}.\footnote{When $\gamma_1 = \gamma_2 = \gamma$, 
the sets $S_i$ in Proposition~\ref{eqn:EP} reduce to straight lines.
In fact, define the new coordinates
\begin{equation}\label{eqn:Ecoordinates}
e: = x - \matt{\frac{k_1 (1-q_2)}{\gamma_1} \\ \frac{k_2 q_1}{\gamma_2}}.
\end{equation}
The continuous dynamics of $e$ are given by
\begin{equation}\label{eqn:EcoordinatesDot}
\dot{e} = \dot{x} = \matt{k_1(1-q_2)-\gamma_1x_1\\
k_2q_1-\gamma_2x_2}  =  -\gamma e,
\end{equation}
which implies that the trajectories on the plane are straight lines.}
\end{proof}}

\subsection{Stability analysis}

For convenience in the following analysis, we rewrite the flow set $C$ as $C=\bigcup^4_{i=1}C_i$ \NotForConf{(see Figure \ref{fig:LimitCycle1})}, where
\IfConf{\begin{eqnarray*}
C_1 &:=& \{z\in \Z: q_1=0, q_2=0, x_1\leq\theta_1+h_1,\\
& & \hspace{1.7in} x_2\leq\theta_2+h_2\},\\
C_2 &:=& \{z\in \Z: q_1=1, q_2=0, x_1\geq\theta_1-h_1,\\
& & \hspace{1.7in}  x_2\leq\theta_2+h_2\},\\
C_3 &:=& \{z\in \Z: q_1=1, q_2=1, x_1\geq\theta_1-h_1,\\
& & \hspace{1.7in}  x_2\geq\theta_2-h_2\},\\
C_4 &:=& \{z\in \Z: q_1=0, q_2=1, x_1\leq\theta_1+h_1,\\
& & \hspace{1.7in}  x_2\geq\theta_2-h_2\}.
\end{eqnarray*}}
{\begin{eqnarray*}
C_1 &:=& \{z\in \Z: q_1=0, q_2=0, x_1\leq\theta_1+h_1, x_2\leq\theta_2+h_2\},\\
C_2 &:=& \{z\in \Z: q_1=1, q_2=0, x_1\geq\theta_1-h_1, x_2\leq\theta_2+h_2\},\\
C_3 &:=& \{z\in \Z: q_1=1, q_2=1, x_1\geq\theta_1-h_1, x_2\geq\theta_2-h_2\},\\
C_4 &:=& \{z\in \Z: q_1=0, q_2=1, x_1\leq\theta_1+h_1, x_2\geq\theta_2-h_2\}.
\end{eqnarray*}}

\subsubsection{Asymptotic stability of isolated equilibrium points}

The following propositions determine the stability properties of 
the isolated equilibrium points in Table~\ref{tab:EqPoints}.
\NotForConf{
\ExtendedVersion{Their proofs are in Appendix~\ref{app:Stable1Proof} and Appendix~\ref{app:Stable2Proof}.
}
}
\begin{proposition}
\label{eqn:Stable1}For cases 1, 2, and 4 in Table \ref{tab:EqPoints}, 
the corresponding equilibrium points to $\cal H$ in \eqref{eqn:H2} 
are globally asymptotically stable.
\end{proposition}

\begin{proposition} 
\label{eqn:Stable2}
For case 3 in Table \ref{tab:EqPoints}, 
if $z(0, 0)\in C_2$, then we have that $\lim_{t+j \to \infty} z(t, j) = z_1^*$; if $z(0, 0)\in C_1$ 
or 
$z(0, 0)\in C_4,$ then $\lim_{t+j \to \infty} z(t, j) = z^*_2$. 
If $z(0, 0)\in C_3$, then $\lim_{t+j \to \infty} z(t, j) = z_1^*$ or $z_2^*.$ 
Furthermore, $z^*_1$ and $z_2^*$ are stable.
\end{proposition}

\subsubsection{Stability properties of the limit cycle}

\IfConf{
Now, we determine conditions on the parameters under which the limit cycle $S$ defined in  \eqref{eqn:S} is asymptotically stable.
As shown in Figure~\ref{fig:NaturalAndDecreasingDistance},
the natural metric (shown in dashed line) defined by the distance 
between the trajectories $z$ of $\HS$ and the set $S$ is not necessarily decreasing,
even though Figure~\ref{fig:PlanarTrajectoryWithVirtual} shows that the trajectory converges to $S$.
In fact, as depicted in the figures, the trajectory $x$ approaches $S$ for some time
and then gets far away from it (around the corners), until a jump to a new value of $q$ occurs.
}
{
Now, we determine conditions on the parameters under which the limit cycle $S$ defined in  \eqref{eqn:S} is asymptotically stable.
As shown in Figure~\ref{fig:NaturalDistance},
the natural metric defined by the distance 
between the trajectories $z$ of $\HS$ and the set $S$ is not necessarily decreasing,
even though Figure~\ref{fig:PlanarTrajectory} shows that the trajectory converges to $S$.
In fact, as depicted in the figures, the trajectory $x$ approaches $S$ for some time
and then gets far away from it (around the corners), until a jump to a new value of $q$ occurs.
}

\NotForConf{
\begin{figure}[H]
\begin{center}
\subfigure[Trajectory $x$ on the plane converging to $S$. \label{fig:PlanarTrajectory}]
{\psfragfig*[width=0.45\textwidth]{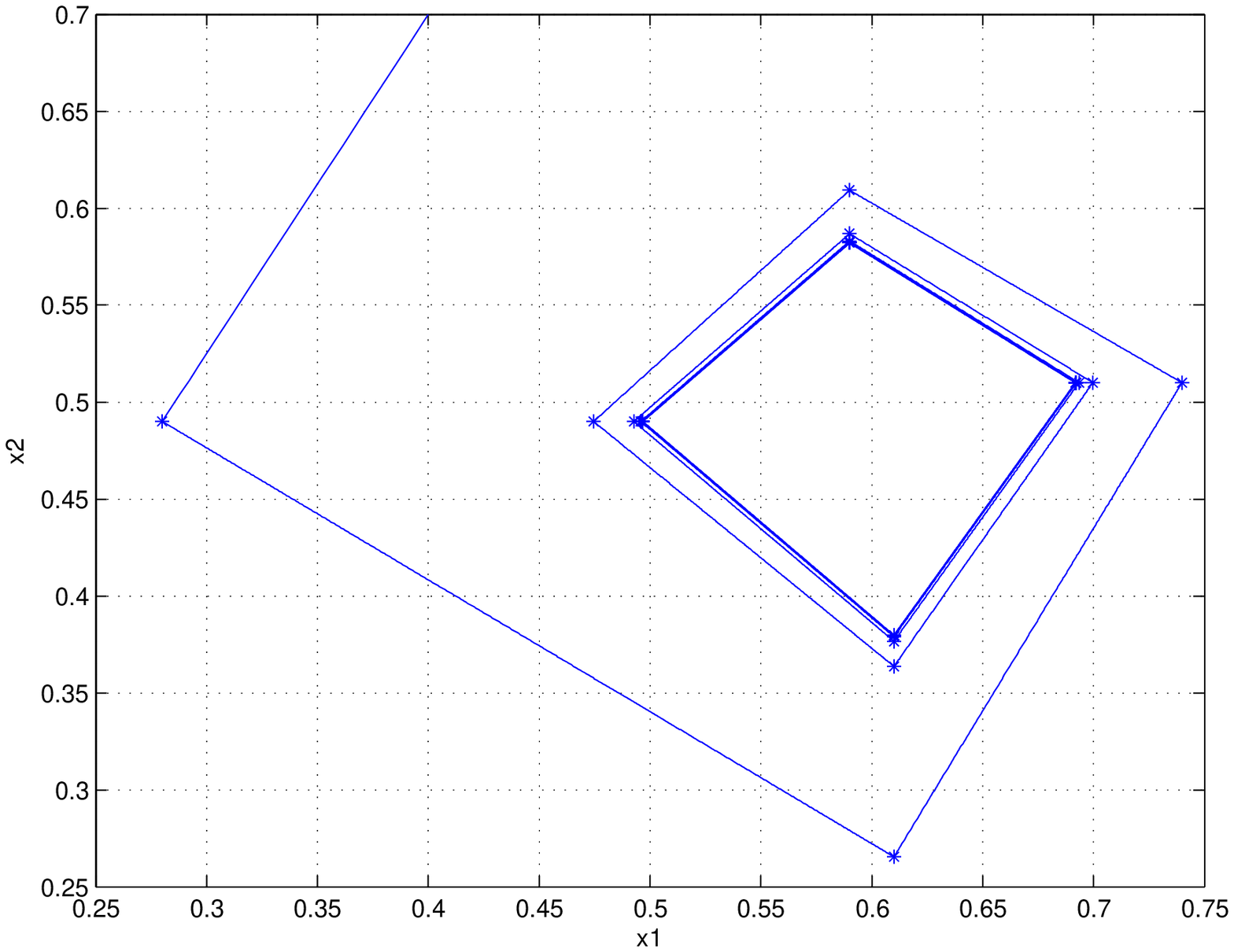}
{
\psfrag{x1}[][][0.65]{\hspace{-0.2in}$x_1$}
\psfrag{x2}[][][0.65][-90]{$x_2$}
\psfrag{t [sec]}[][][0.65]{$t [sec]$}
\psfrag{d(A,x)}[][][0.65][-90]{$|z|_{S}$}
}
}
\subfigure[Distance between trajectory $x$ and the set $S$, denoted $|z|_S$.\label{fig:NaturalDistance}]
{
\psfragfig*[width=0.45\textwidth]{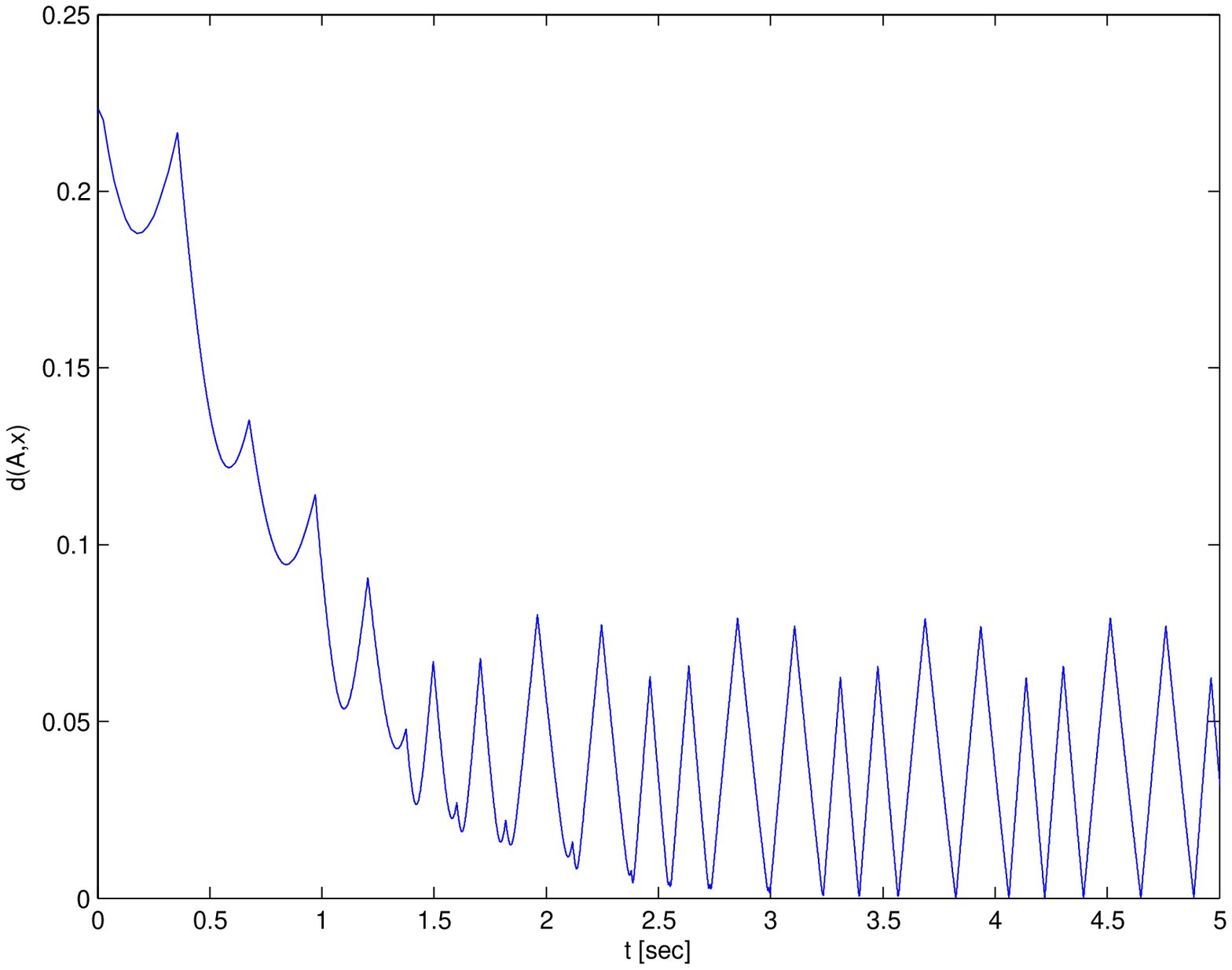}
{
\psfrag{x1}[][][0.65]{\hspace{-0.2in}$x_1$}
\psfrag{x2}[][][0.65][-90]{$x_2$}
\psfrag{t [sec]}[][][0.65]{$t [sec]$}
\psfrag{d(A,x)}[][][0.65][-90]{$|z|_{S}$}
}
}
\end{center}
\caption{\label{fig:Dis}\emph{Trajectory $z$ of $\HS$ on the plane and distance between it and the set $S$ with $\theta_1=0.6$, $\theta_2=0.5$, $\gamma_1=1$, $\gamma_2=1$, $k_1=1$, $k_2=1$, $h_1=0.01$,  $h_1=0.01$, $x_1(0,0)=0.4$, $x_2(0,0)=0.4$, $q_1(0,0)=0$, and $q_2(0,0)=0$.}}
\end{figure}
%}
}

To overcome this issue, we augment the hybrid system $\HS$
with a state $\zeta \in \reals^2$ and with continuous dynamics
governed by a flow map given by a copy of the one for $x$, that is,
$$
\dot{\zeta} = \matt{
k_1 (1-q_2)-\gamma_1 \zeta_1\\
k_2 q_1-\gamma_2 \zeta_2}.
$$
The discrete dynamics of $\zeta$ are chosen so that 
jumps occur when jumps of $\HS$ occur and, at such jumps, 
$\zeta$ is updated via the difference inclusion
$$
\zeta^+ \in \widetilde{G}(x,q,\zeta).
$$
\IfConf{
To define the jump map $\widetilde{G}$,
we consider the case $\gamma_1 = \gamma_2$ and
}
{
To define the jump map $\widetilde{G}$,
consider the case $\gamma_1 = \gamma_2$ and, using
Corollary~\ref{coro:LimitCycleLineCase}, 
}
we extend to $\mathbb{R}^{2}$ the set of points $S_i$, $i \in \{1,2,3,4\}$, that is, we define the (unbounded) set
\begin{equation}\label{eqn:widetildeS}
\widetilde{S}=\bigcup^4_{i=1}\widetilde{S}_i,
\end{equation}
where\\
$
\widetilde{S}_1  = \defset{x \in \mathbb{R}^{2}}{ x_2=m_1x_1-m_1p_1(1)+p_1(2)}\times \{(0,0)\}$,\\
$\qquad\widetilde{S}_2  = \defset{x \in \mathbb{R}^{2}}{ x_2=m_2x_1-m_2p_1(1)+p_1(2)}\times \{(1,0)\}$,\\
$\qquad\widetilde{S}_3  = \defset{x \in \mathbb{R}^{2}}{ x_2=m_3x_1-m_3p_3(1)+p_3(2)}\times \{(1,1)\}$,\\
$\qquad\widetilde{S}_4  = \defset{x \in \mathbb{R}^{2}}{ x_2=m_4x_1-m_4p_3(1)+p_3(2)}\times \{(0,1)\}.$
\IfConf{
The constants $m_i$ are defined as \begin{equation}\label{eqn:m}
\begin{array}{l} \displaystyle
m_1 \ = \ \frac{p_0(2)-p_1(2)}{p_0(1)-p_1(1)},\ 
m_2 \ = \ \frac{p_2(2)-p_1(2)}{p_2(1)-p_1(1)},\\ \displaystyle
m_3 \ = \ \frac{p_2(2)-p_3(2)}{p_2(1)-p_3(1)},\
m_4 \ = \ \frac{p_0(2)-p_3(2)}{p_0(1)-p_3(1)}
\end{array}
\end{equation}
}

During flows, the set $\widetilde{S}$ is forward invariant for the state component $\zeta$ (both during flows and jumps)
along the dynamics of $q$ governed by $\HS$.
This is the reason
we restrict $\zeta$ to belong to $\widetilde{S}$ for the current value of $q$.
Then, due to the stability properties of the error system with state $\zeta - x$, the distance between $x$ and $\zeta$ strictly decreases during flows.
With this useful property of the trajectories while flowing, 
at jumps due to $\HS$, which occur when $(x(t,j),q(t,j)) \in D$ and map $q(t,j)$ to $q(t,j+1)$ 
(following the definition of $G$ in \eqref{eqn:G}),
the jump map $\widetilde{G}$ is constructed to map the state $\zeta$ to satisfy $(\zeta(t,j+1),q(t,j+1)) \in \widetilde{S}$
such that, 
if $(\zeta(t,j),q(t,j)) \in \widetilde{S}_{q(t,j)}$ before the jump, 
then $(\zeta(t,j+1),q(t,j+1)) \in \widetilde{S}_{q(t,j+1)}$ and with the property that 
$$
\dist(x(t,j+1),\zeta(t,j+1)) \leq \dist(x(t,j),\zeta(t,j))
$$
where $\dist$ is
the Euclidean distance between two points in $\reals^2$.
In this way, 
the new value of $\zeta$ at jumps can be determined 
for each $x \in \reals^2$ from the set
\IfConf
{\\ $
\widetilde{G}(x,q,\zeta) :=
$
\\ $
\left\{\zeta': \dist(x,\zeta') \leq \dist(x,\zeta), (\zeta',q')\in \widetilde{S}_{q'}, (x,q') \in G(x,q)\right\}
$
}{$$
\widetilde{g}(x,q,\zeta):= \defset{\zeta'}{\dist(x,\zeta') \leq \dist(x,\zeta), (\zeta',q')\in \widetilde{S}_{q'}, (x,q') \in G(x,q)}
$$}
(when it is not empty).  
Since the distance between $x$ and $\zeta$ decreases during
flows, 
asymptotic stability of $\widetilde{S}$ can be established
when $\widetilde{G}(x,q,\zeta)$ is nonempty since this guarantees
that the distance between $x$ and $\zeta$ is nonincreasing.  
The following result
imposes conditions on the parameters 
guaranteeing that $\widetilde{G}$ is nonempty 
and, furthermore, extends the attractivity property to the set $S$.

\begin{theorem}
\label{thm:AsymptoticStabilityLineCase}
For positive constants $k_1$, $k_2$, $\gamma_1$, $\gamma_2$,
$\theta_1$, $\theta_1^{\max}$, $\theta_2$, $\theta_2^{\max}$, 
$h_1$, and $h_2$ such that
\IfConf
{\begin{eqnarray}\label{eqn:StabilityConditionLineCase}
\gamma_1&=& \gamma_2\ \  =\ \ \gamma,\\
|m_1| &\leq& \min\{|m_2|,|m_4|\}, \\
|m_3| &\leq& \min\{|m_2|,|m_4|\},
\end{eqnarray}}{\begin{eqnarray}\label{eqn:StabilityConditionLineCase}
\gamma_1=\gamma_2 = \gamma, \qquad 
|m_1| \leq \min\{|m_2|,|m_4|\}, \qquad 
|m_3| \leq \min\{|m_2|,|m_4|\},
\end{eqnarray}}
where, for each $i \in \{1,2,3,4\}$, 
$m_i$ are given in \eqref{eqn:m},
the following holds:
\begin{enumerate}
\item
The set $\widetilde{S}$ is globally asymptotically stable for $\HS$.  In particular, 
each maximal solution to $\HS$ satisfies
\IfConf{\begin{eqnarray}
\begin{array}{ll}
d((x(t,j),q(t,j)),\widetilde{S}) & \leq \\ 
& \hspace{-0.7in} \exp(-\gamma t) d((x(0,0),q(0,0)),\widetilde{S}) 
\end{array}
\end{eqnarray}}
{\begin{eqnarray}
d((x(t,j),q(t,j)),\widetilde{S}) \leq \exp(-\gamma t) d((x(0,0),q(0,0)),\widetilde{S}) 
\end{eqnarray}}
for all
$(t,j) \in \dom (x,q)$, where
$d((x,q),\widetilde{S}) = \min_{(\zeta,q) \in \widetilde{S}} |x - \zeta|.$
\item 
The set $S$ in case 5 of Table~\ref{tab:EqPoints} is globally attractive for $\HS$, i.e., 
every solution to $\HS$ converges to $S$.
\end{enumerate}
\end{theorem}
\IfConf{\begin{figure}[H]
\begin{center}
\psfrag{x1 (blue), z1 (red)}[][][0.65]{\hspace{-0.3in}$x_1, \zeta_1$}
\psfrag{x2 (blue), z2 (red)}[][][0.65][-90]{\hspace{-0.3in}$x_2, \zeta_2$}
\psfrag{t [sec]}[][][0.65]{$t [sec]$}
\psfrag{d(x,Aext)}[][][0.65]{$\dist(x,\zeta)$ (solid), $|x|_S$ (dashed)}
\subfigure[Trajectories $x$ and $\zeta$ on the plane. \label{fig:PlanarTrajectoryWithVirtual}]
{\includegraphics[width=0.23\textwidth]{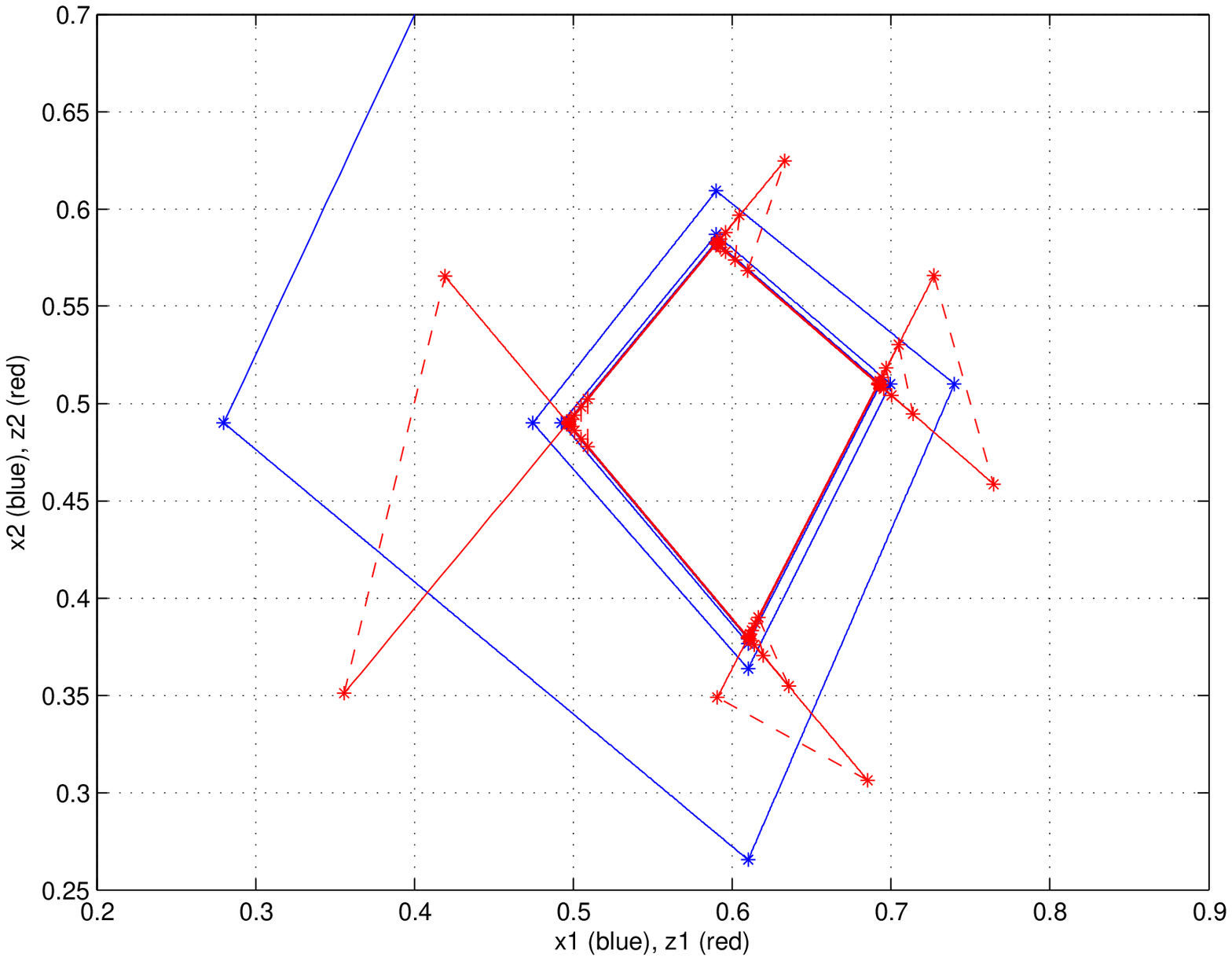}}
\subfigure[Distance between trajectory $x$ and $\zeta$ (solid),
and distance between $x$ and $S$ (dashed). \label{fig:NaturalAndDecreasingDistance}]
{\includegraphics[width=0.23\textwidth]{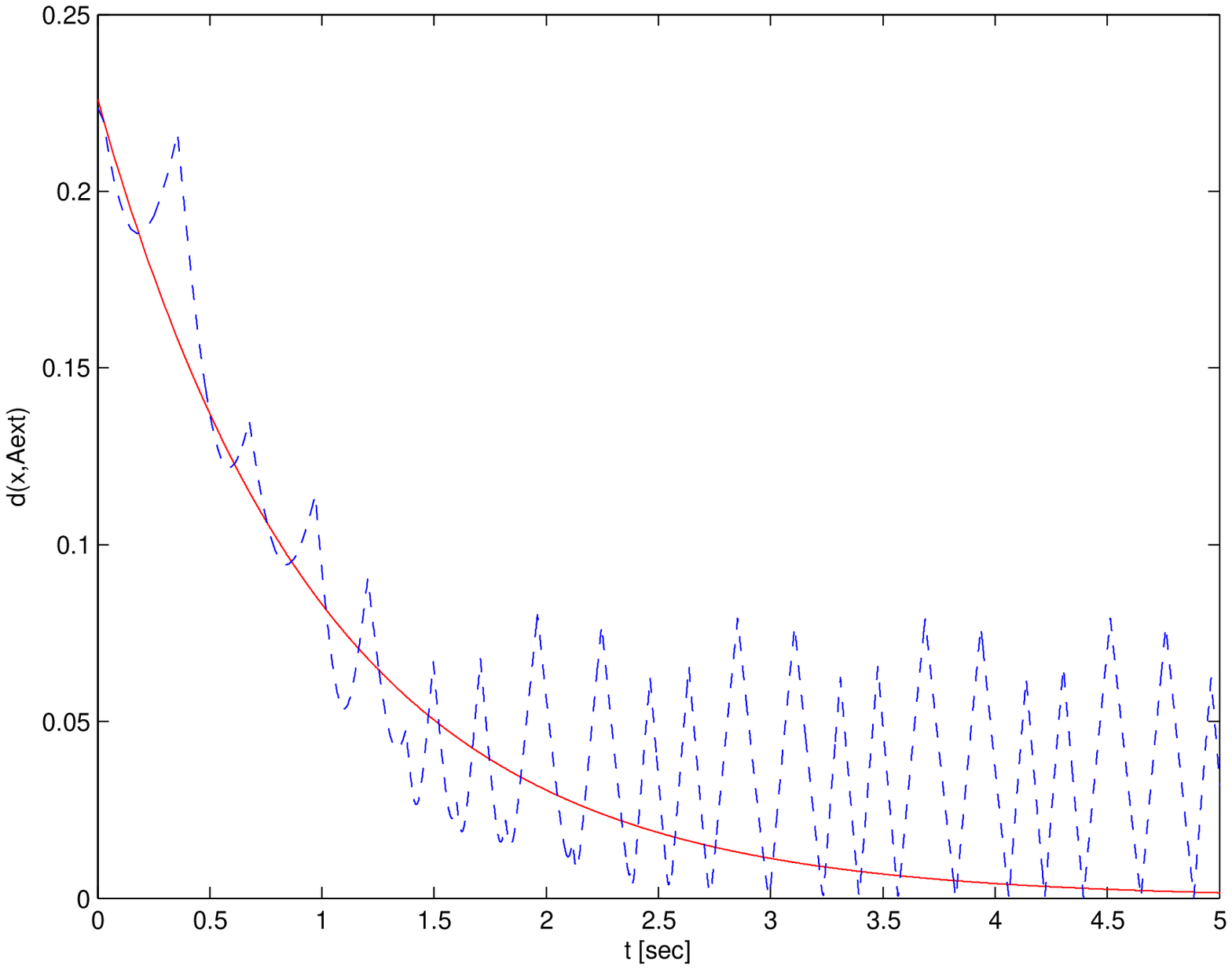}}\qquad
\end{center}
\vspace{-0.1in}
\caption{\label{fig:DecreasingDis}\emph{Trajectories $x$ and $\zeta$ on the plane, and distance between $x$ and $\zeta$ compared to distance between $x$ and the set $S$ (dashed).}}
\end{figure}}
{
%\begin{figure}[H]
%\begin{center}
%\psfrag{x1 (blue), z1 (red)}[][][0.65]{\hspace{-0.3in}$x_1, \zeta_1$}
%\psfrag{x2 (blue), z2 (red)}[][][0.65][-90]{\hspace{-0.3in}$x_2, \zeta_2$}
%\psfrag{t [sec]}[][][0.65]{$t [sec]$}
%\psfrag{d(x,Aext)}[][][0.65]{$\dist(x,\zeta)$ (solid), $|x|_S$ (dashed)}
%\subfigure[Trajectories $x$ and $\zeta$ on the plane. \label{fig:PlanarTrajectoryWithVirtual}]
%{\includegraphics[width=0.45\textwidth]{Figures/TwoGenes_PlanarTrajectoryWithVirtual}}
%\subfigure[Distance between trajectory $x$ and $\zeta$ (solid),
%and distance between $x$ and $S$ (dashed). \label{fig:NaturalAndDecreasingDistance}]
%{\includegraphics[width=0.45\textwidth]{Figures/TwoGenes_NaturalAndDecreasingDistance}}\qquad
%\end{center}
%\vspace{-0.1in}
%\caption{\label{fig:DecreasingDis}\emph{Trajectories $x$ and $\zeta$ on the plane, and distance between $x$ and $\zeta$ compared to distance between $x$ and the set $S$ (dashed) with the 
%same parameters and initial conditions as in Figure~\ref{fig:Dis}.}}
%\end{figure}
\begin{figure}[H]
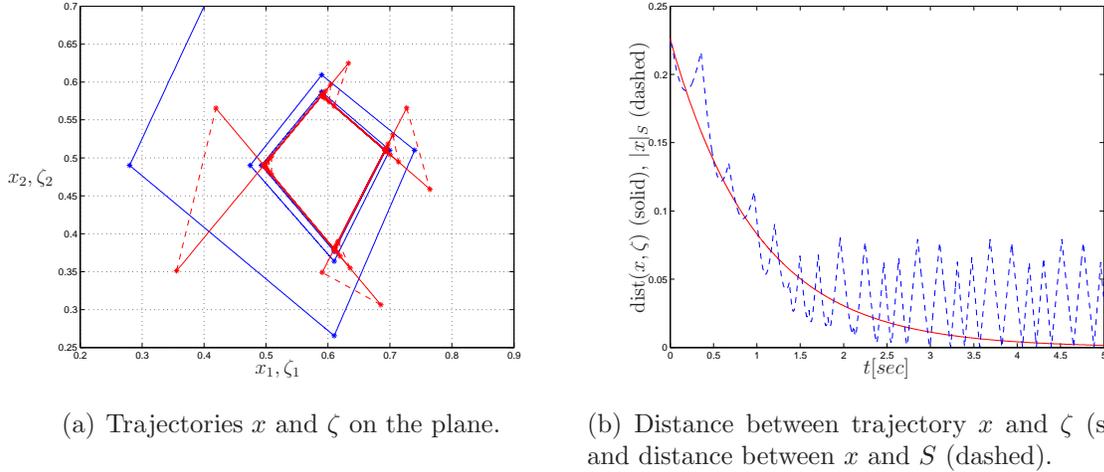

\begin{center}
\subfigure[Trajectories $x$ and $\zeta$ on the plane. \label{fig:PlanarTrajectoryWithVirtual}]
{\psfragfig*[width=0.45\textwidth]{Figures/TwoGenes_PlanarTrajectoryWithVirtual}
{
\psfrag{x1 (blue), z1 (red)}[][][0.65]{\hspace{-0.3in}$x_1, \zeta_1$}
\psfrag{x2 (blue), z2 (red)}[][][0.65][-90]{\hspace{-0.3in}$x_2, \zeta_2$}
\psfrag{t [sec]}[][][0.65]{$t [sec]$}
\psfrag{d(x,Aext)}[][][0.65]{$\dist(x,\zeta)$ (solid), $|x|_S$ (dashed)}
}
}
\subfigure[Distance between trajectory $x$ and $\zeta$ (solid),
and distance between $x$ and $S$ (dashed). \label{fig:NaturalAndDecreasingDistance}]
{\psfragfig*[width=0.45\textwidth]{Figures/TwoGenes_NaturalAndDecreasingDistance}
{
\psfrag{x1 (blue), z1 (red)}[][][0.65]{\hspace{-0.3in}$x_1, \zeta_1$}
\psfrag{x2 (blue), z2 (red)}[][][0.65][-90]{\hspace{-0.3in}$x_2, \zeta_2$}
\psfrag{t [sec]}[][][0.65]{$t [sec]$}
\psfrag{d(x,Aext)}[][][0.65]{$\dist(x,\zeta)$ (solid), $|x|_S$ (dashed)}
}
}\qquad
\end{center}
\vspace{-0.1in}
\caption{\label{fig:DecreasingDis}\emph{Trajectories $x$ and $\zeta$ on the plane, and distance between $x$ and $\zeta$ compared to distance between $x$ and the set $S$ (dashed) with the 
same parameters and initial conditions as in Figure~\ref{fig:Dis}.}}
\end{figure}
}

Figure~\ref{fig:DecreasingDis} shows trajectories $x$ and $\zeta$ 
as well as the distance between them obtained from the hybrid system
augmented with the state $\zeta$.  As Figure~\ref{fig:NaturalAndDecreasingDistance}
indicates, this distance (solid) decreases to zero while, as pointed out earlier,
the natural distance between $x$ and $S$ (dashed) does not. 
\NotForConf{
The extended version of the hybrid system $\HS$ in \eqref{eqn:H2}
can be written as
\begin{equation}\label{eqn:H2ext}
\widetilde{\cal H}: (x,q,\zeta)\in \Z\times\reals^2_{\geq 0} \left\{
\begin{array}{ll}
\matt{\dot{x}_1 
\\
\dot{x}_2
\\
\dot{q}_1 
\\
\dot{q}_2
\\
\dot{\zeta}_1 
\\
\dot{\zeta}_2
}= \matt{
k_1(1-q_2)-\gamma_1x_1\\
k_2q_1-\gamma_2x_2\\
0\\
0\\
k_1(1-q_2)-\gamma_1 \zeta_1\\
k_2q_1-\gamma_2 \zeta_2
} =: \widetilde{F}(x,q,\zeta) \\
\hspace{2in} (x,q)\in C, (\zeta,q) \in \widetilde{S},\\
\matt{z^+\\ \zeta^+ }\in \matt{G(z)\\ \widetilde{G}(x,q,\zeta)} =: \widetilde{G}(x,q,\zeta) \\
\hspace{2in}  (x,q)\in D, (\zeta,q) \in \widetilde{S}.
\end{array}\right.
\end{equation}
}

\IfConf{}{We are now ready to prove Theorem~\ref{thm:AsymptoticStabilityLineCase}.}
\IfConf{}{\begin{proof} {\bf (of Theorem~\ref{thm:AsymptoticStabilityLineCase})}
First, we show that $\widetilde{G}$ is nonempty for each $(x,q,\zeta)$ such that $(x,q)\in D$ and $(\zeta,q) \in \widetilde{S}$.
For each $(x,q) \in D$, the minimum possible value for $\dist(x,\zeta)$
with $(\zeta,q) \in \widetilde{S}$ is given by the minimum distance 
between $x$ and the projection on $\reals^2$ of $\widetilde{S}$ for the chosen $q$.
There are four possible cases for this distance (one per possible value of $q$)
and each distance can be computed as the minimum distance between
the point $x$ and the line defined by $\widetilde{S}$ for the chosen $q$.
For jumps from $q = (0,1)$ to $q = (0,0)$,
in which case $x_1 \in [0,p_0(1)]$, $x_2 = \theta_2 - h_2$, 
the minimum distance is
\begin{equation}\label{eqn:SamplePreDistance}
\frac{|m_4| |x_1- p_0(1)|}{\sqrt{m_4^2 + 1}}
\end{equation}
Similarly, the minimum distance from $x$ to the line defined by $\widetilde{S}$ for 
$q = (0,0)$, which is the distance between $(x,q)$ and $\widetilde{S}$ after the jump,
is given by
\begin{equation}\label{eqn:SamplePostDistance}
\frac{|m_1| |x_1- p_0(1)|}{\sqrt{m_1^2 + 1}}.
\end{equation}
Then, imposing that 
\eqref{eqn:SamplePostDistance} is no larger than
\eqref{eqn:SamplePreDistance} guarantees that, in the worst case,
$\dist(x,\zeta') \leq \dist(x,\zeta)$.
Then, we require
\begin{equation}
\frac{|m_1| |x_1- p_0(1)|}{\sqrt{m_1^2 + 1}} \leq
\frac{|m_4| |x_1- p_0(1)|}{\sqrt{m_4^2 + 1}}
\qquad
\Longleftrightarrow
\qquad
|m_1| \leq |m_4|.
\end{equation}
Proceeding in this way, for jumps from $q = (0,0)$ to $q = (1,0)$, 
from $q = (1,0)$ to $q = (1,1)$, and
from $q = (1,1)$ to $q = (0,1)$ 
we require
\begin{equation}
|m_1| \leq |m_2|, \qquad 
|m_3| \leq |m_2|, \qquad 
|m_3| \leq |m_4|,
\end{equation}
respectively.  Under these conditions,
which can be rewritten as in \eqref{eqn:StabilityConditionLineCase}, $\widetilde{G}$ is nonempty.

For each $(x,q,\zeta) \in \reals^2\times \{0,1\}^2 \times \reals^2$,
let
$$
V(x,q,\zeta) =  \dist(x,\zeta)^2
$$
and note that $V$ is positive definite with respect to the closed set
\begin{equation}
\A := \defset{(x,q,\zeta)}{ x = \zeta, (x,q) \in C \cup D, (\zeta,q) \in \widetilde{S} }.
\end{equation}
For each $(x,q)\in C$ and $(\zeta,q) \in \widetilde{S}$,
we obtain
\begin{eqnarray} \non
\langle \nabla V(x,q,\zeta), \widetilde{F}(x,q,\zeta)  \rangle 
&=& -2 \left( \gamma_1 (x_1 - \zeta_1)^2 + \gamma_2 (x_2 - \zeta_2)^2 \right)
\\ \label{eqn:VdecreaseFlows}
&=& -2 \gamma V(x,q,\zeta),
\end{eqnarray}
where we have used the condition $\gamma_1 = \gamma_2 = \gamma$.
For each $(x,q)\in D$ and $(\zeta,q) \in \widetilde{S}$,
we have
\begin{eqnarray}\non
\max_{\xi \in \widetilde{G}(x,q,\zeta)} V(\xi) - V(x,q,\zeta)
&=& 
\max_{(x,\xi_2) \in G(x,q), (\xi_3,\xi_2) \in \widetilde{G}(x,q,\zeta)}  \dist(x,\xi_3)^2 -  \dist(x,\zeta)^2
\\ \label{eqn:VdecreaseJumps}
& \leq &  0
\end{eqnarray}
since, by definition of $\widetilde{G}$, 
we have that any possible value of $\xi_3$ obtained from $\widetilde{G}$ is such that $\dist(x,\xi_3)^2 \leq \dist(x,\zeta)^2$.
Then, since every maximal solution to $\HS$ (and, hence, to $\widetilde{\HS}$) is complete
and has a hybrid time domain unbounded in the $t$ direction,
\cite[Proposition 3.29]{83} implies that $\A$ is globally asymptotically stable.\footnote{The 
same result can be obtained using the invariance principle for hybrid systems in \cite{SanfeliceGoebelTeel05}.}
In fact, 
combining
\eqref{eqn:VdecreaseFlows} and \eqref{eqn:VdecreaseJumps}, and simple integration, we get that
every solution $(x,q,\zeta)$ to $\widetilde{\HS}$ satisfies
\begin{eqnarray}\label{eqn:dKLBound}
\dist(x(t,j),\zeta(t,j)) \leq \exp(-\gamma t) \dist(x(0,0),\zeta(0,0)) 
\end{eqnarray}
for all
$(t,j) \in \dom (x,q,\zeta)$.

Now, we relate the asymptotic stability property above to $\widetilde{S}$. 
The bound
\eqref{eqn:dKLBound}
holds for any $\zeta(0,0)$ such that $(\zeta(0,0),q(0,0)) \in \widetilde{S}$, 
in particular, when $\zeta(0,0)$ is such that\footnote{Note that we could also pick $\zeta(0,0)$ such that the distance to $S$ matches.}
$\dist(x(0,0),\zeta(0,0)) = d((x(0,0),q(0,0)),\widetilde{S})$.
Moreover, note that since $(\zeta(t,j),q(t,j))\in \widetilde{S}$ for all $(t,j) \in \dom (x,q,\zeta)$, we have
\begin{eqnarray}
d((x(t,j),q(t,j)),\widetilde{S}) \leq \dist(x(t,j),\zeta(t,j))
\end{eqnarray}
for all
$(t,j) \in \dom (x,q,\zeta)$.
Then, from \eqref{eqn:dKLBound} and the above arguments, we obtain
\begin{eqnarray}
\label{eqn:dKLBoundS}
d((x(t,j),q(t,j)),\widetilde{S}) \leq \exp(-\gamma t) d((x(0,0),q(0,0)),\widetilde{S}) 
\end{eqnarray}
for all
$(t,j) \in \dom (x,q,\zeta)$.

To show that the components $(x,q)$ of the solutions to $\widetilde{\HS}$ converge to $S$, 
we proceed by contradiction and suppose 
that there exists a maximal solution to $\widetilde{\HS}$
with components $(x,q)$ with $\omega$-limit set $\Omega(x,q)$
such that $\Omega(x,q) \cap (\widetilde{S} \setminus S) \not= \emptyset$.
Let $z^\circ \in \Omega(x,q) \cap (\widetilde{S} \setminus S) \not= \emptyset$.
By the properties of the $\omega$-limit set of complete solutions to hybrid systems (see \cite[Definition 3.2 and Lemma 3.3]{SanfeliceGoebelTeel05}),
there exists at least one solution starting from $z^\circ$, which is impossible since points 
in $\widetilde{S} \setminus S$ are not in $C \cup D$ and $\widetilde{\HS}$
satisfies the hybrid basic conditions.  
Then, $\Omega(x,q)$ cannot contain points that are 
not in $S$, which implies that $\Omega(x,q) \subset S$.
Convergence 
of components $(x,q)$ of the solutions to $\widetilde{\HS}$ to $S$
follows by the very definition of $\omega$-limit set of a solution.

\end{proof}}

\IfConf{
Due to the regularity properties of the data of $\HS$,
the asymptotic stability guaranteed by Theorem~\ref{thm:AsymptoticStabilityLineCase}
is robust to small perturbations.%; see \cite{Shu.Sanfelice.13.TR}.
}
{}

\NotForConf{
\subsection{Robustness properties}

When the system $\HS$ in \eqref{eqn:H2} is restricted to a compact set of the initial conditions for the state component $x$,
the asymptotic stability of the set $\widetilde{S}$ guaranteed
in Theorem~\ref{thm:AsymptoticStabilityLineCase} 
is robust to small perturbations.
We define this set of initial conditions as the compact box 
in $\reals^2_{\geq 0}$ as
$$
K := [0,x_1^{\max}]\times [0,x_2^{\max}]
$$
with positive constants $x_1^{\max}$ and $x_2^{\max}$
such that $S \subset K \times \{0,1\}^2$.
We consider perturbations on the state and on the continuous dynamics of the system.
The signal $d_1:\realsgeq \to \delta_1 \ball \subset \reals^2$ defines 
the perturbation on the state and
the signal $d_2:\realsgeq \to \delta_2 \ball \subset \reals^2$ defines 
the perturbation on the flow of $x$, where $\delta_1, \delta_2 > 0$.
In this way, the perturbed hybrid system is given by
\begin{equation}\label{eqn:H2pert}
{\cal H}_{\delta}: z\in \Z \left\{
\begin{array}{ll}
\dot{z}= \matt{
k_1(1-q_2)-\gamma_1(x_1+d_{11}(t)) + d_2(t)\\
k_2q_1-\gamma_2(x_2+d_{12}(t)) + d_2(t)\\
0\\
0}
& (x+d_1(t),q)\in C \cap K\\
z^+\in G(z)&  (x+d_1(t),q)\in D \cap K,
\end{array}\right.
\end{equation}
where $C$ is defined in \eqref{eqn:C}, $G$ in \eqref{eqn:G}, and $D$ in \eqref{eqn:D}.
The perturbation $d_1$ captures uncertainty in the values of the protein concentrations $x$
while $d_2$ models the uncertainty in the dynamical model governing $x$.\footnote{Perturbations on each of the system parameters, in particular, the thresholds $\theta_i$ and hysteresis half widths $h_i$, can be treated similarly.}  In particular, the
latter perturbation allows for uncertainty in the parameters $k_1, k_2$.  
For instance, if $k_1$ is replaced by $k_1 + k_1^{\delta}$ 
with 
$k_1^{\delta} \in \reals$ then the continuous dynamics of $x_1$ 
along a solution $(x,q)$ to $\HS$ can be rewritten as
\IfConf
{\begin{eqnarray*}
\frac{d}{dt}x_1(t,j) &=& (k_1 + k_1^{\delta}) (1-q_2(t,j)) - \gamma_1(x_1+d_{11}(t))\\
& =& k_1 (1-q_2(t,j)) - \gamma_1(x_1(t,j)+d_{11}(t))\\
& & +k_1^{\delta} (1-q_2(t,j)),
\end{eqnarray*}}{\begin{eqnarray*}
\frac{d}{dt}x_1(t,j) &=& (k_1 + k_1^{\delta}) (1-q_2(t,j)) - \gamma_1(x_1+d_{11}(t))\\
& =& k_1 (1-q_2(t,j)) - \gamma_1(x_1(t,j)+d_{11}(t)) +
k_1^{\delta} (1-q_2(t,j)),
\end{eqnarray*}}
which leads to\footnote{For each $t$ such that $(t,j)\in \dom (x,q)$, 
the function $j:\realsgeq \to \nats$ is given by
$j(t) = j'$,  
where $j' = \max \defset{j}{(t,j) \in \dom (x,q)}$.}
$d_{21}(t) = k_1^{\delta} (1-q_2(t,j(t)))$.
Note that since $q_2$ takes values from $\{0,1\}$, then 
we have that $|d_{2}(t)| \leq \delta_2$
when $|k_1^{\delta}| \leq \frac{\sqrt{2}}{2}\delta_2$.  

Due to $\HS$ satisfying 
conditions (A1)-(A3) in Lemma~\ref{lemma:HBC},
the stability property guaranteed by Theorem~\ref{thm:AsymptoticStabilityLineCase}
is robust to small perturbations.  This 
property follows from the 
results on robustness of stability for hybrid systems 
in \cite{83}.

\begin{theorem}
\label{thm:RobustStability}
For each 
positive constants $x_1^{\max}$ and $x_2^{\max}$
defining 
$
K := [0,x_1^{\max}]\times [0,x_2^{\max}]
$
such that $S \subset K \times \{0,1\}^2$
and system constants satisfying 
case 5 of Table~\ref{tab:EqPoints},
there exists\footnote{A function $\beta$ is of class $\classKL$
if it is continuous, $r \mapsto \beta(r,s)$ is zero at         
zero and nondecreasing, and $s \mapsto \beta(r,s)$ is        
nonincreasing and converges to zero as $s$ goes to $\infty$.} $\beta \in \classKL$
such that,
for each $\eps >0$ 
there exists $\delta>0$ such that for each measurable
functions
${d}_1:\realsgeq \to \delta_1\ball$,
${d}_2:\realsgeq \to \delta_2\ball$
with $\delta_1, \delta_2 \in (0,\delta]$,
every solution
$(x,q)$ to ${\HS}_{\delta}$ 
with $(x(0,0),q(0,0)) \in K$
satisfies
\IfConf
{\begin{eqnarray*}
|(x(t,j),q(t,j))|_{\widetilde{S} \cap K} &\leq& \beta(|(x(0,0),q(0,0))|_{\widetilde{S} \cap K}, t+j)\\
& & +\eps\forall (t,j) \in \dom (x,q).
\end{eqnarray*}}{\begin{eqnarray*}
|(x(t,j),q(t,j))|_{\widetilde{S} \cap K}\leq \beta(|(x(0,0),q(0,0))|_{\widetilde{S} \cap K}, t+j)+\eps\  \ \ \forall (t,j) \in \dom (x,q).
\end{eqnarray*}}
\end{theorem}
}

\section{Numerical results}
\label{sec:4}

\NotForConf{
In this section, we simulate the hybrid system $\cal H$ in \eqref{eqn:H2} within Matlab/Simulink using the
HyEQ Toolbox \cite{80}.
}

\NotForConf{
\subsection{Isolated equilibrium points in Table \ref{tab:EqPoints}}

We perform simulations with parameters satisfying the conditions in Table~\ref{tab:EqPoints} 
for which there are isolated equilibrium points.
\IfConf{
Due to space constraints, we illustrate cases 1 and 3 of Table~\ref{tab:EqPoints} and refer
the reader to \cite{Shu.Sanfelice.13.TR} for the other cases.
}

\subsubsection{Case 1 of Table~\ref{tab:EqPoints}}

Figure \ref{fig:P1} illustrates that, when $\theta_1+h_1 < \frac{k_1}{\gamma_1}< \theta_1^{\max}$, $0< \frac{k_2}{\gamma_2}<\theta_2+h_2$, the solution converges to $z^*_1=[\frac{k_1}{\gamma_1} , \frac{k_2}{\gamma_2}, 1, 0]^\top$. Initially, the concentration of protein {\emph A} $(x_1)$ is low, which inhibits the expression of gene ${\emph b}$, hence the concentration of protein {\emph B} $(x_2)$  decreases and activates the expression of gene {\emph a}. However, after finite time, while the concentration of protein {\emph A}  is above the level $\theta_1+h_1$, which can permit the expression of gene {\emph b}, the concentration of protein {\emph B} increases. Finally, the concentrations of protein {\emph A} and {\emph B}  come to the equilibrium  $(\frac{k_1}{\gamma_1}, \frac{k_2}{\gamma_2}).$ This confirms the result in Proposition~\ref{eqn:Stable1}.

\begin{figure}[h!]
\begin{center}
\subfigure[$x$ components.]
{\psfragfig*[width=0.48\textwidth]{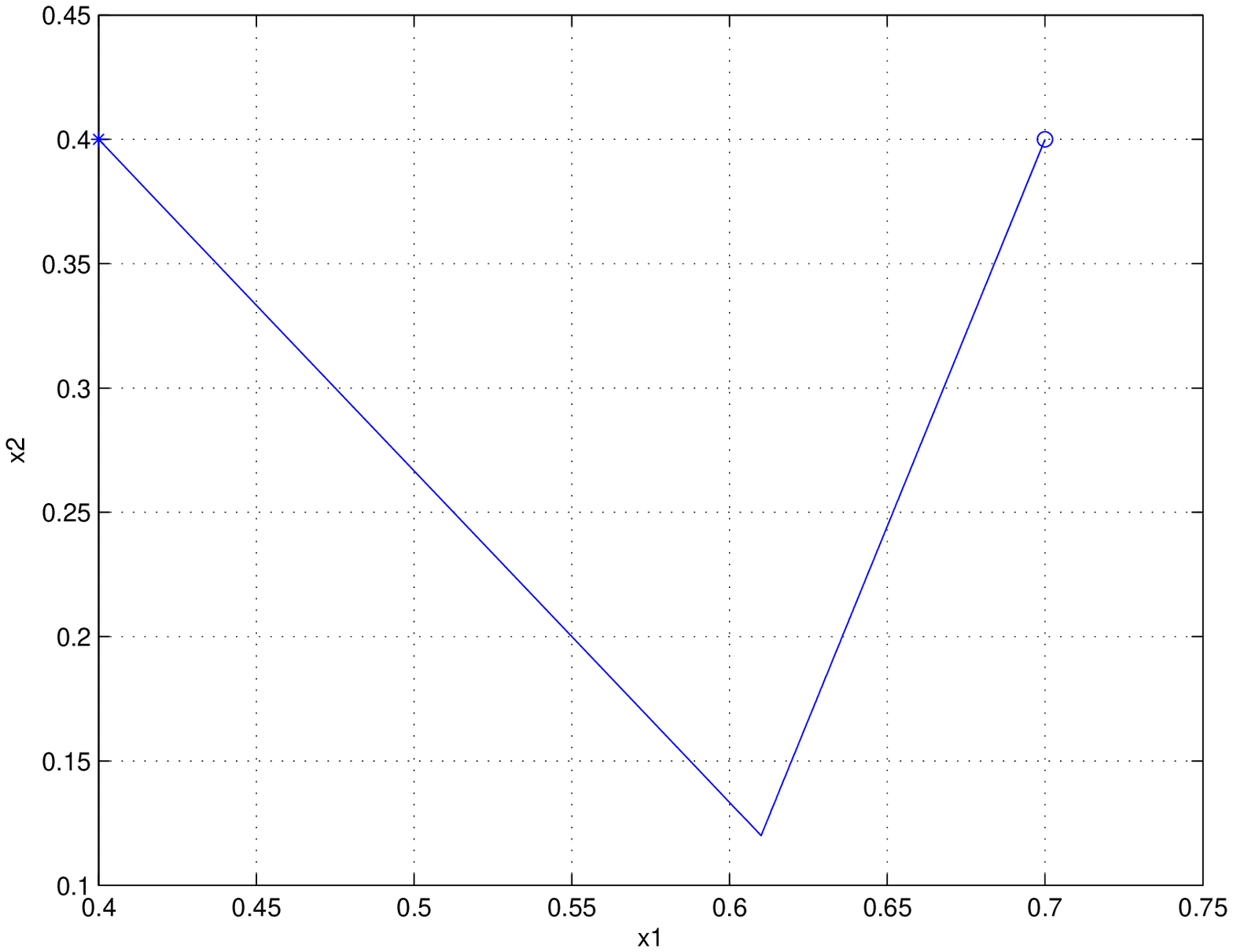}
{
\psfrag{x1}[][][0.9]{$x_1$}
\psfrag{x2}[][][0.9][-90]{$x_2$}
\psfrag{q1}[][][0.9][-90]{$q_1$}
\psfrag{q2}[][][0.9][-90]{$q_2$}
\psfrag{t}[][][0.7]{\quad $t [sec]$}
}
}
\subfigure[$q$ components.]
{\psfragfig*[width=0.48\textwidth]{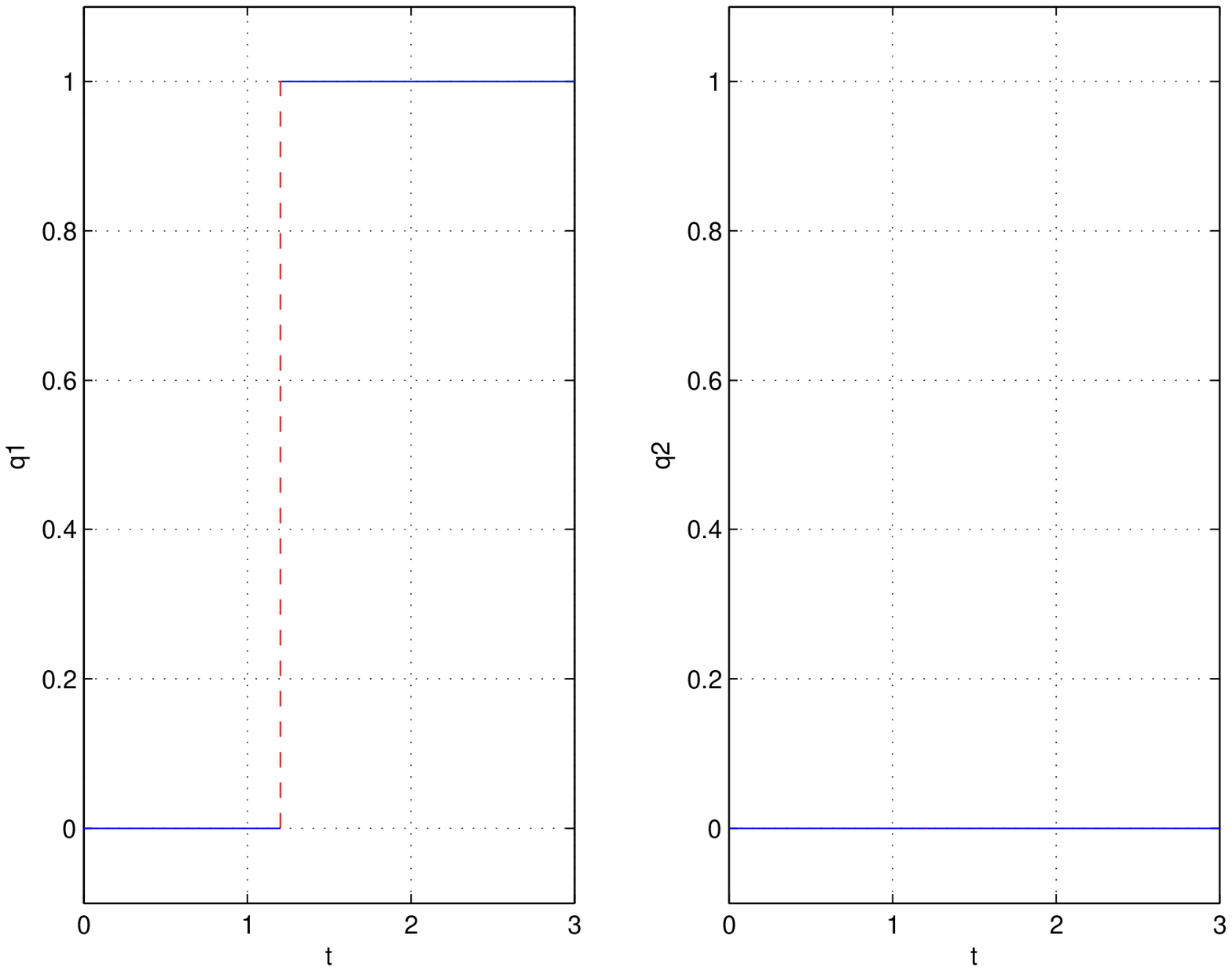}
{
\psfrag{x1}[][][0.9]{$x_1$}
\psfrag{x2}[][][0.9][-90]{$x_2$}
\psfrag{q1}[][][0.9][-90]{$q_1$}
\psfrag{q2}[][][0.9][-90]{$q_2$}
\psfrag{t}[][][0.7]{\quad $t [sec]$}
}
}
\end{center}
\caption{\label{fig:P1}\emph{The $x$ and $q$ components of a solution to $\cal H$ in \eqref{eqn:H2} converging to $z^*_1.$ 
The initial condition is given by
$q_1(0,0)=0$, $q_2(0,0)=0$, $x_1(0,0)=0.4$, $x_2(0,0)=0.4$.
The parameters are as follows: $\theta_1=0.6$, $\theta_2=0.5$, $k_1=0.7$, $k_2=0.4$, $\gamma_1=1$, $\gamma_2=1$, $h_1=0.01$, $h_2=0.01$. The symbol $*$ denotes the initial point and $\circ$ the point that the solution converges to (i.e., $z^*_1$). }}
\end{figure}

\subsubsection{Case 2 of Table~\ref{tab:EqPoints}}

Figure \ref{fig:P2} shows a solution to the equilibrium point $z^*_2=[\frac{k_1}{\gamma_1}, 0, 0, 0]^\top$ with $0<\frac{k_1}{\gamma_1}<\theta_1-h_1$. While both gene {\emph a} and gene {\emph b} are expressed at rate $k_i$, for gene {\emph a}, its degradation is faster than synthesis. When the concentration of protein {\emph A} $(x_1)$ is below some level, gene {\emph b} is inhibited. This confirms the result in Proposition ~\ref{eqn:Stable1}.

\begin{figure}[h!]
\begin{center}
\subfigure[$x$ components.]
{\psfragfig*[width=0.48\textwidth]{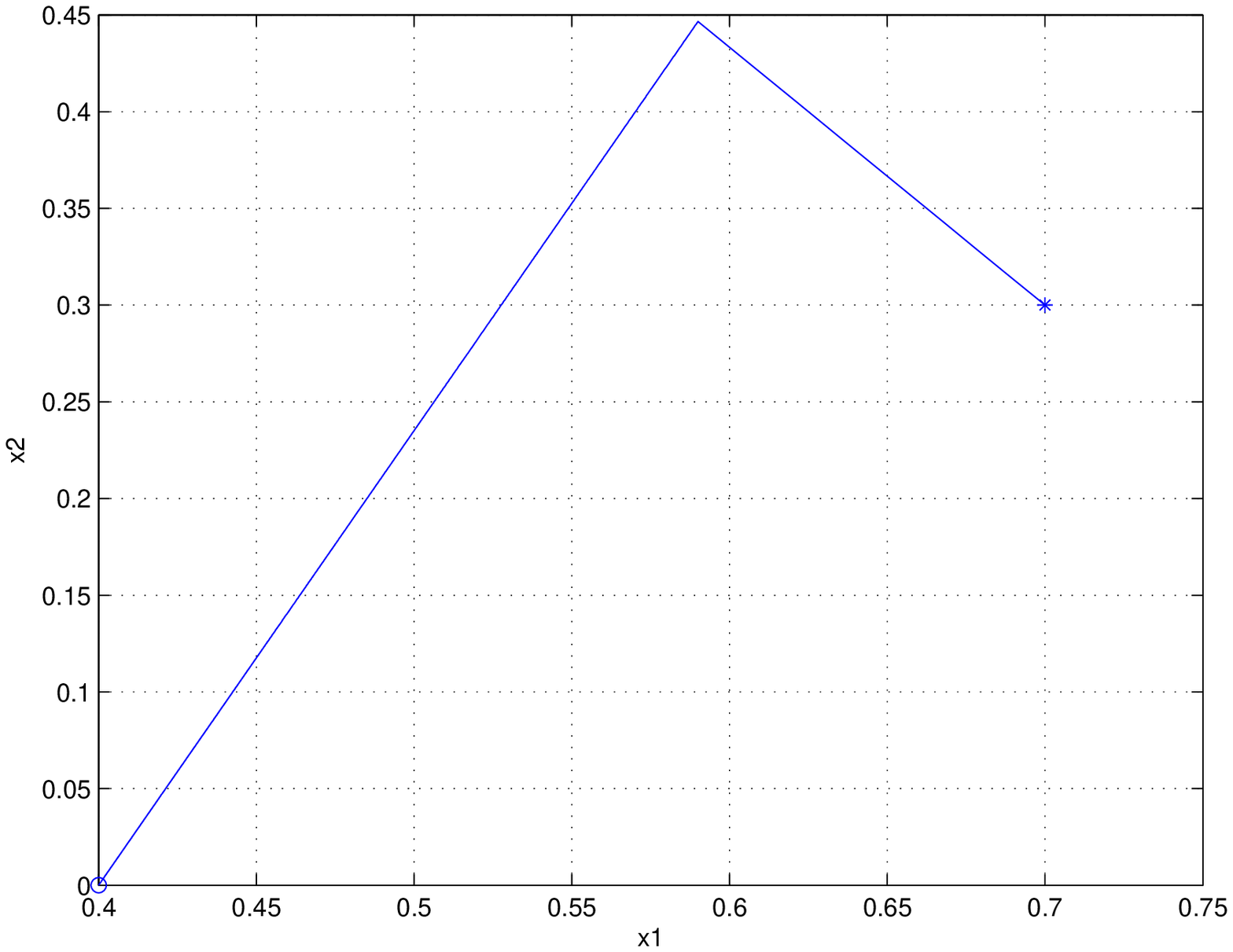}
{
\psfrag{x1}[][][0.9]{$x_1$}
\psfrag{x2}[][][0.9][-90]{$x_2$}
}
}
\subfigure[$q$ components.]
{\psfragfig*[width=0.48\textwidth]{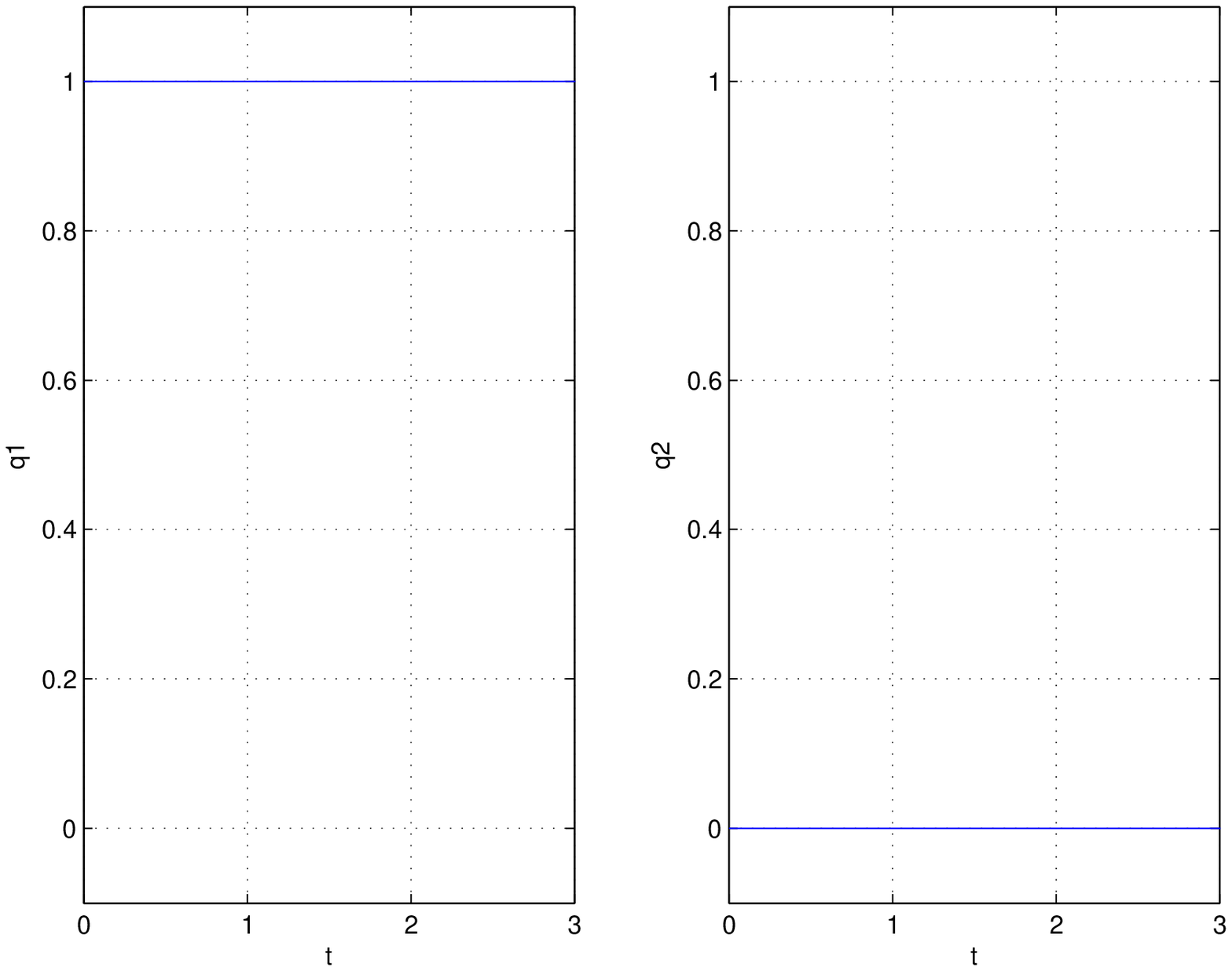}
{
\psfrag{q1}[][][0.9][-90]{$q_1$}
\psfrag{q2}[][][0.9][-90]{$q_2$}
\psfrag{t}[][][0.7]{\quad $t [sec]$}
}
}
\end{center}
\caption{\label{fig:P2}\emph{The $x$ and $q$ components of a solution to $\cal H$ in \eqref{eqn:H2} converging to $z^*_2.$ 
The initial condition is given by
$q_1(0,0)=1$, $q_2(0,0)=0$, $x_1(0,0)=0.7$, $x_2(0,0)=0.3$.
The parameters are as follows: $\theta_1=0.6$, $\theta_2=0.5$, $k_1=0.4$, $k_2=0.7$, $\gamma_1=1$, $\gamma_2=1$, $h_1=0.01$, $h_1=0.01$. The symbol $*$ denotes the initial point and $\circ$ is the point that the solution converges to  (i.e., $z_2^*$).}}
\end{figure}

\subsubsection{Case 3 of Table~\ref{tab:EqPoints}}

Figure \ref{fig:P4} indicates that, when $\theta_1-h_1 < \frac{k_1}{\gamma_1} < \theta_1+h_1$, $0 < \frac{k_2}{\gamma_2} < \theta_2+h_2$ with the initial value $z(0, 0)\in C_2:=\{q_1=1, q_2=0, x_1\geq\theta_1-h_1, x_2\leq\theta_2+h_2\}$. The solution flows towards $z^*_1=[\frac{k_1}{\gamma_1}, \frac{k_2}{\gamma_2}, 1, 0]^\top$. Under these conditions, gene {\emph a}  and gene {\emph b} are expressed at rate $k_i$, $i=1, 2$, respectively. However, for gene {\emph a}, its degradation is faster than its synthesis. This confirms the result in Proposition~\ref{eqn:Stable2}.

\begin{figure}[h!]
\begin{center}
\subfigure[$x$ components.]
{\psfragfig*[width=0.48\textwidth]{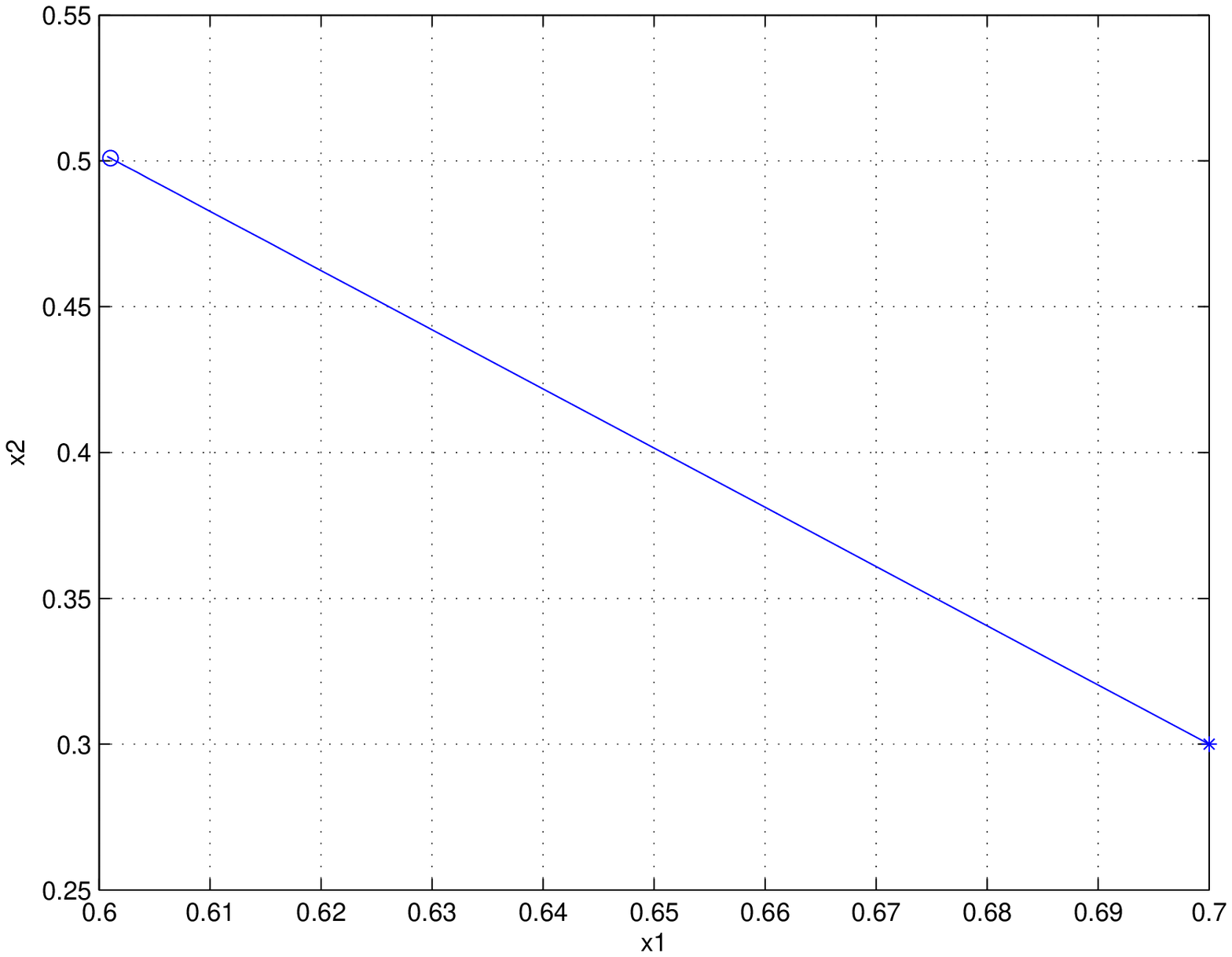}
{
\psfrag{x1}[][][0.9]{\!\!\!\!\!\!$x_1$}
\psfrag{x2}[][][0.9][-90]{$x_2$}
}
}
\subfigure[$q$ components.]
{\psfragfig*[width=0.48\textwidth]{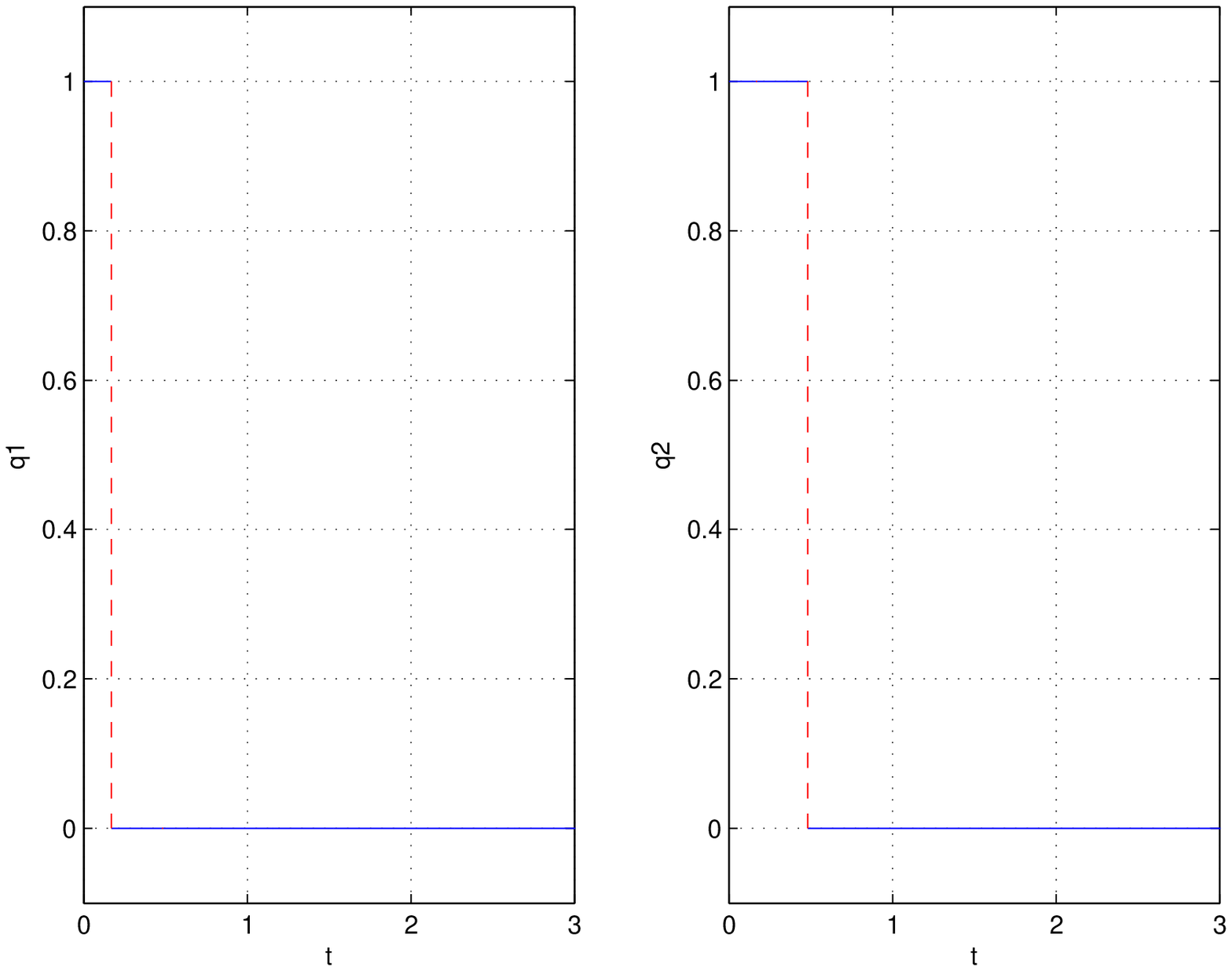}
{
\psfrag{q1}[][][0.9][-90]{$q_1$}
\psfrag{q2}[][][0.9][-90]{$q_2$}
\psfrag{t}[][][0.7]{\quad $t [sec]$}
}
}
\end{center}
\caption{\label{fig:P4}\emph{The $x$ and $q$ components of a solution to $\cal H$ in \eqref{eqn:H2} converging to $z^*_1.$ 
The initial conditions are given by
$q_1(0,0)=1$, $q_2(0,0)=0$, $x_1(0,0)=0.7$, $x_2(0,0)=0.3$.
The parameters are as follows: $\theta_1=0.6$, $\theta_2=0.5$, $k_1=0.601$, $k_2=0.501$, $\gamma_1=1$, $\gamma_2=1$, $h_1=0.02$, $h_1=0.02$. The symbol $*$ denotes the initial point and $\circ$ is the point that the solution converges to (i.e., $z_1^*$).}}
\end{figure}

\begin{figure}[h!]
\begin{center}
\subfigure[$x$ components.]
{\psfragfig*[width=0.48\textwidth]{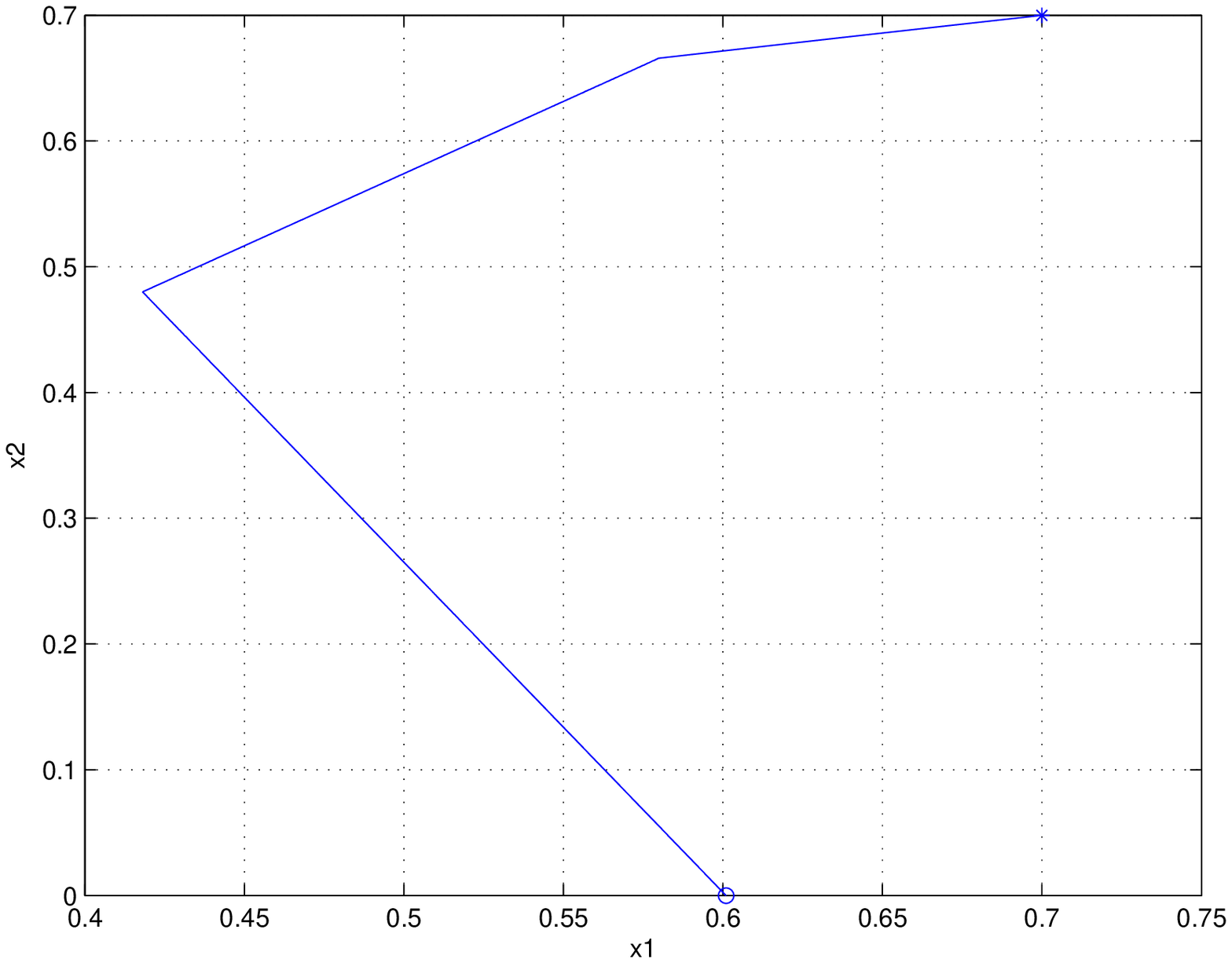}
{
\psfrag{x1}[][][0.9]{$x_1$}
\psfrag{x2}[][][0.9][-90]{$x_2$}
}
}
\subfigure[$q$ components.]
{\psfragfig*[width=0.48\textwidth]{Figures/TwoGenes_Case3b}
{
\psfrag{q1}[][][0.9][-90]{$q_1$}
\psfrag{q2}[][][0.9][-90]{$q_2$}
\psfrag{t}[][][0.7]{\quad $t [sec]$}
}
}
\end{center}
\vspace{-0.15in}
\caption{\label{fig:P3}\emph{The $x$ and $q$ components of a solution to $\cal H$ in \eqref{eqn:H2} converging to $z^*_1.$ 
The initial conditions are given by
$q_1(0,0)=1$, $q_2(0,0)=1$, $x_1(0,0)=0.7$, $x_2(0,0)=0.7$.
The parameters are as follows: $\theta_1=0.6$, $\theta_2=0.5$, $k_1=0.601$, $k_2=0.501$, $\gamma_1=1$, $\gamma_2=1$, $h_1=0.02$, $h_1=0.02$. The symbol $*$ is the initial point and $\circ$ is the point that the solution converges to (i.e., $z_2^*$).}}
\end{figure}

Figure \ref{fig:P3} illustrates the case when $\theta_1-h_1\leq\frac{k_1}{\gamma_1}\leq\theta_1+h_1$, $\frac{k_2}{\gamma_2}\leq\theta_2+h_2.$ With the initial value $z(0, 0)\notin C_2$, the solution converges to $z^*_2=[\frac{k_1}{\gamma_1}, 0, 0, 0]^\top$. With these conditions, initially, gene {\emph b} is expressed at $k_2$ and gene {\emph a} is inhibited. After finite time, as the concentration of protein {\emph A} $(x_1)$ is lower than $\theta_1-h_1$, gene {\emph b} becomes inhibited. Gene {\emph a} is expressed at $k_1$ while the concentration of protein {\emph B} $(x_2)$ is below a certain level. This confirms the result in Proposition~\ref{eqn:Stable2}.
}

\NotForConf{
\subsubsection{Case 4 of Table~\ref{tab:EqPoints}}

Figure \ref{fig:P5} indicates that, when $\theta_1-h_1 < \frac{k_1}{\gamma_1} < \theta_1+h_1$, $\theta_2+h_2<\frac{k_2}{\gamma_2} < \theta_2^{max}$ with the initial value $z(0, 0)\in C_2:=\{q_1=1, q_2=0, x_1\geq\theta_1-h_1, x_2\leq\theta_2+h_2\}$. The solution flows towards $z^*_2=[\frac{k_1}{\gamma_1}, 0, 0, 0]^\top$. Under these conditions, gene {\emph a}  and gene {\emph b} are expressed at rate $k_i$, $i=1, 2$ initially. After some time, the concentration of protein  {\emph B} exceeds a centain level, which triggers a jump, after which the expression of gene  {\emph a} is inhibited. When the concentraion of protein {\emph A}  decreases enough, another jump occurs, after which the expression of gene  {\emph b} is inhibited as well. Eventually, the concentration of protein  {\emph B} reaches a low enough value to trigger another jump, after which the expression of gene  {\emph a} is activated, and the concentrations approach a steady-state value. This simulation confirms the result in Proposition~\ref{eqn:Stable2}.

\begin{figure}[h!]
\begin{center}
\subfigure[$x$ components.]
{\psfragfig*[width=0.48\textwidth]{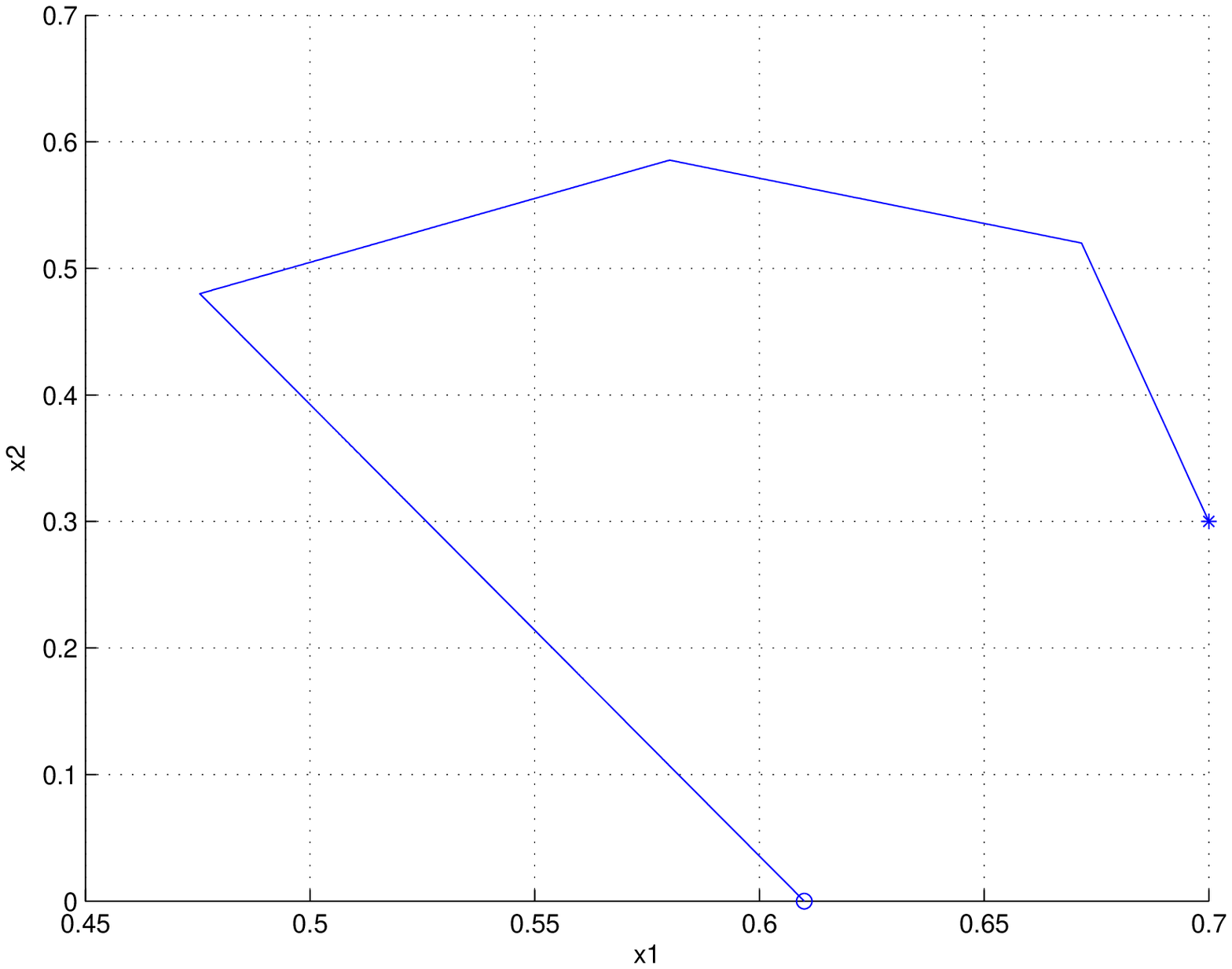}
{
\psfrag{x1}[][][0.9]{\!\!\!\!\!\!$x_1$}
\psfrag{x2}[][][0.9][-90]{$x_2$}
}
}
\subfigure[$q$ components.]
{\psfragfig*[width=0.48\textwidth]{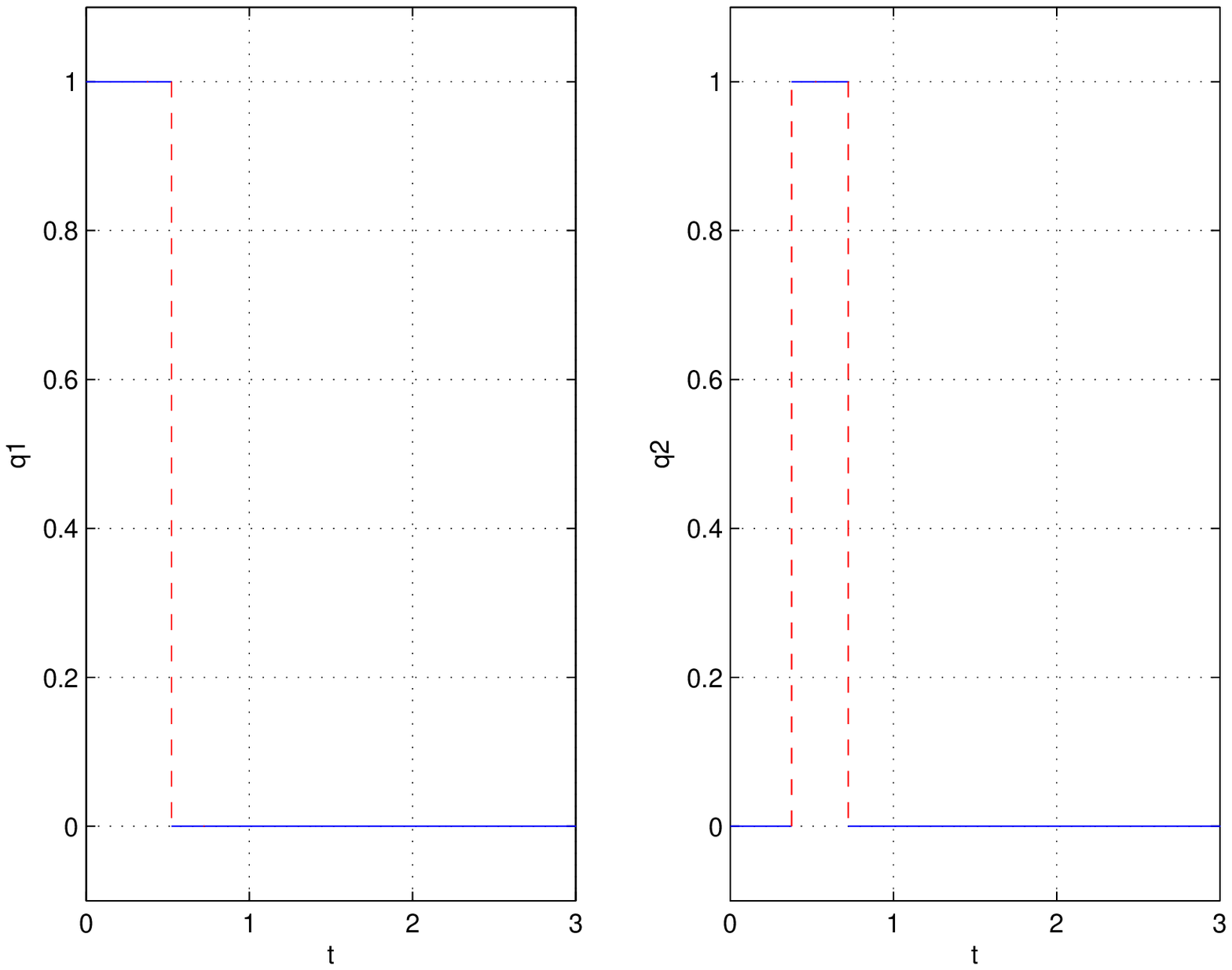}
{
\psfrag{q1}[][][0.9][-90]{$q_1$}
\psfrag{q2}[][][0.9][-90]{$q_2$}
\psfrag{t}[][][0.7]{\quad $t [sec]$}
}
}
\end{center}
\caption{\label{fig:P5}\emph{The $x$ and $q$ components of a solution to $\cal H$ in \eqref{eqn:H2} converging to $z^*_2.$ 
The initial conditions are given by
$q_1(0,0)=1$, $q_2(0,0)=0$, $x_1(0,0)=0.7$, $x_2(0,0)=0.3$.
The parameters are as follows: $\theta_1=0.6$, $\theta_2=0.5$, $k_1=0.61$, $k_2=1$, $\gamma_1=1$, $\gamma_2=1$, $h_1=0.02$, $h_1=0.02$. The symbol $*$ denotes the initial point and $\circ$ is the point that the solution converges to (i.e., $z_2^*$).}}
\end{figure}
}

\IfConf{
We illustrate numerically the more interesting case when the parameters lead to a limit cycle.
}
{
\subsection{Equilibrium set $S$}
}

When the parameters are in the region $\theta_1+h_1 < \frac{k_1}{\gamma_1} < \theta_1^{\max}, \theta_2+h_2<\frac{k_2}{\gamma_2}< \theta_2^{\max}$, the set of points $S$ in \eqref{eqn:S} defines the equilibria. 
First, we compute this set of points for particular 
values of $k_1, k_2, h_1, h_2, \gamma_1 = \gamma_2 = \gamma, \theta_1, \theta_2.$
Let $k_1=1$, $k_2=1$, $\gamma_1=\gamma_2= \gamma = 1$, $\theta_1=0.6$, $\theta_2=0.5$, $h_1=0.01$, $h_2=0.01$.
\NotForConf{
Then, using Corollary~\ref{coro:LimitCycleLineCase}, the point $p_0$ 
is given by  $p_0(1)=0.4966$. Then, from \eqref{eqn:p0}-\eqref{eqn:p3}, we obtain
$p_0=\left[\begin{array}{c}0.4966\\0.49\end{array}\right],\ p_1=\left[\begin{array}{c}0.61\\0.3796\end{array}\right],\ p_2=\left[\begin{array}{c}0.692\\0.51\end{array}\right],\ p_3=\left[\begin{array}{c}0.59\\0.5822\end{array}\right].$
With the values of $p_0, p_1, p_2, p_3$, the set $S$ in \eqref{eqn:S} is given by 
%\IfConf
%{
%$S_1 = \{x: x_2=-0.973381x_1+0.973381, 0.4966\leq x_1 \leq 0.61, 0.3796\leq x_2\leq0.49\} \times \{(0,0)\}$,
%$S_2 = \{x: x_2=1.590722x_1-0.590722, 0.61\leq x_1 \leq 0.692, 0.3796\leq x_2 \leq 0.51\} \times \{(1,0)\}$,
%$S_3 = \{x: x_2=-0.7081296x_1+1, 0.59 \leq x_1\leq 0.692, 0.51 \leq x_2\leq0.5822\} \times \{(1,1)\}$,
%$S_4 = \{x: x_2=0.9871896x_1-0.000238, 0.4966 \leq x_1\leq0.59, 0.49\leq x_2 \leq 0.5822\} \times \{(0,1)\}$.
%}
%{
\begin{eqnarray*}
S_1 &=& \{x: x_2=-0.973381x_1+0.973381,\\
 & & \hspace{0.2in} 0.4966\leq x_1 \leq 0.61, 0.3796\leq x_2\leq0.49\}\times \{(0,0)\},\\
S_2 &=& \{x: x_2=1.590722x_1-0.590722, \\
& & \hspace{0.2in} 0.61\leq x_1 \leq 0.692, 0.3796\leq x_2 \leq 0.51\}\times \{(1,0)\},\\
S_3 &=& \{x: x_2=-0.7081296x_1+1, \\
& &\hspace{0.2in} 0.59 \leq x_1\leq 0.692, 0.51 \leq x_2\leq0.5822\}\times \{(1,1)\},\\
S_4 &=& \{x: x_2=0.9871896x_1-0.000238, \\
& & \hspace{0.2in} 0.4966 \leq x_1\leq0.59, 0.49\leq x_2 \leq 0.5822\}\times \{(0,1)\}.
\end{eqnarray*}
%}
}
\begin{figure}[h!]
\begin{center}
\subfigure[$S$ \label{fig:lc}]
{\psfragfig*[width=0.4\columnwidth]{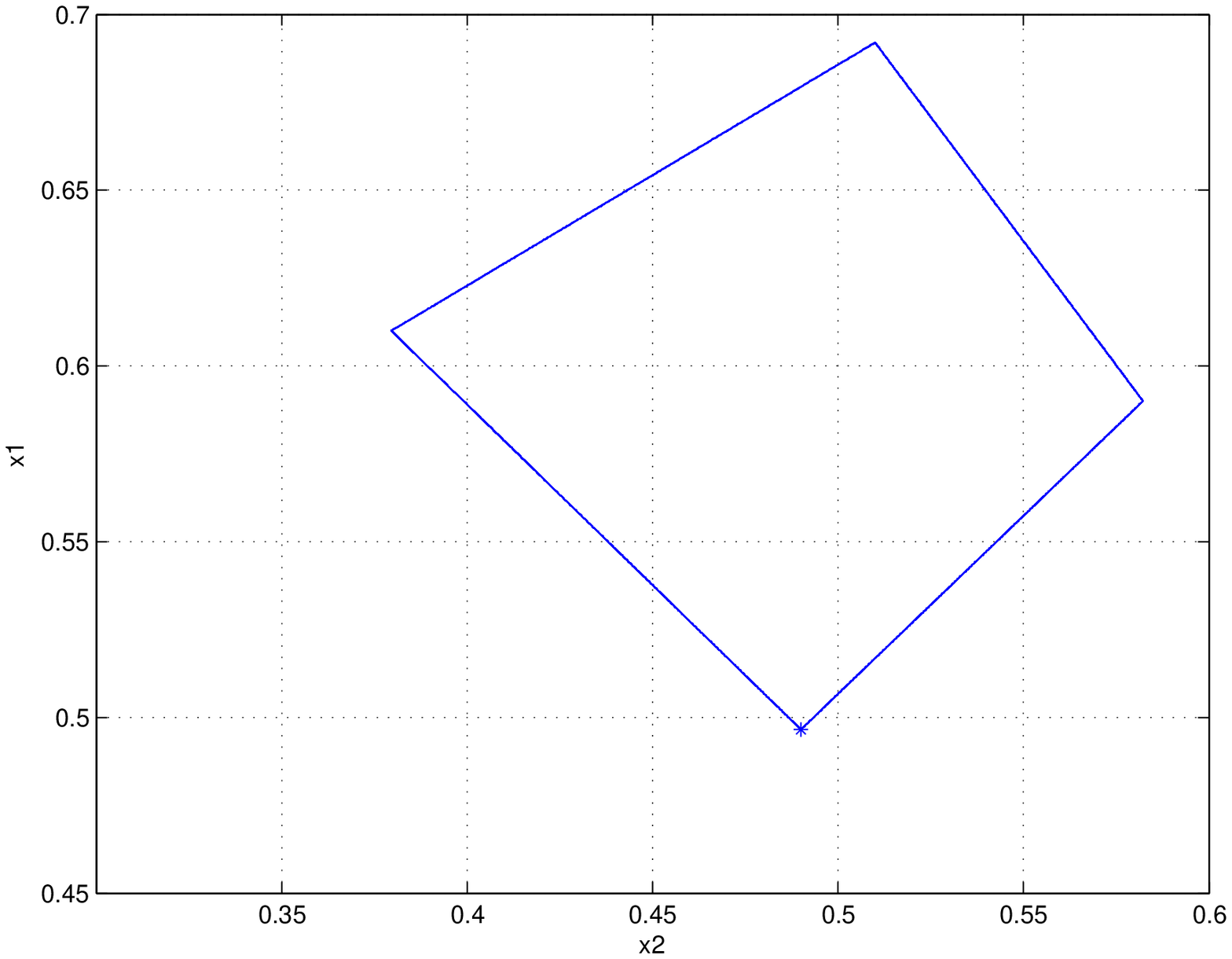}
{
\psfrag{x1}[][][0.9][-90]{$x_1$}
\psfrag{x2}[][][0.9]{\!\!\!\!\!\!$x_2$}
\psfrag{t}[][][0.85]{\!\!\!\!\!\!$t [sec]$}
}
}
\subfigure[$x_1$ component \label{fig:period}]
{\psfragfig*[width=0.4\columnwidth]{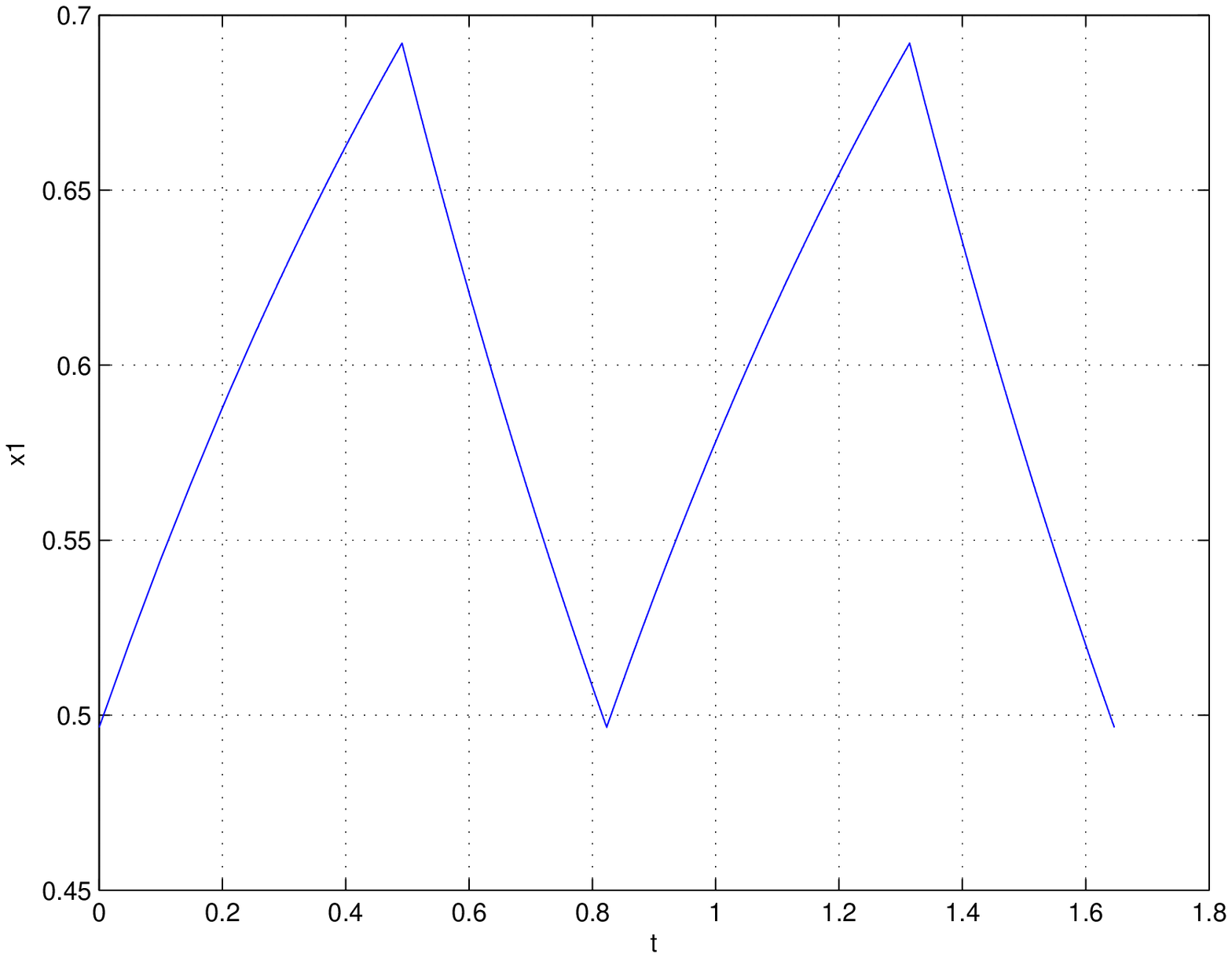}
{
\psfrag{x1}[][][0.9][-90]{$x_1$}
\psfrag{x2}[][][0.9]{\!\!\!\!\!\!$x_2$}
\psfrag{t}[][][0.85]{\!\!\!\!\!\!$t [sec]$}
}
}
\vspace{-0.15in}
\caption{\label{fig:lcmain}\emph{Set $S$ for parameters $k_1=1$, $k_2=1$, $\gamma_1=1$, $\gamma_2=1$, $\theta_1=0.6$, $\theta_2=0.5$, $h_1=0.01$, $h_2=0.01.$}}
\end{center}
\end{figure}

Figure~\ref{fig:lc} shows the set of points $S$ projected to $\reals^2$ for these parameters.
For the same parameter values, the period of the limit cycle \NotForConf{obtained from Corollary~\ref{coro:LimitCycleLineCase}} is
$T=0.8230 \mbox{ sec},$ where
$t_1' =0.2552 \mbox{ sec}$, $t_2' =0.2359 \mbox{ sec}$, $t_3' =0.1594 \mbox{ sec}$, $t_4' =0.1724 \mbox{ sec}.$ 
Figure \ref{fig:period} confirms this result.

\IfConf{
\begin{figure}[H]
\psfrag{x1}[][][0.9]{$x_1$}
\psfrag{x2}[][][0.9][-90]{$x_2$}
\begin{center}
\subfigure[$h_1=0.015$, $h_2=0.015$\label{fig:hys1}]
{\includegraphics[width=0.23\textwidth]{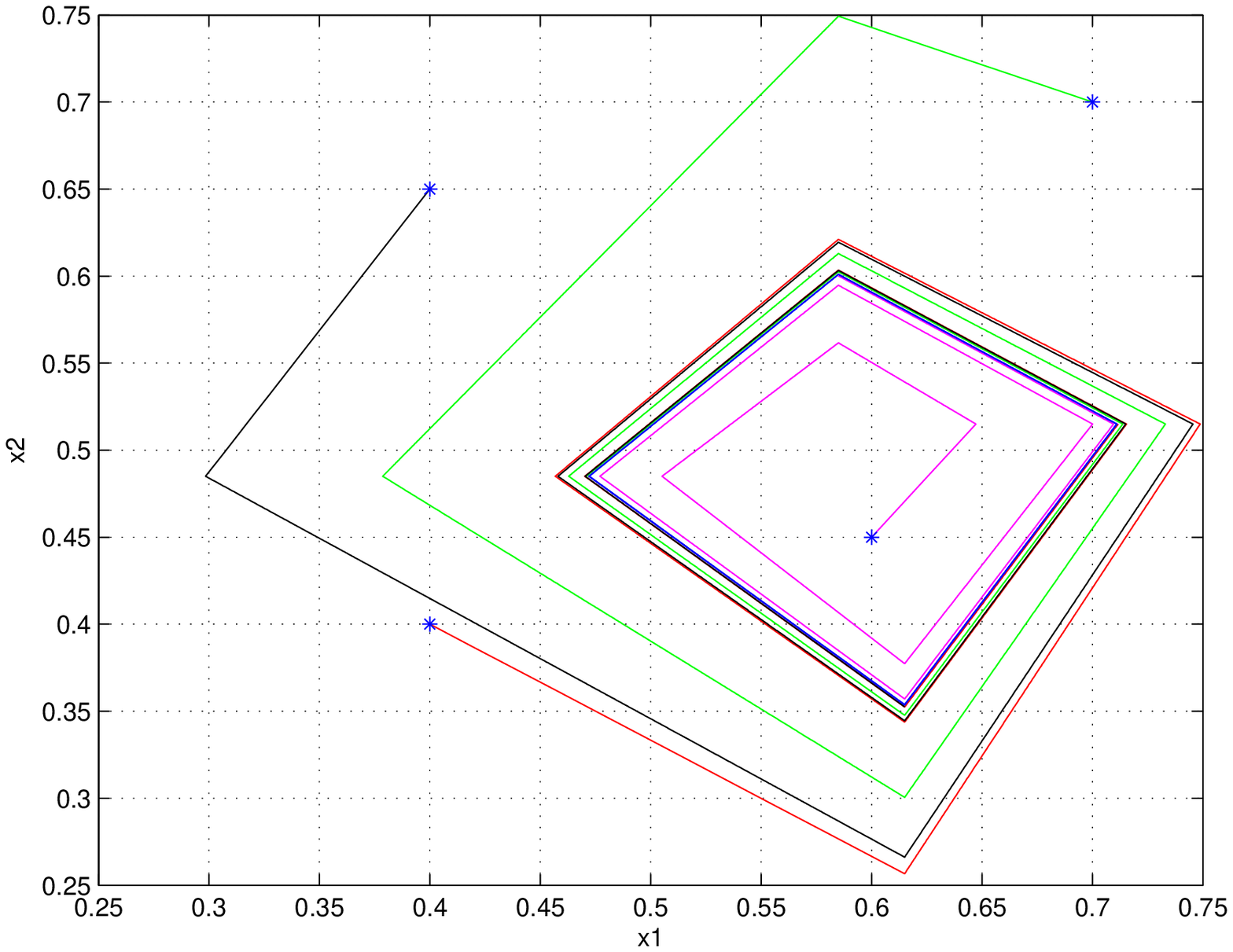}}
\psfrag{x1}[][][0.9]{$x_1$}
\psfrag{x2}[][][0.9][-90]{$x_2$}
\subfigure[$h_1=0.01$, $h_2=0.01$\label{fig:hys11}]
{\includegraphics[width=0.23\textwidth]{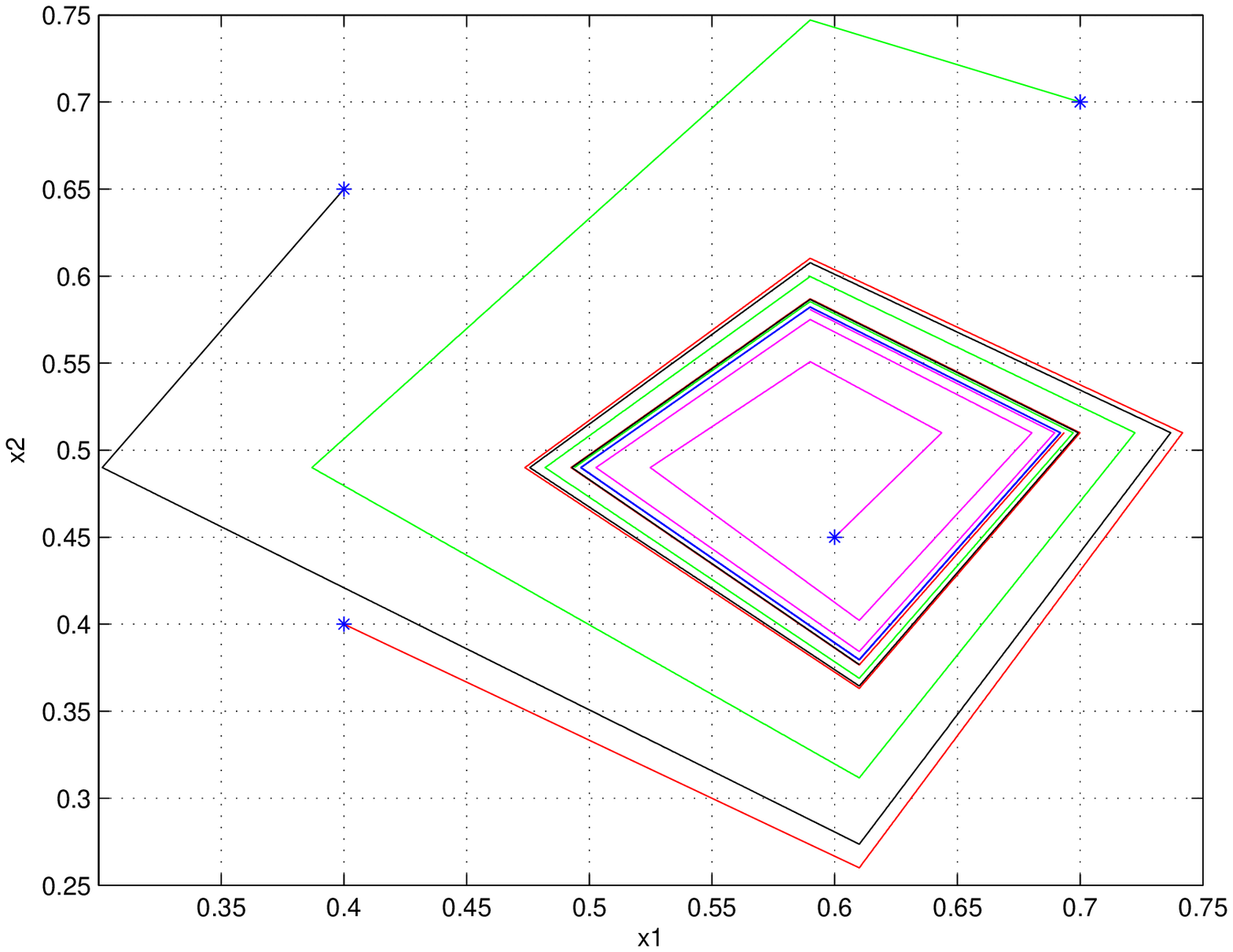}}\\
\psfrag{x1}[][][0.9]{$x_1$}
\psfrag{x2}[][][0.9][-90]{$x_2$}
\subfigure[$h_1=0.006$, $h_2=0.006$\label{fig:hys2}]
{\includegraphics[width=0.23\textwidth]{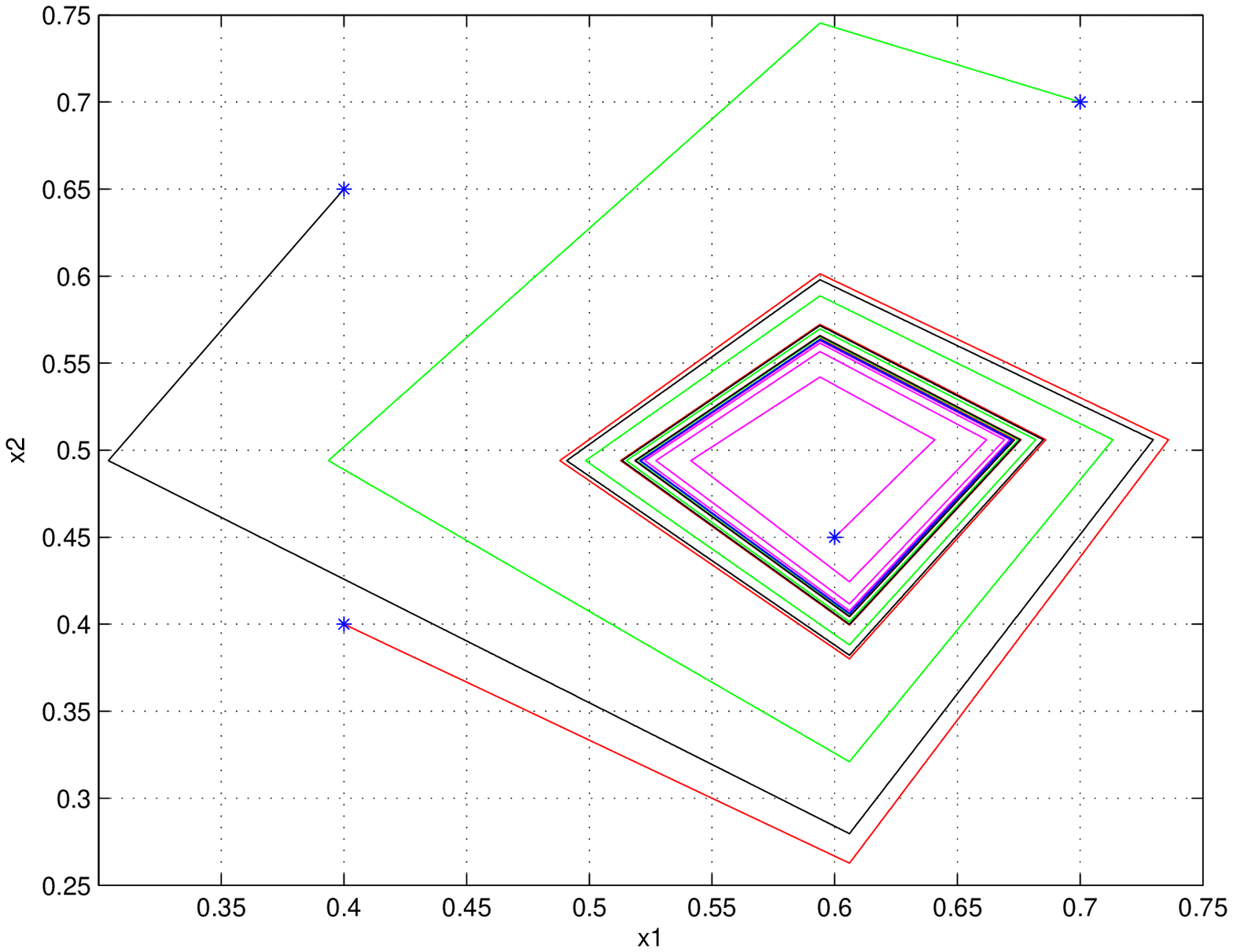}}
\subfigure[$h_1=0$, $h_2=0$\label{fig:hys4}]
{\includegraphics[width=0.23\textwidth]{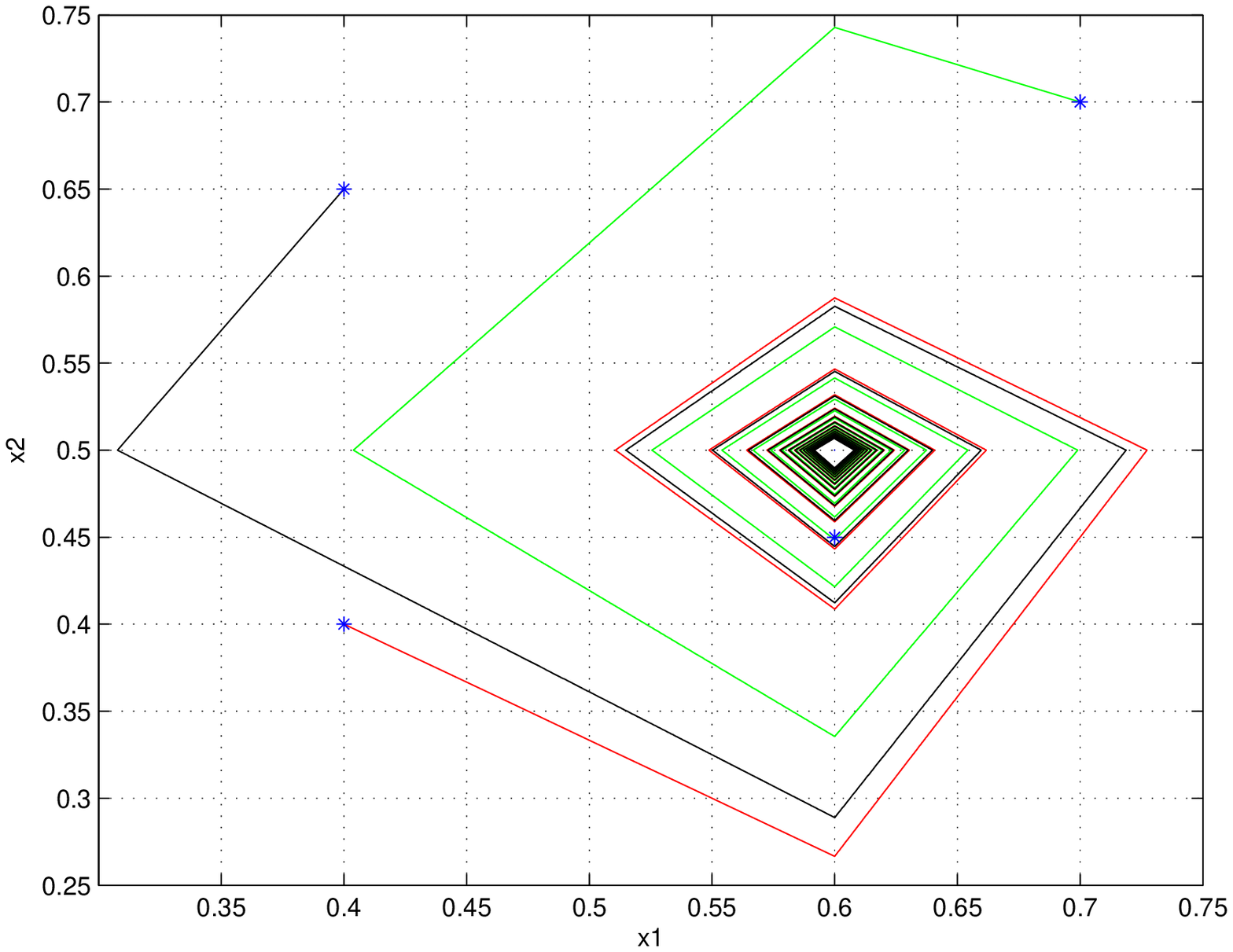}}
\end{center}
\caption{\label{fig:hys}\emph{Solutions approaching the set $S$ with different initial conditions of $z$ and fixed parameters $\theta_1=0.6$, $\theta_2=0.5$, $\gamma_1=1$, $\gamma_2=1$, $k_1=1$, $k_2=1$. }}
\end{figure}
}
{
\begin{figure}[H]
\begin{center}
\subfigure[$h_1=0.015$, $h_2=0.015$\label{fig:hys1}]
{\psfragfig*[width=0.49\textwidth]{Figures/h015}
{
\psfrag{x1}[][][0.9]{$x_1$}
\psfrag{x2}[][][0.9][-90]{$x_2$}
}
}
\subfigure[$h_1=0.01$, $h_2=0.01$\label{fig:hys11}]
{\psfragfig*[width=0.49\textwidth]{Figures/h001}
{
\psfrag{x1}[][][0.9]{$x_1$}
\psfrag{x2}[][][0.9][-90]{$x_2$}
}
}\\
\subfigure[$h_1=0.006$, $h_2=0.006$\label{fig:hys2}]
{\psfragfig*[width=0.49\textwidth]{Figures/h0006}
{
\psfrag{x1}[][][0.9]{$x_1$}
\psfrag{x2}[][][0.9][-90]{$x_2$}
}
}
\subfigure[$h_1=0$, $h_2=0$\label{fig:hys4}]
{\psfragfig*[width=0.49\textwidth]{Figures/h0}
{
\psfrag{x1}[][][0.9]{$x_1$}
\psfrag{x2}[][][0.9][-90]{$x_2$}
}
}
\end{center}
\caption{\label{fig:hys}\emph{Solutions approaching the set $S$ with different initial conditions of $z$ and fixed parameters $\theta_1=0.6$, $\theta_2=0.5$, $\gamma_1=1$, $\gamma_2=1$, $k_1=1$, $k_2=1$. }}
\end{figure}
}

Figure \ref{fig:hys} shows simulations with several initial conditions and common parameters $\theta_1=0.6,\theta_2=0.5,\gamma_1=1, \gamma_2=1,k_1=1, k_2=1$, but decreasing
$h_1, h_2$. 
Each solution  converges to the limit cycle $S$. 
The size of the limit cycle is reduced as $h_1, h_2$ gets smaller.
From our results we know that the size of the limit cycle depends on the value of hysteresis parameters. When the magnitude of hysteresis tends to zero, the set $S$ approaches a point, which is given by ($\theta_1$, $\theta_2$) (see similar case shown in Figure~\ref{fig:hys4}.) 

\NotForConf{
Figure \ref{fig:limit} shows simulations with several initial conditions and common parameters $\theta_1=0.6,\theta_2=0.5,\gamma_1=1, \gamma_2=1, h_1=0.01, h_2=0.01$, but changing $k_1, k_2$. 
Each solution converges to the limit cycle $S$ (in cyan). The blue set of points defines the limit cycle $S$ generated when $k_1= k_2 = 1.$ 
The variations of $k_1$ and $k_2$ can be considered to be perturbations as in Theorem~\ref{thm:RobustStability}.
The simulations show that the smaller the perturbation on these constants, the closer the limit cycle becomes
to the nominal one. 
Figure \ref{fig:limitcyc} shows simulations with several initial conditions and common parameters $\theta_1=0.6,\theta_2=0.5,k_1=1, k_2=1, h_1=0.01, h_2=0.01$, but now with $\gamma_1$ and $\gamma_2$ varying. 

%\IfConf{\begin{figure}[H]
%\psfrag{x1}[][][0.9]{$x_1$}
%\psfrag{x2}[][][0.9][-90]{$x_2$}
%\subfigure[$k_1=0.8$, $k_2=0.8$]{
%\includegraphics[width=0.23\textwidth]{limitcyck08.eps}\label{fig:limit1}}
%\subfigure[$k_1=0.9$, $k_2=0.9$]{
%\includegraphics[width=0.23\textwidth]{limitcyck09.eps}\label{fig:limit2}}
%\subfigure[$k_1=1.1$, $k_2=1.1$]{
%\includegraphics[width=0.23\textwidth]{limitcyck11.eps}}
%\subfigure[$k_1=1.2$, $k_2=1.2$]{
%\includegraphics[width=0.23\textwidth]{limitcyck12.eps}\label{fig:limit4}}
%\caption{\label{fig:limit}\emph{Solutions approaching the set $S$ with different initial conditions of $z$ and fixed parameters $\theta_1=0.6$, $\theta_2=0.5$, $\gamma_1=1$, $\gamma_2=1$, $h_1=0.01$, $h_2=0.01$. }}
%\end{figure}}{
\begin{figure}[H]
\subfigure[$k_1=0.8$, $k_2=0.8$]{\label{fig:limit1}
\psfragfig*[width=0.49\textwidth]{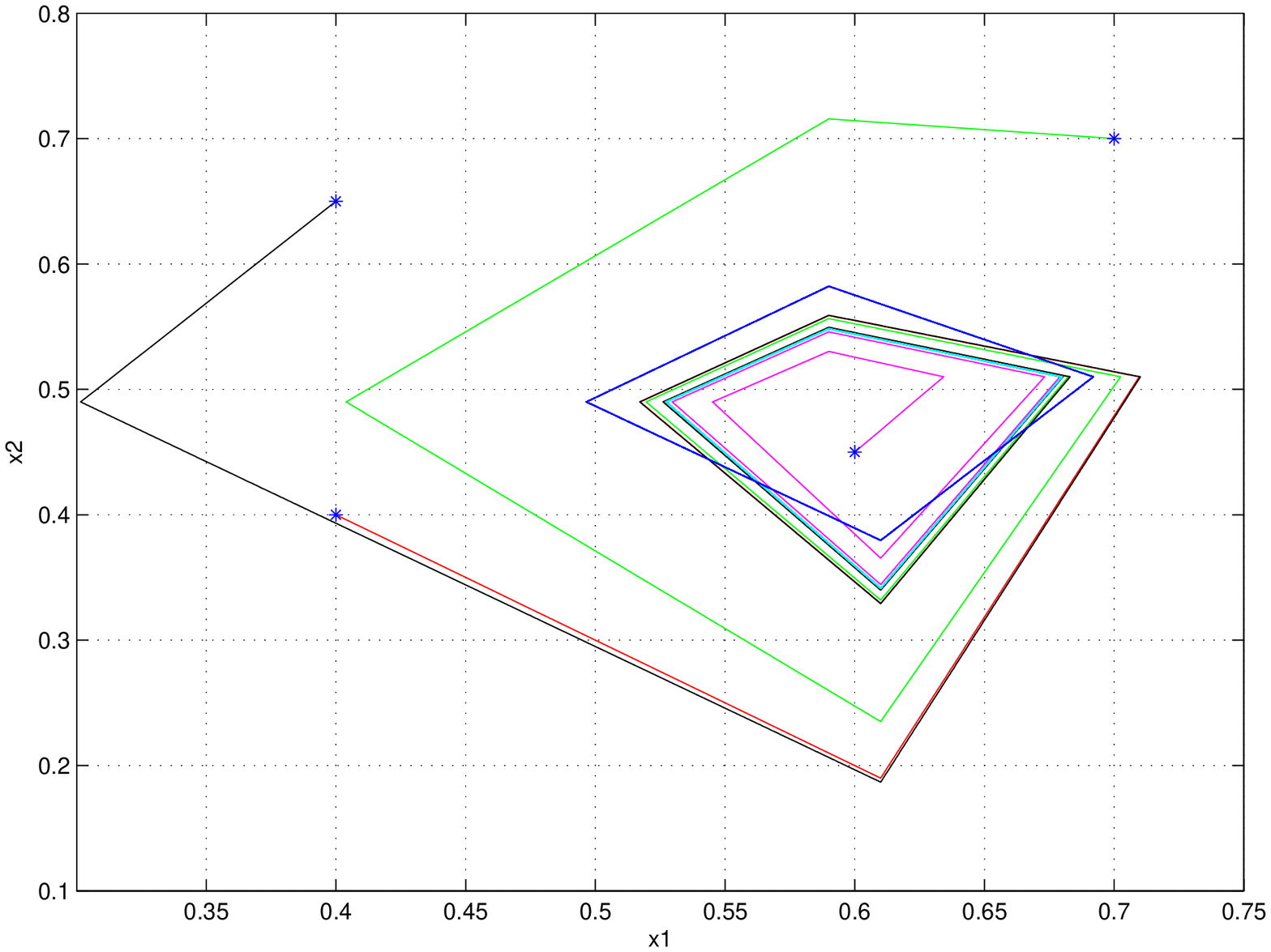}
{
\psfrag{x1}[][][0.9]{$x_1$}
\psfrag{x2}[][][0.9][-90]{$x_2$}
}
}
\subfigure[$k_1=0.9$, $k_2=0.9$]{\label{fig:limit2}
\psfragfig*[width=0.49\textwidth]{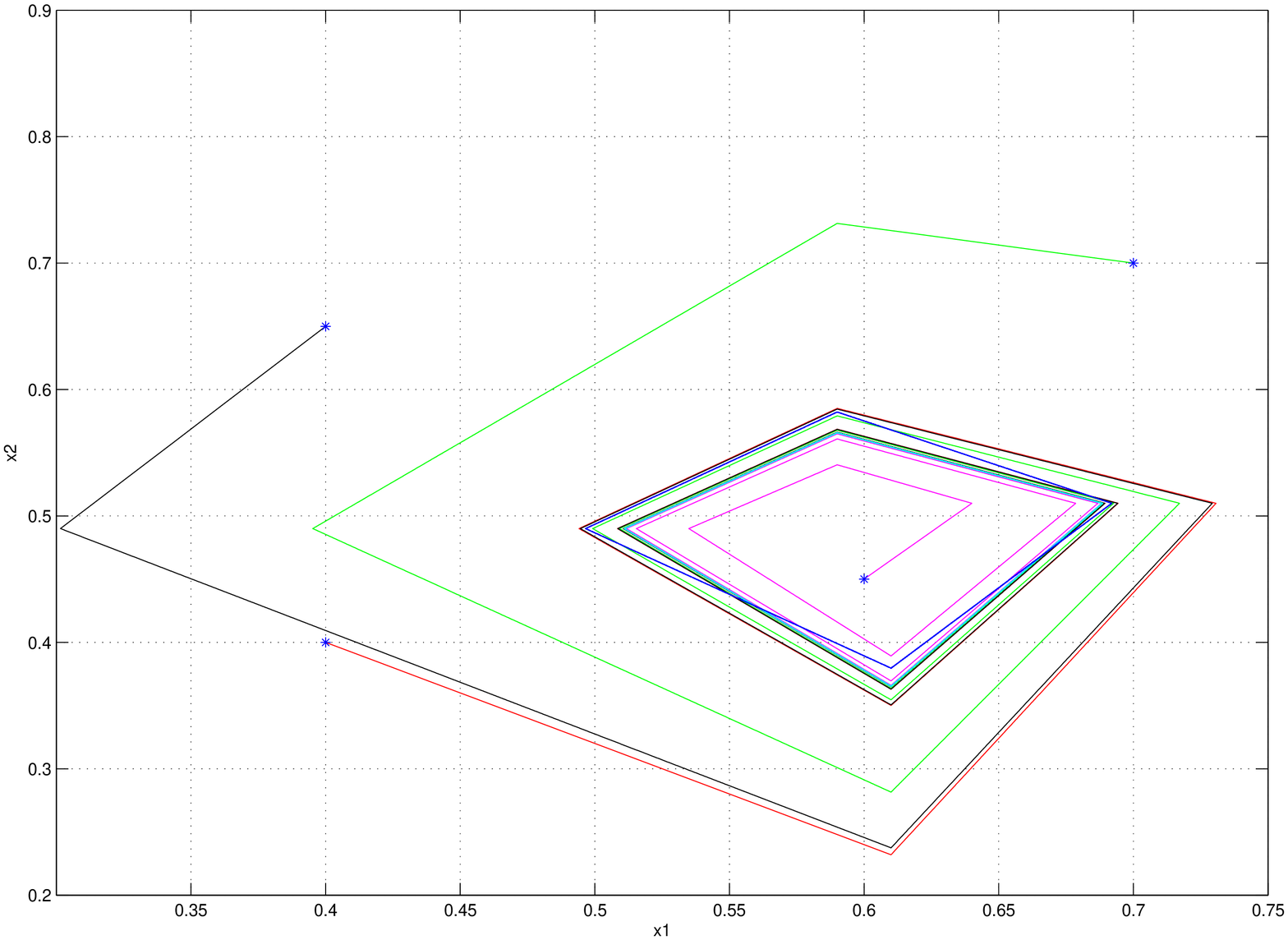}
{
\psfrag{x1}[][][0.9]{$x_1$}
\psfrag{x2}[][][0.9][-90]{$x_2$}
}
}
\subfigure[$k_1=1.1$, $k_2=1.1$]{
\psfragfig*[width=0.49\textwidth]{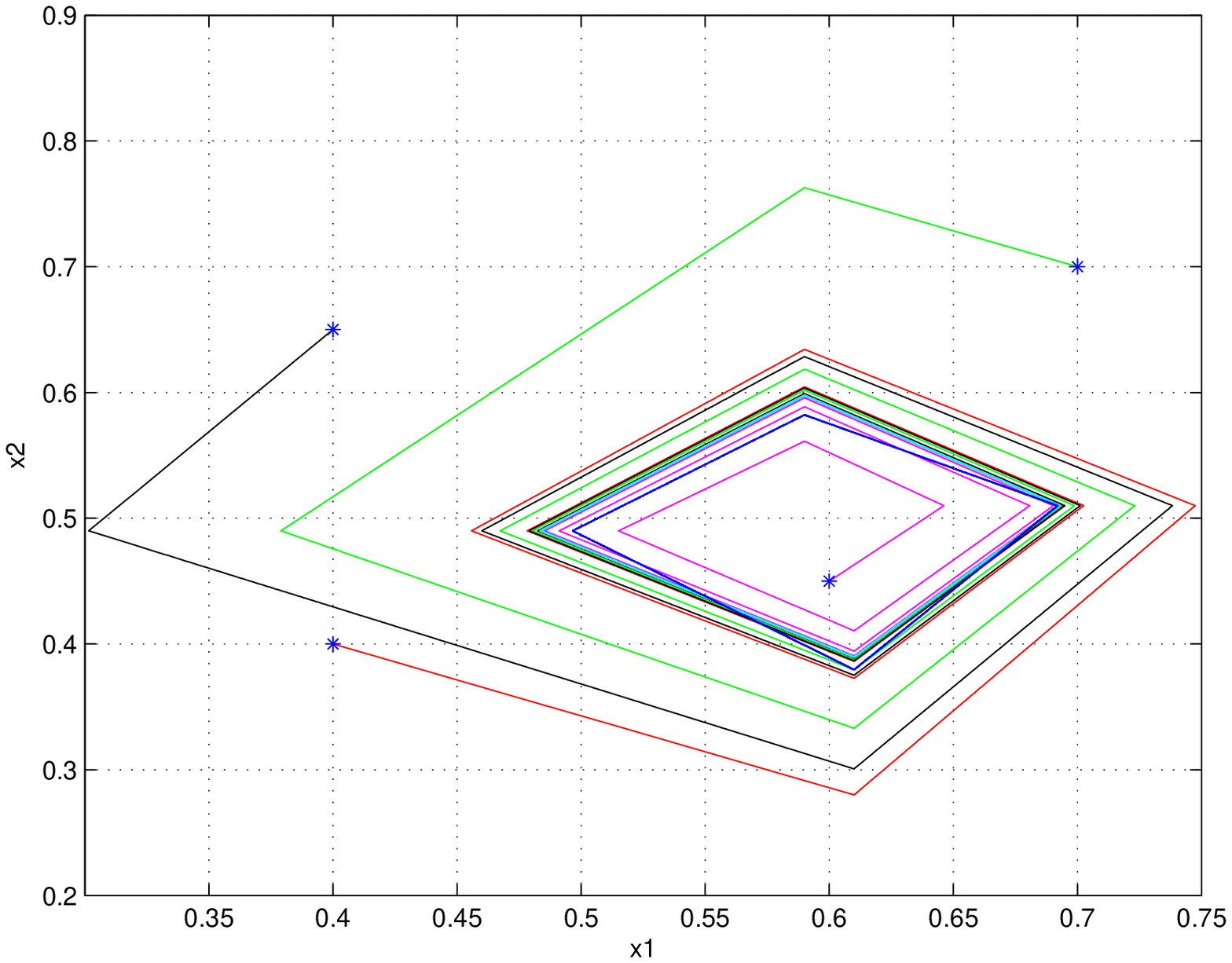}
{
\psfrag{x1}[][][0.9]{$x_1$}
\psfrag{x2}[][][0.9][-90]{$x_2$}
}
}
\subfigure[$k_1=1.2$, $k_2=1.2$]{\label{fig:limit4}
\psfragfig*[width=0.49\textwidth]{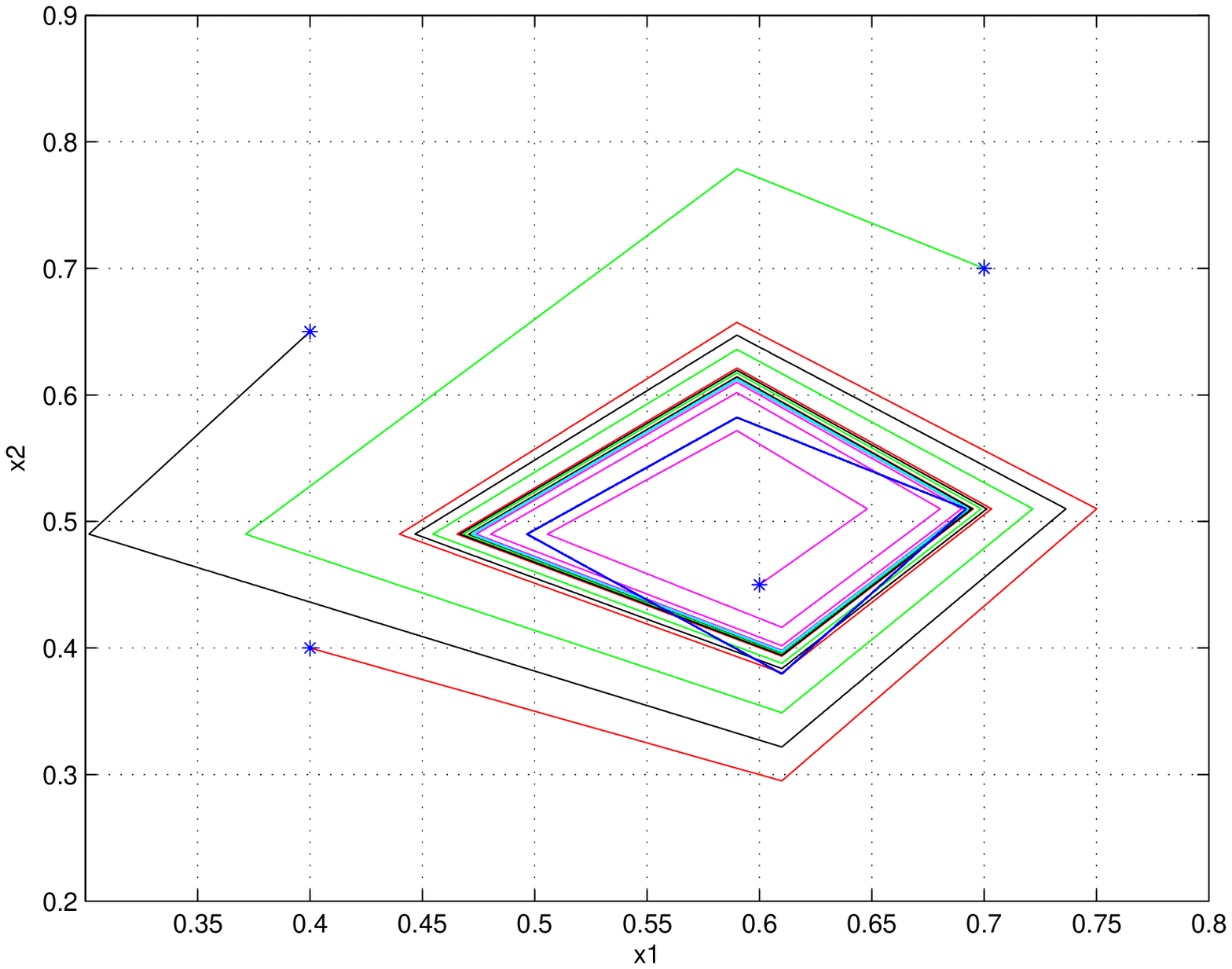}
{
\psfrag{x1}[][][0.9]{$x_1$}
\psfrag{x2}[][][0.9][-90]{$x_2$}
}
}
\caption{\label{fig:limit}\emph{Solutions approaching the set $S$ with different initial conditions of $z$ and fixed parameters $\theta_1=0.6$, $\theta_2=0.5$, $\gamma_1=1$, $\gamma_2=1$, $h_1=0.01$, $h_2=0.01$. }}
\end{figure}
%}

%\IfConf{\begin{figure}[H]
%\subfigure[$k_1=0.8$, $k_2=0.8$]{
%\includegraphics[width=0.22\textwidth]{limitcycr08.eps}\label{fig:limitcyc1}
%\psfrag{x2}[][][0.9]{$x_1$}
%\psfrag{x1}[][][0.9][-90]{$x_2$}}
%\subfigure[$\gamma_1=0.9$, $\gamma_2=0.9$]{
%\includegraphics[width=0.22\textwidth]{limitcycr09.eps}\label{fig:limitcyc2}
%\psfrag{x1}[][][0.9]{$x_1$}
%\psfrag{x2}[][][0.9][-90]{$x_2$}}
%\subfigure[$\gamma_1=1.1$, $\gamma_2=1.1$]{
%\includegraphics[width=0.22\textwidth]{limitcycr11.eps}
%\psfrag{x1}[][][0.9]{$x_1$}
%\psfrag{x2}[][][0.9][-90]{$x_2$}\label{fig:limitcyc3}}
%\subfigure[$\gamma_1=1.2$, $\gamma_2=1.2$]{
%\includegraphics[width=0.22\textwidth]{limitcycr12.eps}\label{fig:limitcyc4}
%\psfrag{x2}[][][0.9]{$x_1$}
%\psfrag{x1}[][][0.9][-90]{$x_2$}}
%\caption{\label{fig:limitcyc}\emph{Solutions approaching the set $S$ with different initial conditions of $z$ and fixed parameters $\theta_1=0.6$, $\theta_2=0.5$, $k_1=1$, $k_2=1$, $h_1=0.01$, $h_2=0.01$. }}
%\end{figure}}{
\begin{figure}[H]
\subfigure[$k_1=0.8$, $k_2=0.8$]{\label{fig:limitcyc1}
\psfragfig*[width=0.49\textwidth]{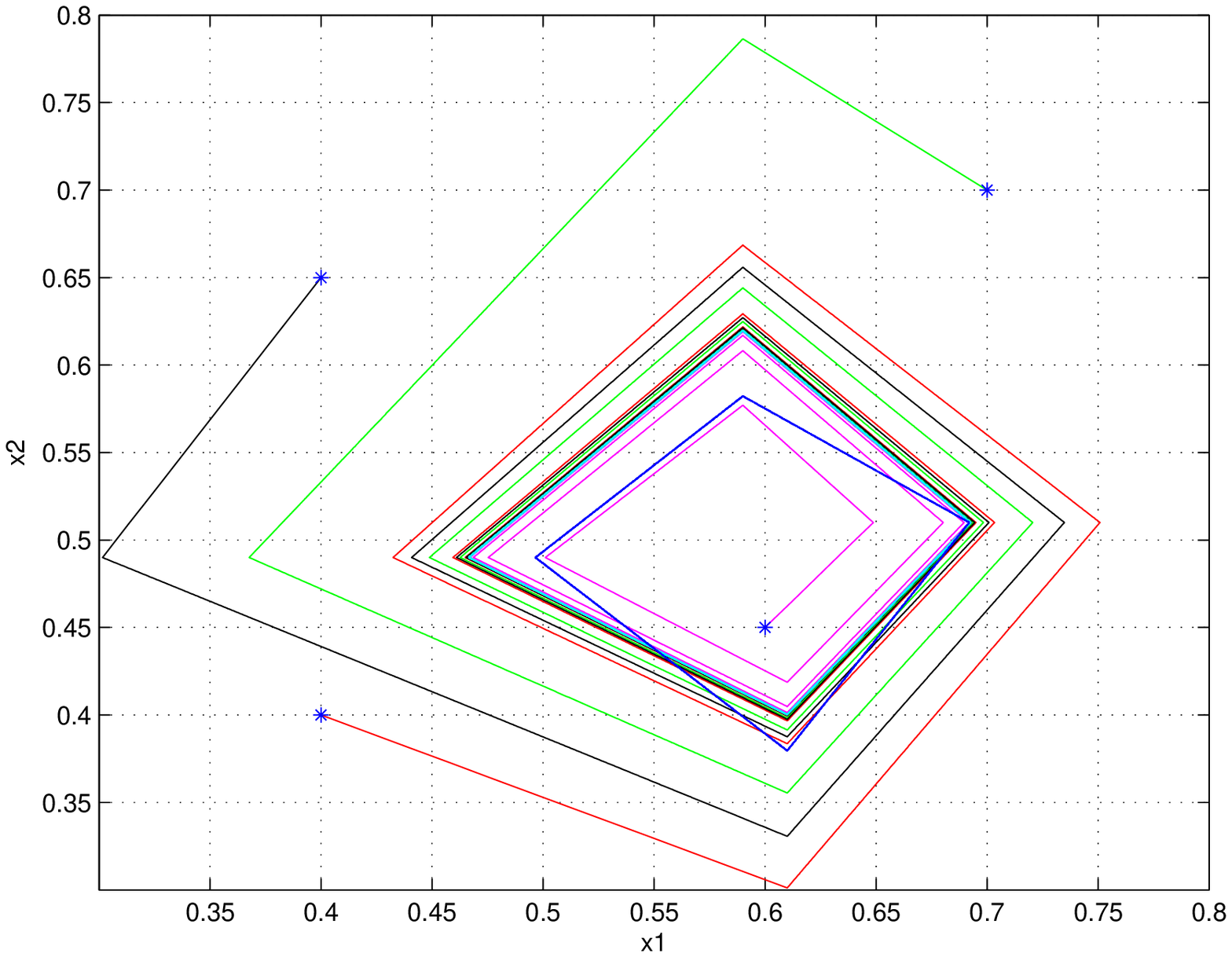}
{
\psfrag{x2}[][][0.9][-90]{$x_2$}
\psfrag{x1}[][][0.9]{$x_1$}
}
}
\subfigure[$\gamma_1=0.9$, $\gamma_2=0.9$]{\label{fig:limitcyc2}
\psfragfig*[width=0.49\textwidth]{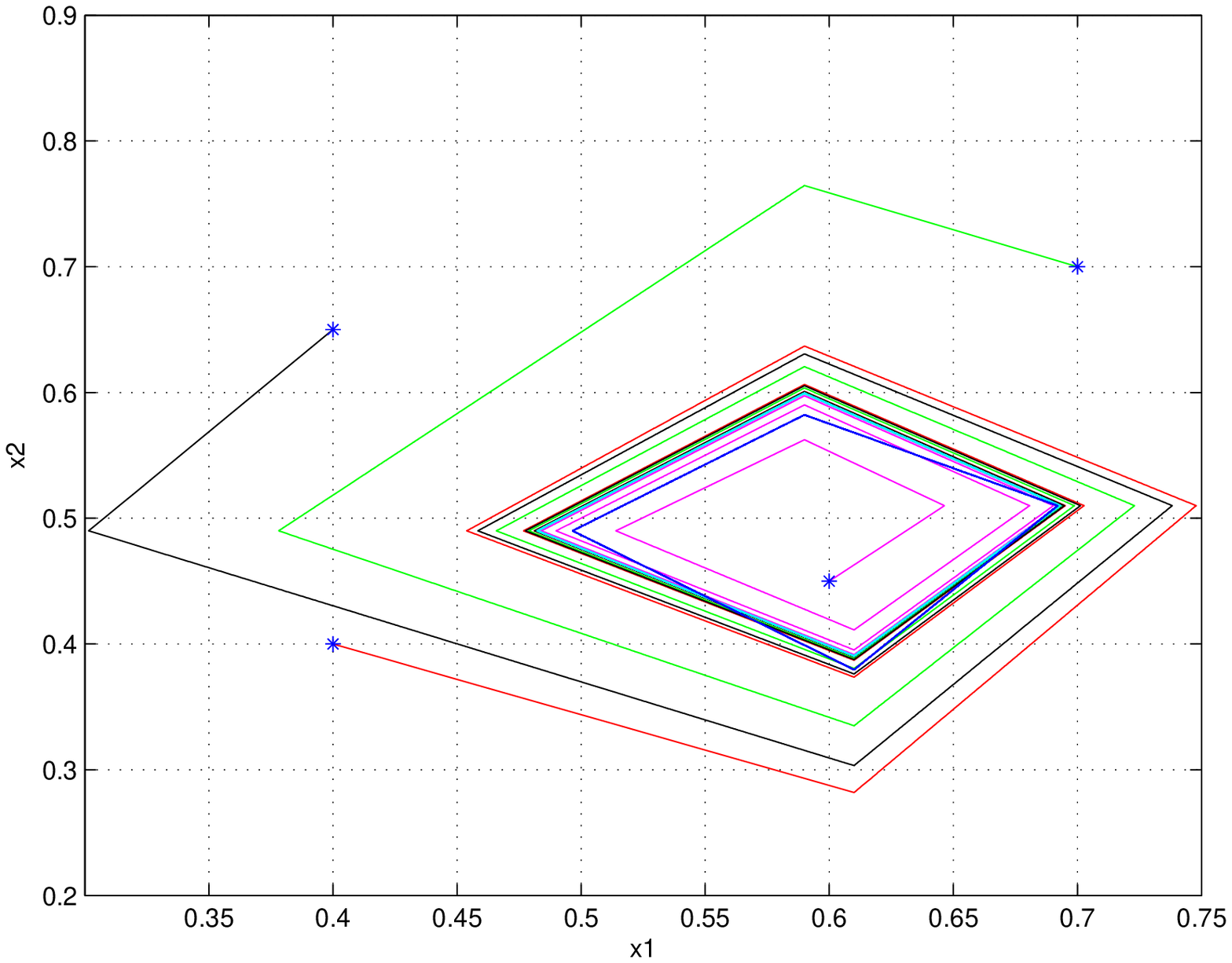}
{
\psfrag{x1}[][][0.9]{$x_1$}
\psfrag{x2}[][][0.9][-90]{$x_2$}
}
}
\subfigure[$\gamma_1=1.1$, $\gamma_2=1.1$]{\label{fig:limitcyc3}
\psfragfig*[width=0.49\textwidth]{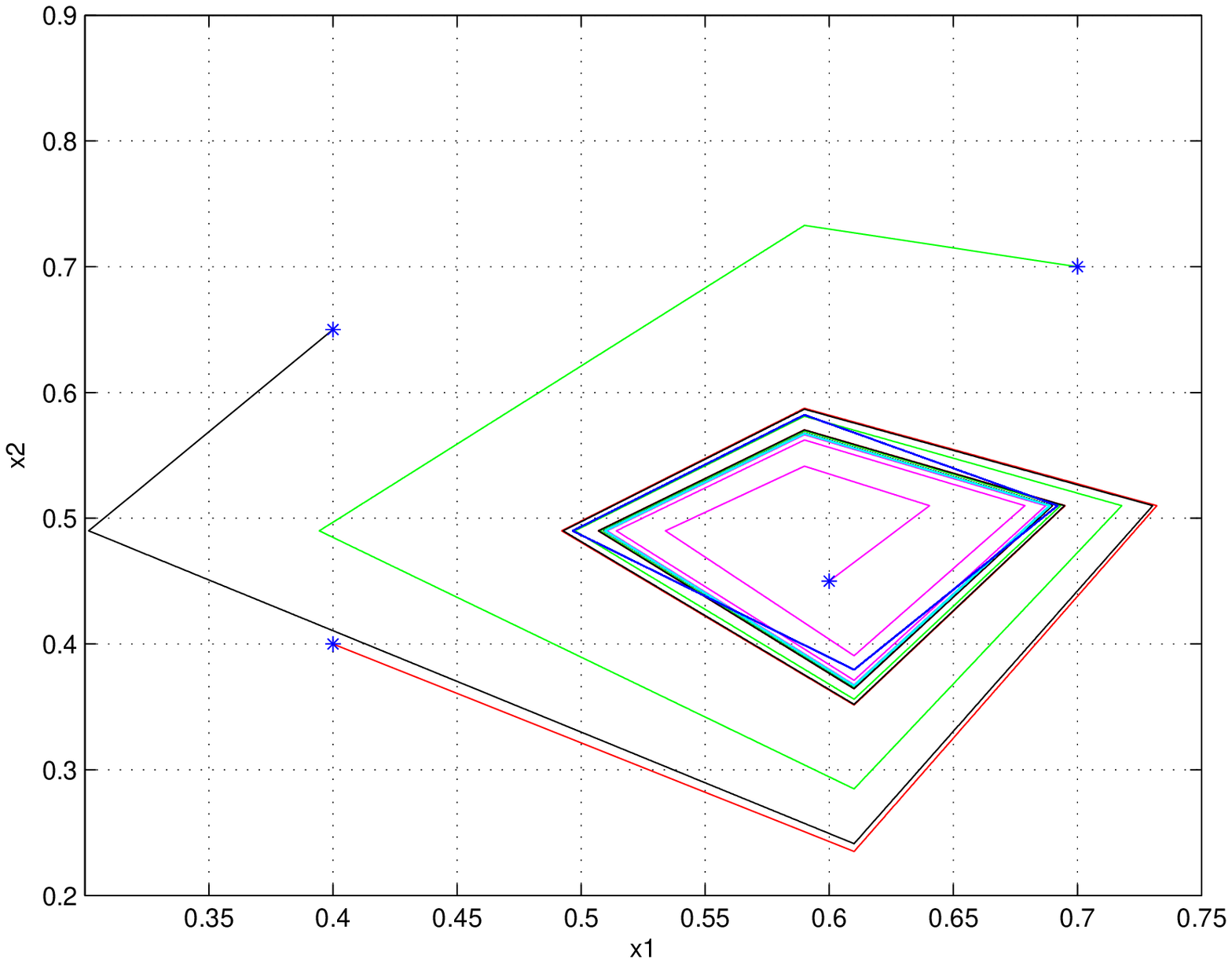}
{
\psfrag{x2}[][][0.9][-90]{$x_2$}
\psfrag{x1}[][][0.9]{$x_1$}
}
}
\subfigure[$\gamma_1=1.2$, $\gamma_2=1.2$]{\label{fig:limitcyc4}
\psfragfig*[width=0.49\textwidth]{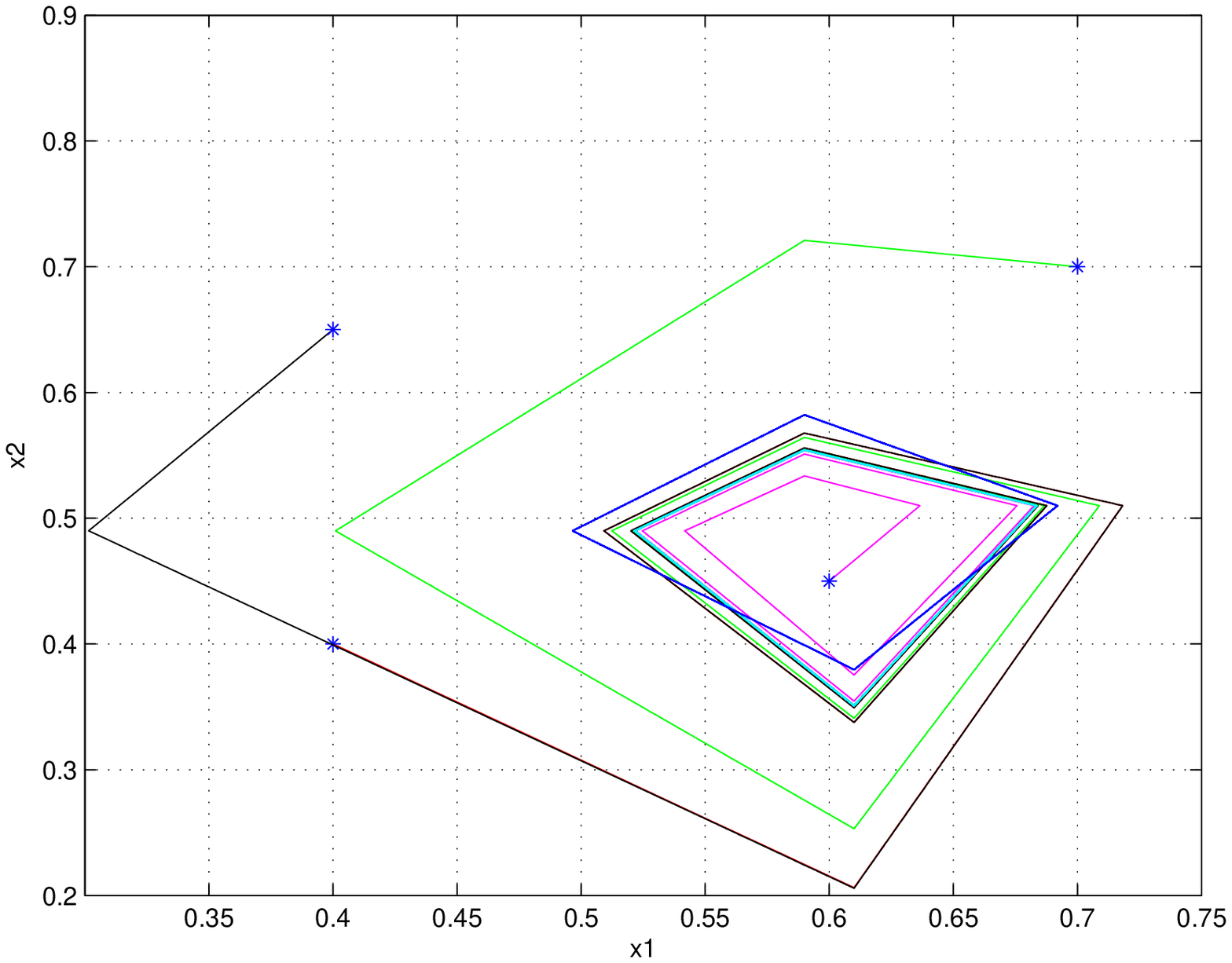}
{
\psfrag{x2}[][][0.9][-90]{$x_2$}
\psfrag{x1}[][][0.9]{$x_1$}
}
}
\caption{\label{fig:limitcyc}\emph{Solutions approaching the set $S$ with different initial conditions of $z$ and fixed parameters $\theta_1=0.6$, $\theta_2=0.5$, $k_1=1$, $k_2=1$, $h_1=0.01$, $h_2=0.01$. }}
\end{figure}
%}
}

Finally, Figure~\ref{fig:limitcurve} shows the case when $\gamma_1\neq\gamma_2$. In this case, the trajectories approach the limit cycle given in \eqref{eqn:S}.  
\IfConf{
\begin{figure}[H]
\begin{center}
\psfrag{x2}[][][0.9]{$x_1$}
\psfrag{x1}[][][0.9][-90]{\!\!$x_2$}
\includegraphics[width=0.25\textwidth]{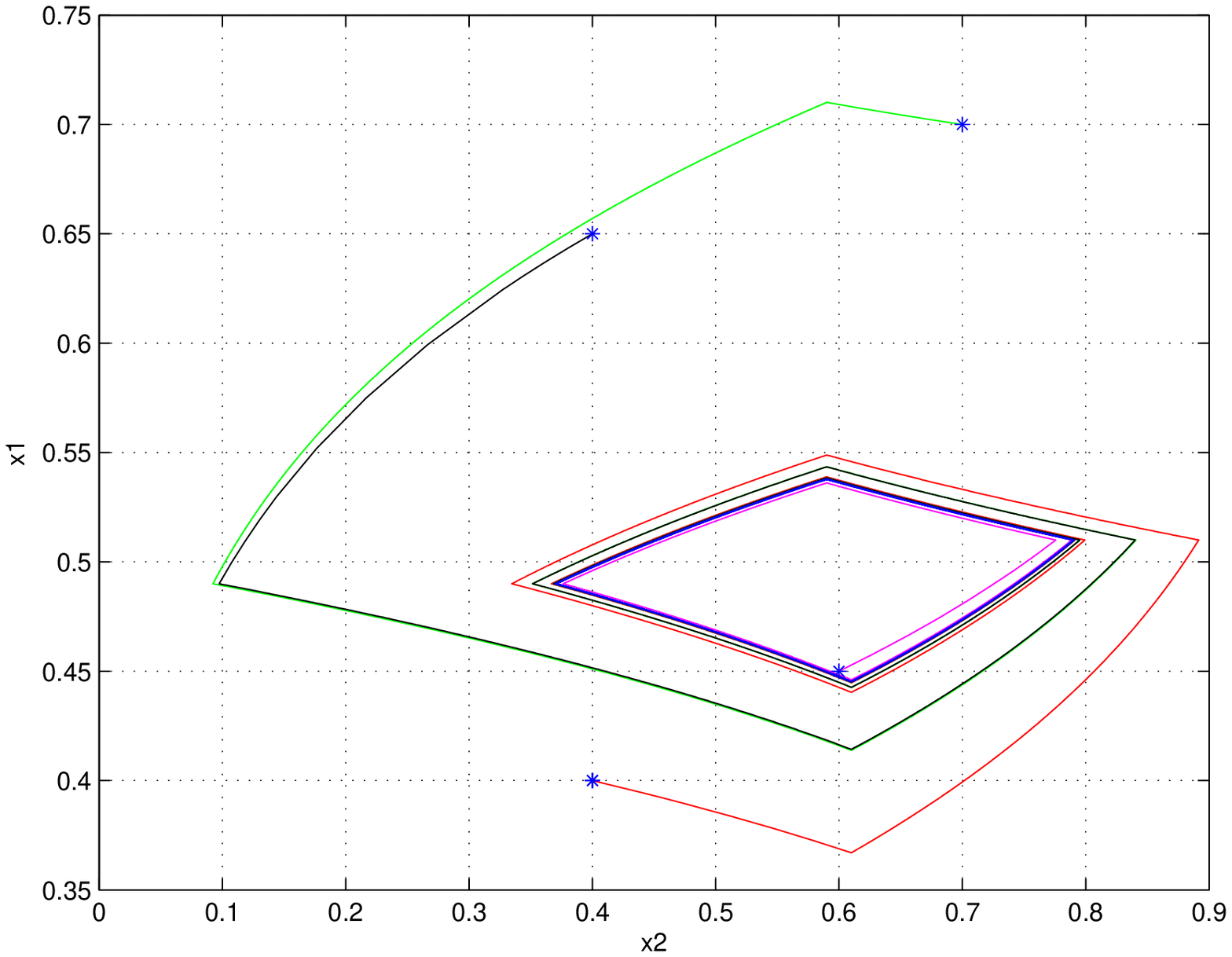}\end{center}
\caption{\label{fig:limitcurve}\emph{Solutions approaching the set $S$ with different initial conditions of $x_i$ and fixed parameters. Values of parameters: $\theta_1=0.6,\theta_2=0.5,\gamma_1=5, \gamma_2=1,k_1=5, k_2=1, h_1=0.01, h_1=0.01$. The blue line is the set $S$. 
The symbols $*$ denote the initial points.}}
\end{figure}
}
{
\begin{figure}[H]
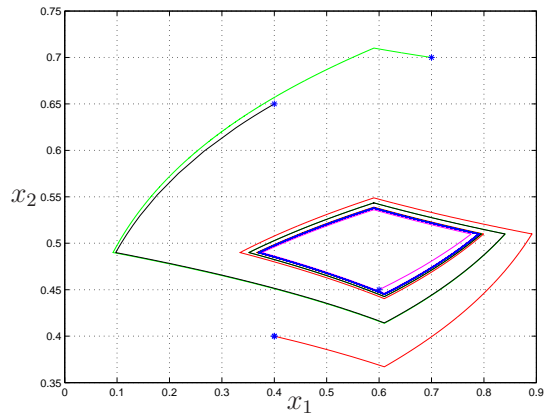

\begin{center}
\psfragfig*[width=0.49\textwidth]{Figures/limitcurve}
{
\psfrag{x2}[][][0.9]{$x_1$}
\psfrag{x1}[][][0.9][-90]{\!\!$x_2$}
}
\end{center}
\caption{\label{fig:limitcurve}\emph{Solutions approaching the set $S$ with different initial conditions of $x_i$ and fixed parameters. Values of parameters: $\theta_1=0.6,\theta_2=0.5,\gamma_1=5, \gamma_2=1,k_1=5, k_2=1, h_1=0.01, h_1=0.01$. The blue line is the set $S$. 
The symbol $*$ denotes the initial point.}}
\end{figure}
}

\section{Conclusion}

In this paper, a mathematical model of a genetic regulatory network has been developed under the formalism of hybrid dynamical systems. The model presented in this paper permits a quantitative analysis of the cellular protein dynamics under the influence of protein concentration thresholds and initial conditions. 
\IfConf{ The analysis of the hybrid model with two genes determines conditions guaranteeing the existence of solutions, the equilibria of the system, and the stability properties of the equilibria (and its robustness).}{ The analysis of the hybrid model with two genes determines conditions guaranteeing the existence of solutions, the equilibria of the system, stability properties of the equilibria and its robustness.} In particular, we have revealed conditions on the parameters that, when hysteresis is present, the interaction between the concentrations of two proteins leads to oscillatory behavior.  Such a behavior is impossible in a two-gene network without hysteresis.  \NotForConf{The obtained results are an important initial step in the analysis of genetic regulatory networks using hybrid systems theory, which we believe
has great potential for the understanding of the complex mechanisms in such networks,
in particular, when treated as (larger than two) interconnections of hybrid systems.}

\bibliographystyle{elsarticle-num}
\bibliography{bibliography}

\NotForConf{
\ExtendedVersion{
\newpage
\appendix

\section{ }

\subsection{Proof of Proposition~\ref{eqn:EP}}
\label{app:EPProof}

We consider the first three cases in Table \ref{tab:EqPoints}. Since for every point in $D$, the jump map $G$ changes the value of at least one of the logic variables, we just need to consider the case when $z^*\in C$ 
to determine isolated equilibrium points $z^*$. 
The continuous state $x$ of the system satisfies
$$\dot{x}=\left[\begin{array}{cc}
k_1(1-q_2)-\gamma_1x_1\\
 k_2q_1-\gamma_2x_2\\
\end{array}\right].
$$

To compute equilibrium points, let $ F(z^*)=0$. Then,
$$\left[\begin{array}{ll}
k_1(1-q_2^*)-\gamma_1x_1^*\\
k_2q_1^*-\gamma_2x_2^*\\
\end{array}\right]=0$$
and solving for $x^*$ leads to
$$x^*=\left[\begin{array}{cc}\frac{k_1(1-q_2^*)}{\gamma_1}\\
\frac{k_2q_1^*}{\gamma_2}\end{array}\right].$$
According to the possible values of $q^*$, all the possibilities of $x^*$ are listed in Table \ref{tab:XqConditions}. These define vectors $z_a^*$, $z_b^*$, $z_c^*$, and $z_d^*$, which are to be checked if they satisfy $z^*\in C$.

\begin{table}[h!]
\caption{\label{tab:XqConditions} Values of $x^*$ based on different combinations of the values of $q^*$.}
 \begin{center}
 \begin{tabular}{|c|c|c|}
\hline
&$x^*$& $q^*$\\\hline
$z_a^*$&$(\frac{k_1}{\gamma_1}, \frac{k_2}{\gamma_2})$ & $q_1^*=1$, $q_2^*=0$ \\\hline
$z_b^*$& $(0, 0)$&$q_1^*=0$, $q_2^*=1$  \\\hline
$z_c^*$&$(0, \frac{k_2}{\gamma_2})$&$q_1^*=1$, $q_2^*=1$  \\\hline
$z_d^*$& $(\frac{k_1}{\gamma_1}, 0)$&$q_1^*=0$, $q_2^*=0$ \\
\hline
 \end{tabular}
\end{center}
\end{table}

 Similar to $C$, the jump set $D$ can also be written as
$D =\bigcup^4_{i=1}D_i$, where
$D_1:=\{z\in \Z: q_1=0, q_2=0, x_1=\theta_1+h_1, x_2\leq\theta_2+h_2\}\cup\{z\in \Z: q_1=1, q_2=0, x_1=\theta_1-h_1, x_2\leq\theta_2+h_2\}$,
$D_2:=\{z\in \Z: q_1=1, q_2=1, x_1=\theta_1-h_1, x_2\geq\theta_2-h_2\}\cup\{z\in \Z: q_1=0, q_2=1, x_1=\theta_1+h_1, x_2\geq\theta_2-h_2\}$,
$D_3:=\{z\in \Z: q_1=1, q_2=0, x_1\geq\theta_1-h_1, x_2=\theta_2+h_2\}\cup\{z\in \Z: q_1=1, q_2=1, x_1\geq\theta_1-h_1, x_2=\theta_2-h_2\}$,
$D_4:=\{z\in \Z: q_1=0, q_2=1, x_1\leq\theta_1+h_1, x_2=\theta_2-h_2\}\cup\{z\in \Z: q_1=0, q_2=0, x_1\leq\theta_1+h_1, x_2=\theta_2+h_2\}$.
Now, we find the value of $x_1^*, x_2^*$, $q_1^*, q_2^*$ such that $z^* \in C.$

\begin{itemize}
\item Case $\rmnum{1}$: Consider parameters such that $\theta_1+h_1
<\frac{k_1}{\gamma_1}<\theta_1^{\max}$, $0<\frac{k_2}{\gamma_2}<\theta_2+h_2.$ Then, it can be checked that\\
$$z_a^*\in C, z_b^*\notin C, z_c^*\notin C, z_d^*\notin C.$$
Then, $z^*_a$ is an equilibrium point.
\item Case $\rmnum{2 }$: Consider parameters such that  $0<\frac{k_1}{\gamma_1}<\theta_1-h_1$, $0<\frac{k_2}{\gamma_2}<\theta_2+h_2.$ Then, it can be checked that\\
$$z_a^*\notin C, z_b^*\notin C, z_c^*\notin C, z_d^*\in C.$$
Then, $z_d^*$ is an equilibrium point.
\item Case $\rmnum{3}$: Consider parameters such that $\theta_1+h_1<\frac{k_1}{\gamma_1}<\theta_1^{\max}$, $\theta_2+h_2<\frac{k_2}{\gamma_2}<\theta_2^{\max}.$ Then, it can be checked that\\
$$z_a^*\notin C, z_b^*\notin C, z_c^*\notin C, z_d^*\notin C.$$
\item Case $\rmnum{4}$: Consider parameters such that $0<\frac{k_1}{\gamma_1}<\theta_1-h_1$, $\theta_2+h_2<\frac{k_2}{\gamma_2}<\theta_2^{\max}.$ Then, it can be checked that\\
$$z_a^*\notin C, z_b^*\notin C, z_c^*\notin C, z_d^*\in C.$$
Then, $z_d^*$ is an equilibrium point.
\item Case $\rmnum{5}$: Consider parameters such that $\theta_1-h_1<\frac{k_1}{\gamma_1}<\theta_1+h_1$, $0<\frac{k_2}{\gamma_2}<\theta_2+h_2.$ Then, it can be checked that\\
$$z_a^*\in C, z_b^*\notin C, z_c^*\notin C, z_d^*\in C$$
Then, $z_a^*$ and $z_d^*$ are equilibrium points.
\item Case $\rmnum{6}$: Consider parameters such that $\theta_1-h_1<\frac{k_1}{\gamma_1}<\theta_1+h_1$, $\theta_2+h_2<\frac{k_2}{\gamma_2}<\theta_2^{max}.$ Then, it can be checked that\\
$$z_a^*\notin C, z_b^*\notin C, z_c^*\notin C, z_d^*\in C$$
Then, $z_d^*$ is an equilibrium point.
\end{itemize}

From the properties above, only $z_a^*$ and $z_d^*$ are candidate isolate equilibrium points. We show that the first four cases in Table \ref{tab:EqPoints} are equilibrium points of the entire system with $z_1^*=z_a^*$ and $z_2^*=z_d^*$. For case 1, which is when $\theta_1+h_1<\frac{k_1}{\gamma_1}<\theta_1^{\max}, 0<\frac{k_2}{\gamma_2}<\theta_2+h_2$, we pick any other $z' \neq z_1^*$. We have that if $z'\in C$ then $F(z')\neq0$, so $z_1^*$ is the isolated equilibrium point of the system. For case 2, when $0<\frac{k_1}{\gamma_1}<\theta_1-h_1$, we pick any other $z'\neq z^*_2$. If $z'\in C$, then $F(z')\neq0$, so $z^*_2$ is the isolated equilibrium point of the system.  For case 3, when $\theta_1-h_1<\frac{k_1}{\gamma_1}<\theta_1+h_1$, $0<\frac{k_2}{\gamma_2}<\theta_2+h_2$, we pick any other $z'\neq z^*_i$ for each $i\in\{1, 2\}$. If $z'\in C$, then $F(z')\neq0$, so $z^{*}_1$ or $z^{*}_2$ are the equilibrium points of the entire system. For case 4, when $\theta_1-h_1<\frac{k_1}{\gamma_1}<\theta_1+h_1$, $\theta_2+h_2<\frac{k_2}{\gamma_2}<\theta_2^{max}$ we pick any other $z'\neq z^*_2$. If $z'\in C$, then $F(z')\neq0$, so $z^*_2$ is the isolated equilibrium point of the system.

We now show that $S$ in \eqref{eqn:S} is an equilibrium set
for case 5 of the parameters in Table \ref{tab:EqPoints},
namely
$\frac{k_1}{\gamma_1}\in (\theta_1+h_1,\theta_1^{\max})$, $\frac{k_2}{\gamma_2} \in (\theta_2+h_2,\theta_2^{\max})$. 
To this end, we establish that for parameters in such range,
the logic variables evolve according to the state transition graph in Figure \ref{fig:T}.
This is due to the system not having an isolated equilibrium point in 
the flow set $C$; see  Case {\it iii} in the proof of Proposition~\ref{eqn:EP}.
Let $Q:=\{(0, 0), (0, 1), (1, 0), (1, 1)\}$.
The $x$ component of the vector field $F$ 
satisfies the following properties on the boundary of $C$.
\begin{figure}[h]
\begin{center}
\psfrag{C1}[][][0.9]{$C_1$}
\psfrag{C2}[][][0.9]{$C_2$}
\psfrag{C3}[][][0.9]{$C_3$}
\psfrag{C4}[][][0.9]{$C_4$}
\psfrag{x1}[][][0.9]{\ \ $x_1$}
\psfrag{x2}[][][0.9]{$x_2$}
\psfrag{q10}[][][0.9]{$q_1\!\!=\!0$}
\psfrag{q21}[][][0.9]{$q_2\!\!=\!1$}
\psfrag{q212}[][][0.9]{\hspace{-0.5in}$q_1\!\!=\!1$\ \ \ \ $q_2\!\!=\!1$}
\psfrag{q11}[][][0.9]{$q_1\!\!=\!1$}
\psfrag{q20}[][][0.9]{$q_2\!\!=\!0$}
\psfrag{t1h1z}[][][0.9]{$\theta_1$}
\psfrag{t1-h1}[][][0.75]{$\theta_1\!\!-\!h_1$}
\psfrag{t1h1}[][][0.75]{$\theta_1\!\!+\!h_1$}
\psfrag{t1m}[][][0.75]{$\theta_1^{\max}$}
\psfrag{t2}[][][0.75]{$\theta_2$}
\psfrag{t2-h2}[][][0.75]{\!\!\!\!\!\!\!$\theta_2-h_2$}
\psfrag{t2h2}[][][0.75]{\!\!\!\!\!\!$\theta_2+h_2$}
\psfrag{t2m}[][][0.75]{$\theta_2^{\max}$}
\includegraphics[width=0.85\columnwidth]{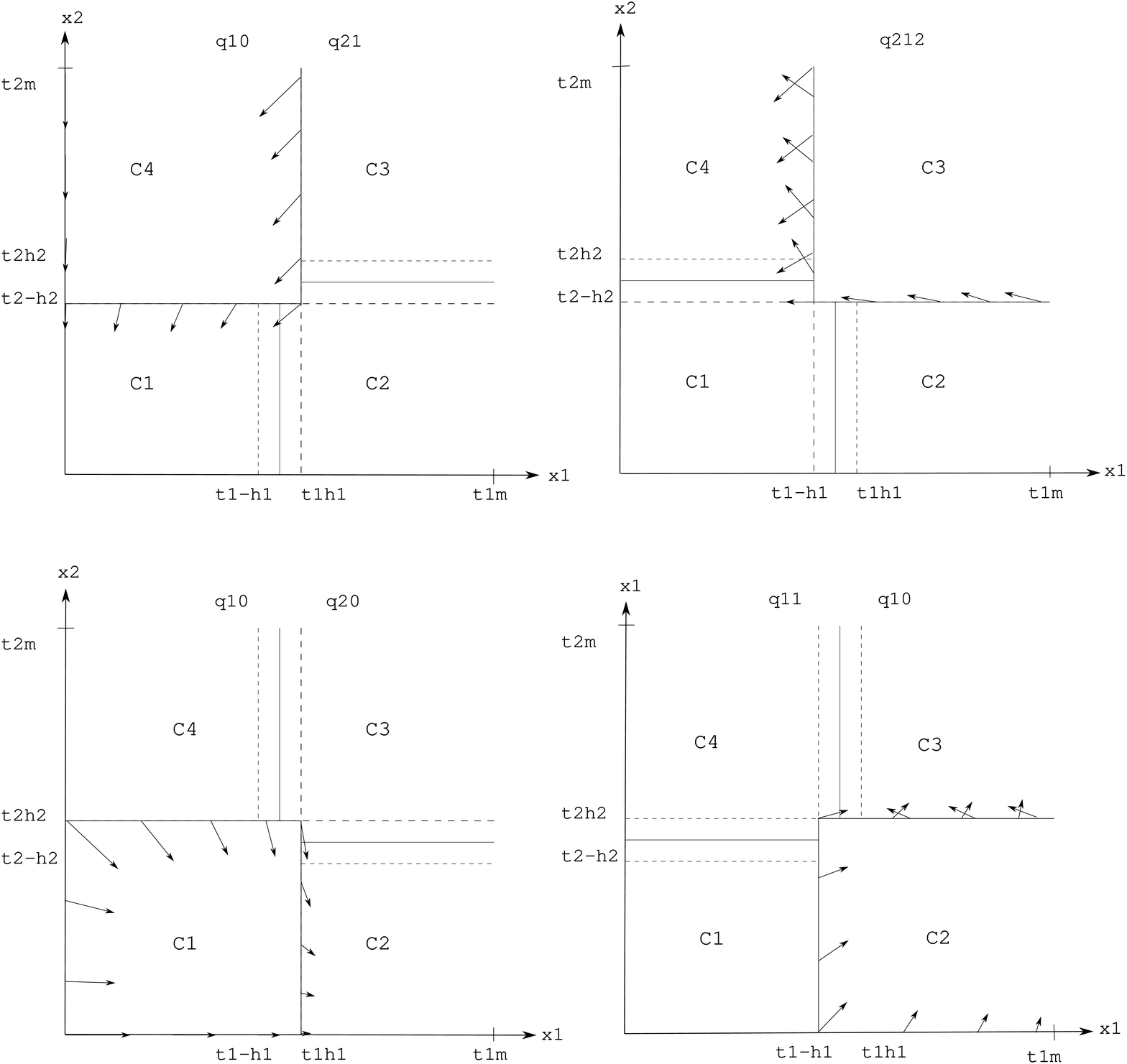}
\end{center}
\caption{\emph{Vector fields on the boundaries of $C$ to case 5 of  Table~\ref{tab:EqPoints}.}\label{fig:Vectorfieldofc}}
\end{figure}

\begin{enumerate}
\item
Points on the boundary of $C_1$:

\begin{itemize}
\item
When $x\in\{x: x_1=0, 0\leq x_2\leq\theta_2+h_2\},$ the first two components of $F$ define the vector $f_1(x):=\left[\begin{array}{cc}k_1\\-\gamma_2x_2\end{array}\right]$. Since $k_1>0$, $\gamma_2>0$, then $f_1(x)$ points inside of $C_1.$ 
\item
When $x\in\{x: 0\leq x_1\leq\theta_1+h_1, x_2=0\},$ $f_1(x):=\left[\begin{array}{cc}k_1-\gamma_1x_1\\0\end{array}\right]$. Since $\theta_1+h_1<\frac{k_1}{\gamma_1}<\theta^{\max}_{1}$, then $f_1(x)$ is tangent to the boundary of $C_1.$ 
\item
When $x\in\{x: 0\leq x_1\leq\theta_1+h_1, x_2=\theta_2+h_2\},$ $f_1(x):=\left[\begin{array}{cc}k_1-\gamma_1x_1\\-\gamma_2(\theta_2+h_2)\end{array}\right]$. Since $\theta_1+h_1<\frac{k_1}{\gamma_1}<\theta^{\max}$,$\gamma_2>0$, then $f_1(x)$ points inside of $C_1.$ 
\item
When $x\in\{x: x_1=\theta_1+h_1, 0\leq x_2\leq\theta_2+h_2\},$ $f_1(x):=\left[\begin{array}{cc}k_1-\gamma_1(\theta_1+h_1)\\-\gamma_2x_2\end{array}\right]$. Since $\theta_1+h_1<\frac{k_1}{\gamma_1}<\theta^{\max}$,$\gamma_2>0$, then $f_1(x)$ points outside of $C_1.$ 
\end{itemize}
Then, since there is no isolated equilibrium point in $C_1$, trajectories starting in $C_1$ are such that the $x$ component flow towards $\{x_1=\theta_1+h_1, 0\leq x_2\leq\theta_2+h_2\}$.
\item

Points on the boundary of $C_2$:

\begin{itemize}
\item
When $x\in\{x: x_1=\theta_1-h_1, 0\leq x_2\leq\theta_2+h_2\},$ 
$f_2(x):=\left[\begin{array}{cc}k_1-\gamma_1(\theta_1-h_1)\\k_2-\gamma_2x_2\end{array}\right]$. Since $\theta_1+h_1<\frac{k_1}{\gamma_1}<\theta^{\max}_{1}$, $\theta_2+h_2<\frac{k_2}{\gamma_2}<\theta^{\max}_{2}$, then $f_2(x)$ points inside of $C_2.$ 
\item
When $x\in\{x: x_1\geq\theta_1-h_1, x_2=0\},$ $f_2(x):=\left[\begin{array}{cc}k_1-\gamma_1x_1\\k_2\end{array}\right]$. Since $\theta_1+h_1<\frac{k_1}{\gamma_1}<\theta^{\max}$,$\theta_2+h_2<\frac{k_2}{\gamma_2}<\theta_2^{\max}$, then $f_2(x)$ points inside of $C_2.$ 
\item
When $x\in\{x: x_1\geq\theta_1-h_1, x_2=\theta_2+h_2\},$ $f_2(x):=\left[\begin{array}{cc}k_1-\gamma_1x_1\\k_2-\gamma_2(\theta_2+h_2)\end{array}\right]$. Since $\theta_1+h_1<\frac{k_1}{\gamma_1}<\theta^{\max}_{1}$, $\theta_2+h_2<\frac{k_2}{\gamma_2}<\theta_2^{\max}$, then $f_2(x)$ points outside of $C_2.$ 
\end{itemize}
Then, since there is no isolated equilibrium point in $C_2$, trajectories starting in $C_2$ are such that the $x$ component flow towards $\{x_1\geq\theta_1-h_1,  x_2=\theta_2+h_2\}$.
\item
Points on the boundary of $C_3$:
\begin{itemize}
\item
When $x\in\{x: x_1\geq\theta_1-h_1, x_2=\theta_2-h_2\},$ $f_3(x):=\left[\begin{array}{cc}-\gamma_1x_1\\k_2-\gamma_2(\theta_2-h_2)\end{array}\right]$. Since $\theta_1+h_1<\frac{k_1}{\gamma_1}<\theta^{\max}_{1}$, $\theta_2+h_2<\frac{k_2}{\gamma_2}<\theta_2^{\max}$, then $f_3(x)$ points inside of $C_3$.
\item
When $x\in\{x: x_1=\theta_1-h_1, x_2\geq\theta_2-h_2\},$ $f_3(x):=\left[\begin{array}{cc}-\gamma_1(\theta_1-h_1)\\k_2-\gamma_2x_2\end{array}\right]$. Since $\theta_1+h_1<\frac{k_1}{\gamma_1}<\theta^{\max}_{1}$, $\theta_2+h_2<\frac{k_2}{\gamma_2}<\theta_2^{\max}$, then $f_3(x)$ points outside of $C_3$.
\end{itemize}
Then, since there is no isolated equilibrium point in $C_3$, the trajectories starting in $C_3$ have $x$ component that flow towards $\{x_1=\theta_1-h_1,  x_2\geq\theta_2-h_2\}$.
\item
Points on the boundary of $C_4$:
\begin{itemize}
\item
When $x\in\{x: x_1=\theta_1+h_1, x_2\geq\theta_2-h_2\},$ $f_4(x):=\left[\begin{array}{cc}-\gamma_1(\theta_1+h_1)\\-\gamma_2x_2\end{array}\right]$. Since $\gamma_1>0$, $\gamma_2>0$, then $f_4(x)$ points inside of $C_4.$ 
\item
When $x\in\{x: x_1=0, x_2\geq\theta_2-h_2\},$ $f_4(x):=\left[\begin{array}{cc}0\\-\gamma_2x_2\end{array}\right]$. Since $\gamma_2>0$, then $f_4(x)$ is tangent to the boundary of $C_4.$ 
\item
When $x\in\{x: 0\leq x_1\leq\theta_1+h_1, x_2=\theta_2-h_2\},$ $f_4(x):=\left[\begin{array}{cc}-\gamma_1x_1\\-\gamma_2(\theta_2-h_2)\end{array}\right]$. Since $\gamma_1>0,$$\gamma_2>0$, then $f_4(x)$ points outside of the $C_4.$ 
\end{itemize}
Then, since there is no isolated equilibrium point in $C_4$, the trajectories starting in $C_4$ have $x$ component that flow towards $\{0\leq x_1\leq\theta_1+h_1,  x_2=\theta_2-h_2\}$.
\end{enumerate}

Combining the above arguments,  Figure \ref{fig:T} shows the transition sequence of $q\in Q$
for case 5 in Table \ref{tab:EqPoints}.

Now, we compute the value of the trajectories as they transition 
according to the said sequence.

The differential equation for the $x$ components of the continuous dynamics of $\cal H$ can be evaluated for each possible value of $q$ and written as
    \begin{equation}
\label{eqn:Sol}
\dot{x}=K(q)-\Gamma x, 
\end{equation}
where $\Gamma=\left[\begin{array}{cc}\gamma_1&0\\0&\gamma_2\end{array}\right]$ and $K:Q\rightarrow\mathbb{R}^{2\times 1}$
is given by
$$K(q)=\left\{\begin{array}{cccc}\left[\begin{array}{cc}k_1\\0\end{array}\right]\quad \mbox{ if}\quad q=(0, 0),\\\left[\begin{array}{cc}0\\0\end{array}\right]\quad \mbox{ if}\quad q=(0, 1),\\\left[\begin{array}{cc}k_1\\k_2\end{array}\right]\quad \mbox{ if}\quad q=(1, 0),\\\left[\begin{array}{cc}0\\k_2\end{array}\right]\quad \mbox{ if}\quad q=(1, 1). \end{array}\right.$$ 
Restricted to $C$, \eqref{eqn:Sol} is a linear time-invariant system. For any initial condition $z(0, 0) = [x(0,0)^\top\ q(0,0)^\top]^\top \in C$, 
the unique solution to \eqref{eqn:Sol} for each $t \geq 0$, up to the first jump, is given by
\begin{equation}
\label{eqn:X}
x(t,0)=u(q(0,0))+\exp(-\Gamma t)(x(0, 0)-u(q(0,0))),
\end{equation}
where $u(q)=\Gamma^{-1}K(q)$ for each $q \in Q$.

\begin{figure}[h]
\begin{center}
\psfrag{a}[][][0.9]{$01$}
\psfrag{b}[][][0.9]{$00$}
\psfrag{c}[][][0.9]{$10$}
\psfrag{d}[][][0.9]{$11$}
\includegraphics[width=0.25\columnwidth]{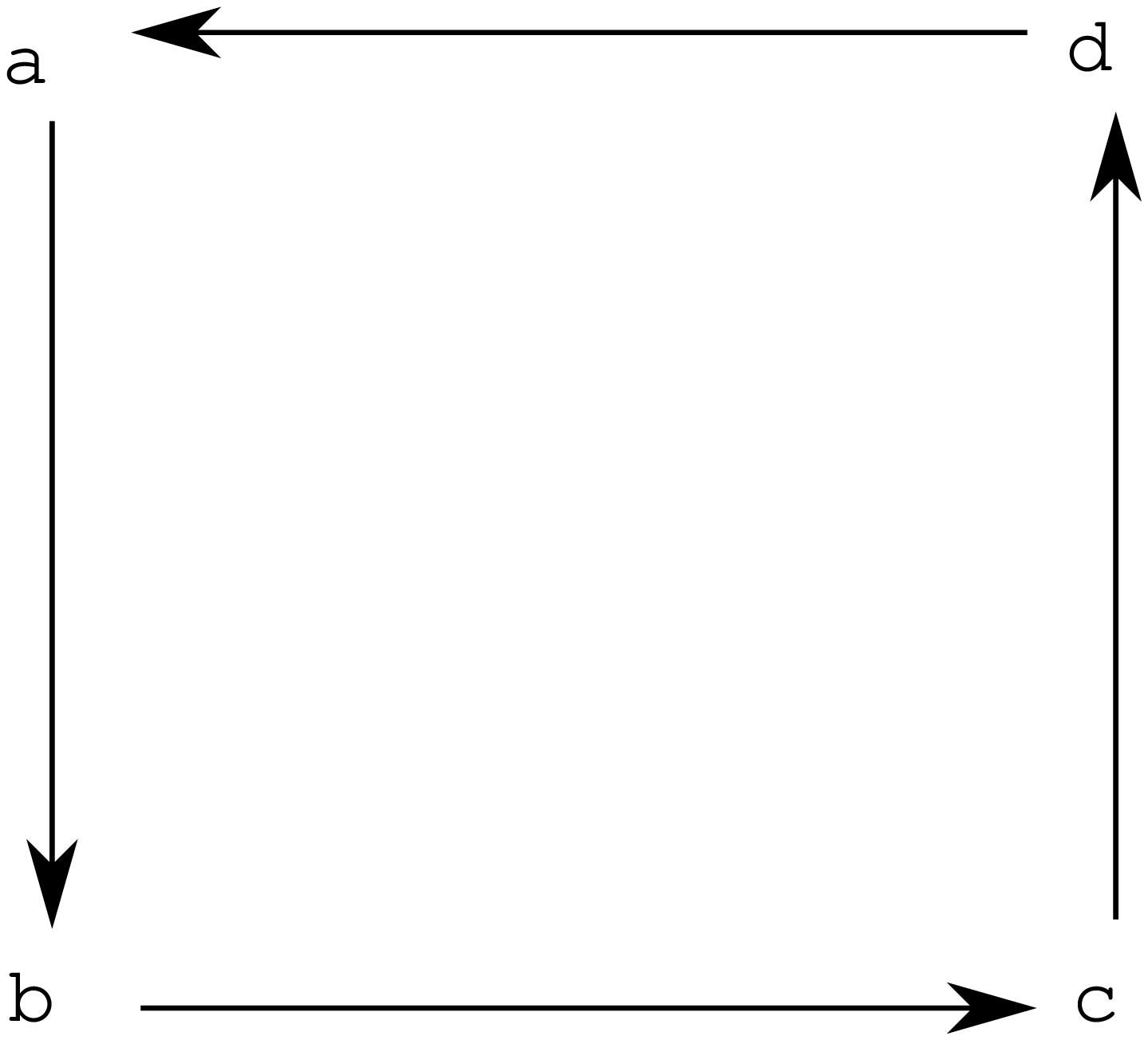}
\end{center}
\caption{\label{fig:T}\emph{Transition graph of $q\in Q$.}}
\end{figure}

For the initial value of $q$ given by $q^0=\left[\begin{array}{c}0\\0\end{array}\right]$ and the initial value of $x$ given by 
$$x(0, 0)=p_0=\left[\begin{array}{cc}p_0(1)\\\theta_2-h_2\end{array}\right],$$
the solution to \eqref{eqn:Sol} is given by
$$x(t, 0)=\left[\begin{array}{ll}
u_1-\left(u_1-x_1(0, 0)\right)\exp(-\gamma_1t)\\
x_2(0, 0)\exp(-\gamma_2t)\end{array}\right]$$ 
where $u_1=\frac{k_1}{\gamma_1}.$ Since $K((0,0))=\left[\begin{array}{cc}k_1\\0\end{array}\right]$ implies that $u(q)=\left[\begin{array}{cc}\frac{k_1}{\gamma_1}\\0\end{array}\right].$ Note that $x_2(t, 0)$ decreases to zero in $C_1$ and that $x_1(t, 0)$ is increasing in $C_1,$ reaching the threshold value $x_1=\theta_1+h_1$ at $t_1'=\ln\left[{\frac{u_1-x_1(0, 0)}{u_1-(\theta_1+h_1)}}\right]^{\frac{1}{\gamma_1}}$. A jump of $q$ to $\left[\begin{array}{cc}1\\0\end{array}\right]$ occurs at $t=t_1, j=0.$

After the jump, the initial value of $x$ is $p_1$, where 
\begin{equation}
p_1=\left[\begin{array}{cc}
\theta_1+h_1\\
p_0(2)\left(\frac{u_1-(\theta_1+h_1)}{u_1-p_0(1)}\right)^{\frac{\gamma_2}{\gamma_1}}\end{array}\right].
\end{equation}
Proceeding similarly as when the initial state was $p_0,$ we obtain the following expressions for $p_2$, $p_3$, and $p_4:$ 
\begin{equation}
\label{eqn:q2q4}
\begin{array}{lll}
p_2=\left[\begin{array}{cc}u_1-(u_1-p_1(1))\left(\frac{u_2-p_2(2)}{u_2-p_1(2)}\right)^{\frac{\gamma_1}{\gamma_2}}\\\theta_2+h_2\end{array}\right],\\
p_3=\left[\begin{array}{cc}\theta_1-h_1\\
u_2-(u_2-p_2(2))\left(\frac{p_3(1)}{p_2(1)}\right)^{\frac{\gamma_2}{\gamma_1}}\end{array}\right], \ \
p_4=\left[\begin{array}{cc}p_3(1)\left(\frac{p_4(2)}{p_3(2)}\right)^{\frac{\gamma_1}{\gamma_2}}\\\theta_2-h_2\end{array}\right],\end{array}
\end{equation}
where $u_2=\frac{k_2}{\gamma_2}.$ 
Also, similarly, we have the expression for $t_2', t_3', t_4':$
$$\begin{array}{lll}t_2'=\ln\left[{\frac{u_2-p_1(2)}{u_2-(\theta_2+h_2)}}\right]^{\frac{1}{\gamma_2}},\quad t_3'=\ln\left[{\frac{p_2(1)}{\theta_1-h_1}}\right]^{\frac{1}{\gamma_1}}, \quad t_4'=\ln\left[{\frac{p_3(2)}{\theta_2-h_2}}\right]^{\frac{1}{\gamma_2}}.\end{array}$$
Then, the period of the limit cycle is given by
$T=t_1'+t_2'+t_3'+t_4'.$
Note that $t_1=t_1',$ $t_2=t_1'+t_2',$ $t_3=t_1'+t_2'+t_3'$ and $t_4=t_1'+t_2'+t_3'+t_4',$ 
where $t_1, t_2, t_3$ and $t_4$ define the jump times $(t_1, 0),$ $(t_2, 1),$ $(t_3, 2),$ and $(t_4, 3).$

Now, we define the map $\rho: [0, \theta_1^{\max}]\rightarrow\mathbb{R}$ as $$\rho(r)=\rho_4\circ \rho_3\circ \rho_2\circ \rho_1(r),$$ where$$\begin{array}{llll}
\rho_1(r)=(\theta_2-h_2)\left(\frac{u_1-(\theta_1+h_1)}{u_1-r}\right)^{\frac{\gamma_2}{\gamma_1}},\\
\rho_2(r)=u_1-(u_1-(\theta_1+h_1))\left(\frac{u_2-(\theta_2+h_2)}{u_2-r}\right)^{\frac{\gamma_1}{\gamma_2}},\\
\rho_3(r)=u_2-(u_2-(\theta_2+h_2))\left(\frac{\theta_1-h_1}{r}\right)^{\frac{\gamma_2}{\gamma_1}},\\
\rho_4(r)=(\theta_1-h_1)\left(\frac{\theta_2-h_2}{r}\right)^{\frac{\gamma_1}{\gamma_2}}.\end{array}$$
Then, $r$ such that
$$\rho(r)=r$$
defines $p_0(1)$. 

Then, combining the above expressions, 
we obtain \eqref{eqn:p0}-\eqref{eqn:p3}.
Finally, using \eqref{eqn:X}, the set $S$ is constructed by combining 
the $x$ components of the (unique) solutions between these points. 
Since each piece of the $x$ component 
corresponding to a constant value of $q$ is a solution to a linear system, 
this set of points has the property that, from every point in it, the only existing solution
from that point stays in the set, i.e., 
the set is strongly forward invariant.

\subsection{Proof of Proposition~\ref{eqn:Existence}}
\label{app:ExistenceProof}

To verify the sufficient conditions for the existence of nontrivial solutions from an initial point in $C\cup D,$ it is enough to show that $F(z)\in T_C(z)$ for every $z\in C\setminus D$ 
in the boundary of $C$ (the {\bf (VC)} condition holds for every point in the interior of $C$.)

Next, we consider each possible case. 

\begin{enumerate}
\item 
 Let $z\in C_1\setminus D_1 = \{z: q_1=0, q_2=0, 0\leq x_1<\theta_1+h_1 , 0\leq x_2\leq\theta_2+h_2\}.$ Let
\begin{eqnarray*}
T^1_{C_{1}}(z) & =& \{w\in\mathbb{R}^2: w_1\geq 0, w_2\leq 0\},\\
T^2_{C_1}(z) &=&\{w\in\mathbb{R}^2: w_1\geq 0\},\\
T^3_{C_{1}}(z) &=& \{ w\in\mathbb{R}^2:  w_1\geq 0, w_2\geq 0\},\\
T^4_{C_1}(z) &=& \{w\in\mathbb{R}^2:  w_2\geq 0\},\\
T^5_{C_{1}}(z) &=&\{ w\in\mathbb{R}^2: w_1\leq 0, w_2\geq 0\},\\
T^6_{C_{1}}(z) &=&\{ w\in\mathbb{R}^2: w_1\leq 0\},\\
T^7_{C_{1}}(z) &=&\{ w\in\mathbb{R}^2: w_1\leq 0, w_2\leq 0\},\\
T^8_{C_{1}}(z) &=&\{ w\in\mathbb{R}^2: w_2\leq 0\}.
\end{eqnarray*}
Then, the tangent cone of $C_1$
at points $z = (x,q)$ is given as follows:
\begin{itemize}
\item For $x\in\{x: x_1=0, x_2=\theta_2+h_2\},$ $T_{C_{1}}(z)=T^1_{C_{1}}(z)$,\\
\item For $x\in\{x: x_1=0, 0<x_2<\theta_2+h_2\},$ $T_{C_{1}}(z)=T^2_{C_{1}}(z)$.\\
\item For $x\in\{x: x_1=0, x_2=0\},$ $T_{C_{1}}(z)=T^3_{C_{1}}(z)$.\\
\item For $x\in\{0< x_1<\theta_1+h_1, x_2=0\},$ $T_{C_{1}}(z)=T^4_{C_{1}}(z)$,\\
\item For $x\in\{x: x_1=\theta_1+h_1, x_2=0\},$ $T_{C_{1}}(z)=T^5_{C_{1}}(z)$.\\
\item For $x\in\{x_1=\theta_1+h_1, 0< x_2<\theta_2+h_2\},$ $T_{C_{1}}(z)=T^6_{C_{1}}(z)$,\\
\item For $x\in\{x_1=\theta_1+h_1, x_2=\theta_2+h_2\},$ $T_{C_{1}}(z)=T^7_{C_{1}}(z)$,\\
\item For $x\in\{0< x_1<\theta_1+h_1, x_2=\theta_2+h_2\},$ $T_{C_{1}}(z)=T^8_{C_{1}}(z)$.\\
\end{itemize}

\begin{figure}[h]
\begin{center}
\psfrag{C1}[][][0.9]{$C_1$}
\psfrag{C2}[][][0.9]{$C_2$}
\psfrag{C3}[][][0.9]{$C_3$}
\psfrag{C4}[][][0.9]{$C_4$}
\psfrag{q1}[][][0.9]{$q_1=0$}
\psfrag{q2}[][][0.9]{$q_2=0$}
\psfrag{(a)}[][][0.9]{$(a)$}
\psfrag{1}[][][0.9]{$1$}
\psfrag{2}[][][0.9]{$2$}
\psfrag{3}[][][0.9]{$3$}
\psfrag{4}[][][0.9]{$4$}
\psfrag{5}[][][0.9]{$5$}
\psfrag{6}[][][0.9]{$6$}
\psfrag{7}[][][0.9]{$7$}
\psfrag{8}[][][0.9]{$8$}
\psfrag{x1}[][][0.9]{$x_1$}
\psfrag{x2}[][][0.9]{$x_2$}
\psfrag{t1m}[][][0.9]{$\theta_1-h_1$}
\psfrag{t1p}[][][0.9]{$\theta_1+h_1$}
\psfrag{t1max}[][][0.9]{$\theta_1^{\max}$}
\psfrag{t2m}[][][0.9]{$\theta_2-h_2$}
\psfrag{t2p}[][][0.9]{$\theta_2+h_2$}
\psfrag{t2max}[][][0.9]{$\theta_2^{\max}$}
\psfrag{(a)}[][][0.9]{}
\includegraphics[width=0.55\columnwidth]{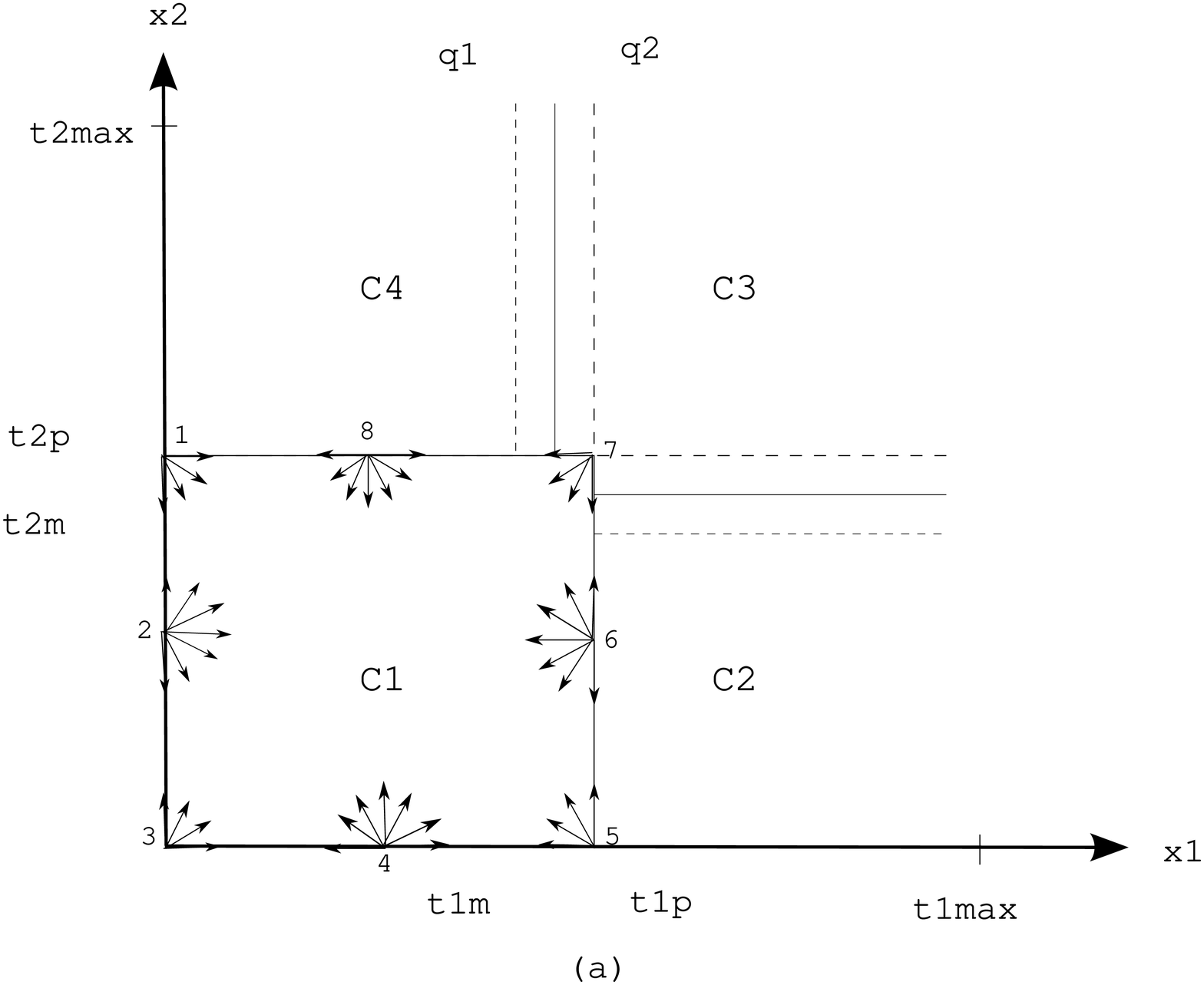}
\end{center}
\caption{\emph{Tangent cones on the boundaries of $C_1$.}\label{fig:TangentConeC_1}}
\end{figure}
Now, we check the vector field $F$ on the boundary of $C_1$ 
away from $D_1$. 

\begin{itemize}
\item
When $x\in\{x: x_1=0, 0\leq x_2\leq\theta_2+h_2\},$ $f_1(x):=\left[\begin{array}{cc}k_1\\-\gamma_2x_2\end{array}\right]$. Since $k_1>0$, $f_1(x)$ points inside of $C_1$.
\item
When $x\in\{x: 0\leq x_1\leq\theta_1+h_1, x_2=0\},$ $f_1(x):=\left[\begin{array}{cc}k_1-\gamma_1x_1\\0\end{array}\right]$. Then,
$f_1(x)$ is tangent to the boundary of $C_1.$ 
\end{itemize}
Then, $F(z)\in T_{C_1}$ holds, implying that {\bf (VC)} holds at each point $z\in C_1\setminus D_1$. 
\item
Let $z\in C_2\setminus D_2 = \{z: q_1=1, q_2=0, x_1 \geq \theta_1-h_1, 0\leq x_2<\theta_2+h_2\}.$ Let
\begin{eqnarray*}
T^1_{C_{2}}(z) & =& \{w\in\mathbb{R}^2:  w_2\geq 0\},\\
T^2_{C_2}(z) &=&\{w\in\mathbb{R}^2: w_1\geq 0, w_2\geq 0\},\\
T^3_{C_{2}}(z) &=& \{ w\in\mathbb{R}^2:  w_1\geq 0\},\\
T^4_{C_2}(z) &=& \{w\in\mathbb{R}^2: w_1\leq 0, w_2\geq 0\},\\
T^5_{C_{2}}(z) &=&\{ w\in\mathbb{R}^2:  w_2\leq 0\}.
\end{eqnarray*}
Then, the tangent cone of $C_2$ is given by as follows:
\begin{itemize}
\item For $x\in\{x: x_1>\theta_1-h_1, x_2=0\},$ $T_{C_{2}}(z)=T^1_{C_{2}}(z)$,\\
\item For $x\in\{x: x_1=\theta_1-h_1, x_2=0\},$ $T_{C_{2}}(z)=T^2_{C_{2}}(z)$.\\
\item For $x\in\{x: x_1=\theta_1-h_1, 0\leq x_2\leq\theta_2+h_2\},$ $T_{C_{2}}(z)=T^3_{C_{2}}(z)$.\\
\item For $x\in\{x_1=\theta_1-h_1, x_2=\theta_2+h_2\},$ $T_{C_{2}}(z)=T^4_{C_{2}}(z)$,\\
\item For $x\in\{x: x_1>\theta_1-h_1, x_2=\theta_2+h_2\},$ $T_{C_{2}}(z)=T^5_{C_{2}}(z)$.\\
\end{itemize}

Figure~\ref{fig:TangentConeC2} depicts the tangent cones on the boundaries of $C_2$.

\begin{figure}[h]
\begin{center}
\psfrag{q1}[][][0.9]{$q_1=1$}
\psfrag{q2}[][][0.9]{$q_2=0$}
\psfrag{C2}[][][0.9]{$C_2$}
\psfrag{x1}[][][0.9]{$x_1$}
\psfrag{x2}[][][0.9]{$x_2$}
\psfrag{t1m}[][][0.9]{$\theta_1\!\!-\!h_1$}
\psfrag{t1p}[][][0.9]{$\theta_1\!\!+\!h_1$}
\psfrag{t1max}[][][0.9]{$\theta_1^{\max}$}
\psfrag{t2m}[][][0.9]{$\theta_2-h_2$}
\psfrag{t2p}[][][0.9]{$\theta_2+h_2$}
\psfrag{t2max}[][][0.9]{$\theta_2^{\max}$}
\psfrag{1}[][][0.9]{$1$}
\psfrag{2}[][][0.9]{$2$}
\psfrag{3}[][][0.9]{$3$}
\psfrag{4}[][][0.9]{$4$}
\psfrag{5}[][][0.9]{$5$}

\includegraphics[width=0.55\columnwidth]{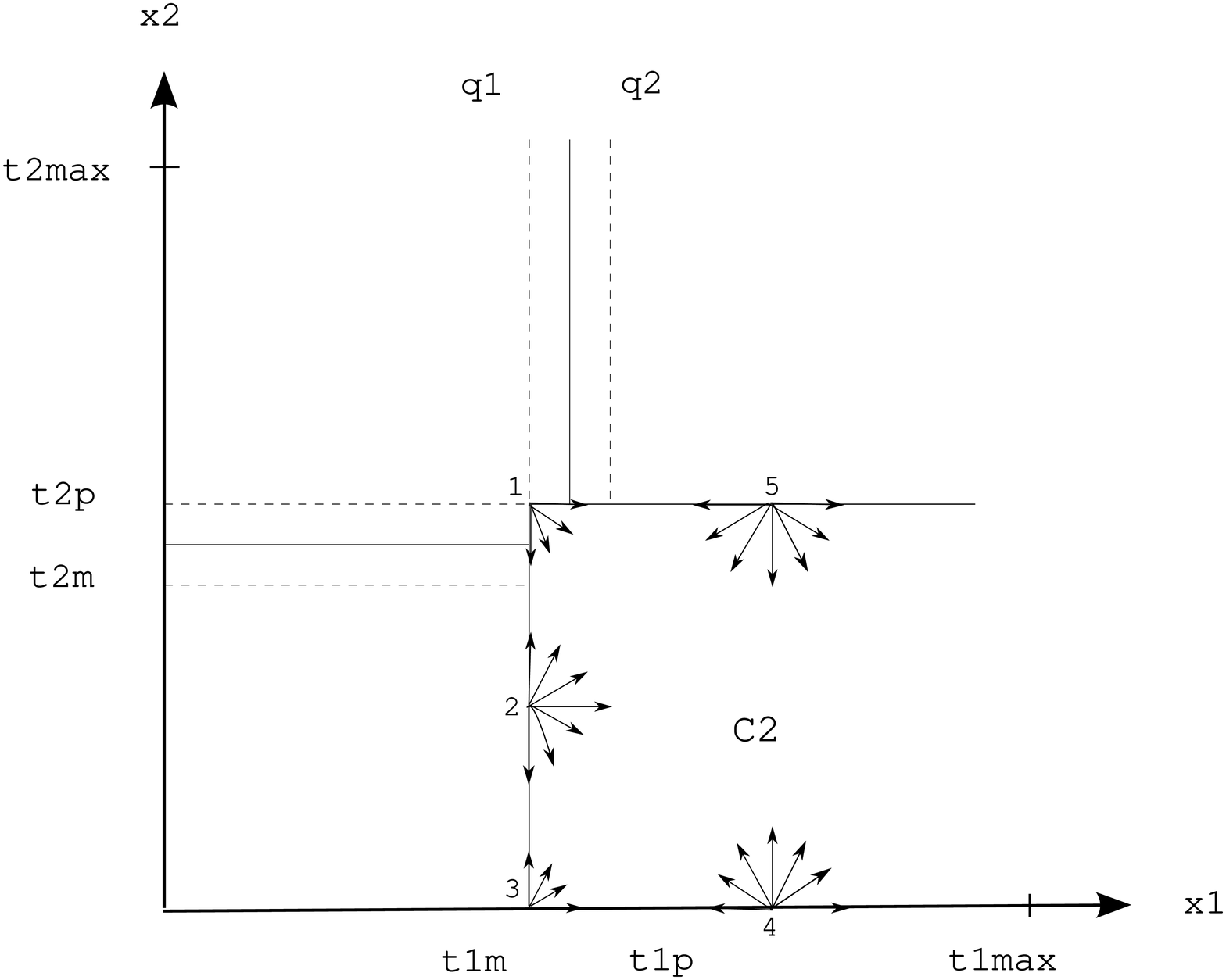}
\end{center}
\caption{\emph{The tangent cone of $C_{2}$ when $\theta_1+h_1<\frac{k_1}{\gamma_1}<\theta^{\max}_{ 1 }$,$0<\frac{k_2}{\gamma_2}<\theta_2+h_2$.}\label{fig:TangentConeC2}}
\end{figure}

Now, we check the vector field $F$ on the boundary of $C_2$ away from $D_2$.
\begin{itemize}
\item
When $x\in\{x: x_1\geq\theta_1-h_1, x_2=0\},$ $f_2(x):=\left[\begin{array}{cc}k_1-\gamma_1x_1\\k_2\end{array}\right]$. 
Since $k_2 > 0$,  we have that $f_2(x)$ points inside of $C_2$.
\end{itemize}
Then, $F(z)\in T_{C_2}(z)$ holds, implying that {\bf (VC)} holds at each point $z\in C_2\setminus D_2$.
\item
Let $z\in C_3\setminus D_3= \{z: q_1=1, q_2=1, x_1>\theta_1-h_1 ,  x_2>\theta_2-h_2\}$. 

Since there are no points in the boundary of $C_3$ that are not in $D_3$, {\bf (VC)} holds for free.

\item
 Let $z\in C_4\setminus D_4 = \{z: q_1=0, q_2=1, 0\leq x_1\leq\theta_1+h_1 ,  x_2>\theta_2-h_2\}.$ Let
\begin{eqnarray*}
T^1_{C_{4}}(z) & =& \{w\in\mathbb{R}^2: w_1\leq 0\},\\
T^2_{C_4}(z) &=&\{w\in\mathbb{R}^2: w_1\leq 0, w_2\geq 0\},\\
T^3_{C_{4}}(z) &=& \{ w\in\mathbb{R}^2: w_2\geq 0\},\\
T^4_{C_4}(z) &=& \{w\in\mathbb{R}^2:  w_1\geq 0, w_2\geq 0\},\\
T^5_{C_{4}}(z) &=&\{ w\in\mathbb{R}^2: w_1\geq 0\}.
\end{eqnarray*}
Then, the tangent cone of $C_4$ is given by as follows:
\begin{itemize}
\item For $x\in\{x: x_1=\theta_1+h_1, x_2>\theta_2-h_2\},$ $T_{C_{4}}(z)=T^1_{C_{4}}(z)$,\\
\item For $x\in\{x: x_1=\theta_1+h_1, x_2=\theta_2-h_2\},$ $T_{C_{4}}(z)=T^2_{C_{4}}(z)$.\\
\item For $x\in\{x: 0<x_1<\theta_1+h_1, x_2=\theta_2-h_2\},$ $T_{C_{4}}(z)=T^3_{C_{4}}(z)$.\\
\item For $x\in\{x_1=0, x_2=\theta_2-h_2\},$ $T_{C_{4}}(z)=T^4_{C_{4}}(z)$,\\
\item For $x\in\{x: x_1=0, x_2>\theta_2-h_2\},$ $T_{C_{4}}(z)=T^5_{C_{4}}(z)$.\\
\end{itemize}

\begin{figure}[h]
\begin{center}
\psfrag{C4}[][][0.9]{$C_4$}
\psfrag{q1}[][][0.9]{$q_1\!\!=\!0$}
\psfrag{q2}[][][0.9]{$q_2\!\!=\!1$}
\psfrag{(a)}[][][0.9]{$(b)$}
\psfrag{1}[][][0.9]{$1$}
\psfrag{2}[][][0.9]{$2$}
\psfrag{3}[][][0.9]{$3$}
\psfrag{4}[][][0.9]{$4$}
\psfrag{5}[][][0.9]{$5$}
\psfrag{x1}[][][0.9]{$x_1$}
\psfrag{x2}[][][0.9]{$x_2$}
\psfrag{t1m}[][][0.9]{$\theta_1\!-\!h_1$}
\psfrag{t1p}[][][0.9]{$\theta_1\!+\!h_1$}
\psfrag{t1max}[][][0.9]{$\theta_1^{\max}$}
\psfrag{t2m}[][][0.9]{$\theta_2\!-\!h_2$}
\psfrag{t2p}[][][0.9]{$\theta_2\!+\!h_2$}
\psfrag{t2max}[][][0.9]{$\theta_2^{\max}$}
\psfrag{(b)}[][][0.9]{}
\includegraphics[width=0.55\columnwidth]{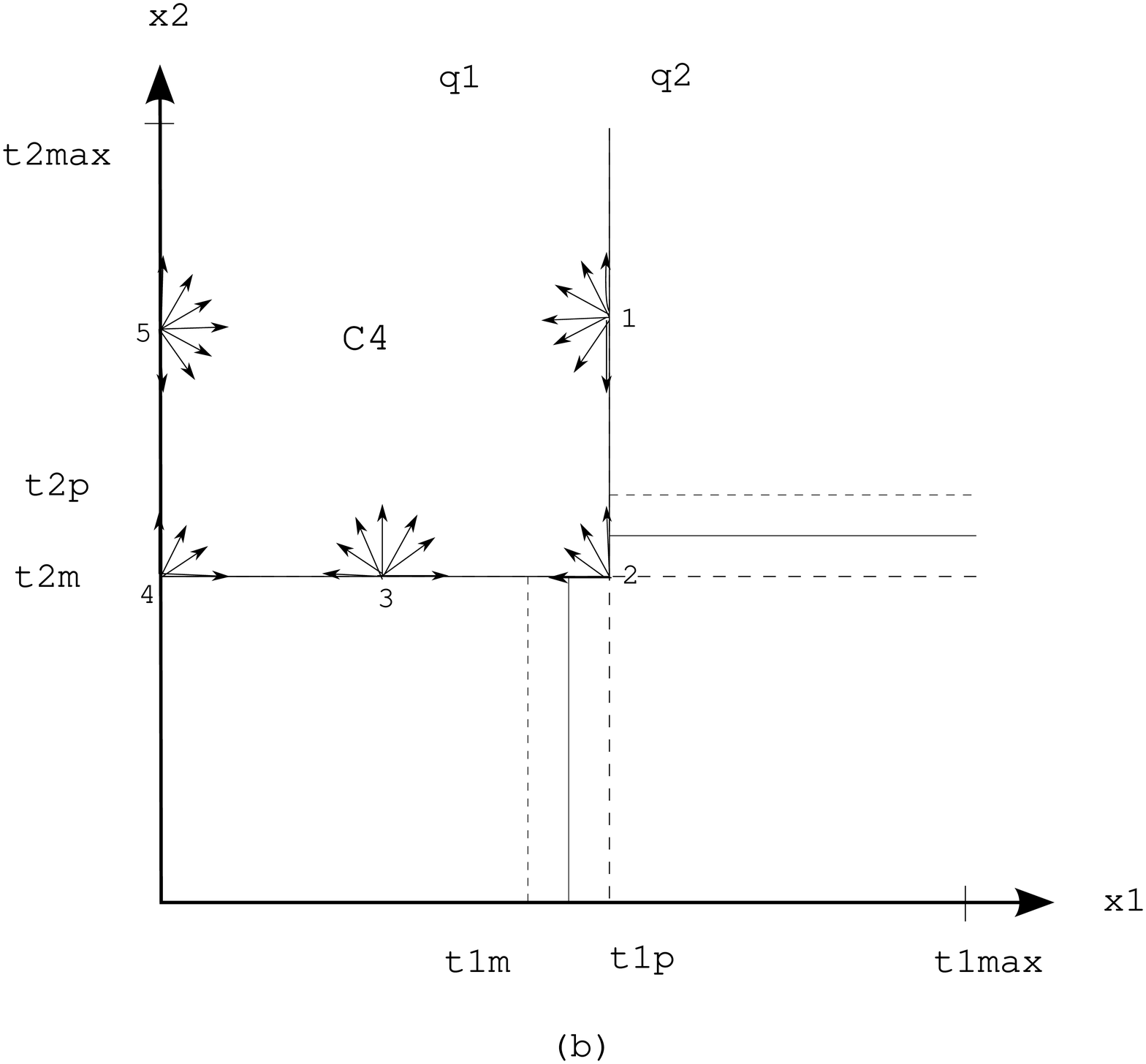}
\end{center}
\caption{\emph{Tangent cones on the boundaries of $C_4$.}\label{fig:TangentConeC_4}}
\end{figure}

Now, we check the vector field $F$ on the boundary of $C_4$ away from $D_4$. 
\begin{itemize}
\item
When $x\in\{x: x_1=0, x_2\geq\theta_2-h_2\},$ $f_4(x):=\left[\begin{array}{cc}0\\-\gamma_2x_2\end{array}\right]$, and $f_4(x)$ is tangent to the boundary of $C_4.$
Then, $F(z)\in T_{C_4}(z)$ holds at each point $z\in C_4\setminus D_4$.
\end{itemize}
\end{enumerate}

Combining the above arguments, for each case in Table \ref{tab:EqPoints},
{\bf (VC)} holds and nontrivial solutions to $\cal H$ in \eqref{eqn:H2} exist.

Since $f_i$ is linear, every solution to $\dot{z}=F(z)$ subject to $z\in C$ does not escape to infinity by flowing. Then, condition 2) below (VC) does not hold. Since $G(D)\subset C\cup D$, then condition 3) therein does not hold either.

Finally, every solution is not Zeno, since at most after the second jump,
every solution needs to flow (linearly) from the 
value after the jump, which is given by $G$,
for at least
$2\min\{h_1,h_2\}$ in the $x_1$ or in the $x_2$ direction (for certain initial conditions, 
e.g., $z(0,0) = [\theta_1 - h_1, \theta_2 - h_2, 1,1]$,
solutions jump twice consecutively, and after that, flow for the said amount).

\subsection{Proof of Proposition~\ref{eqn:Stable1}}
\label{app:Stable1Proof}

Note that when $\theta_1+h_1<\frac{k_1}{\gamma_1}<\theta_1^{max}$,$0<\frac{k_2}{\gamma_2}<\theta_2+h_2$,   $q_1=1$, $q_2=0,$
\\
$$\begin{array}{lll}
0=\dot{x}_1 =k_1-\gamma_1x_1\quad\Rightarrow\quad k_1-\gamma_1x_1=0\quad \Rightarrow \quad x_1^*=\frac{k_1}{\gamma_1}\\
0=\dot{x}_2 = k_2-\gamma_2x_2\quad\Rightarrow\quad k_2-\gamma_2x_2=0\quad \Rightarrow \quad x_2^*=\frac{k_2}{\gamma_2},
\end{array}$$
from where we have $z^*_1 = [{x^*}^\top\ 1\ 0]^\top\in C_2$. 
Now, change to $e$ coordinates given by\\
$$e_1 = x_1-x_1^*, \quad e_2 =x_2-x_2^*.$$
We have that
\begin{eqnarray*}
\dot{e}_1 & = & \dot{x}_1-\dot{x}_1^*= k_1-\gamma_1x_1-0=k_1-\gamma_1(e_1+x_1^*)\\
& = & k_1-\gamma_1e_1-k_1=-\gamma_1e_1,\\
\dot{e}_2 & = & \dot{x}_2-\dot{x}_2^*=k_2-\gamma_2x_2-0=k_2-\gamma_2(e_2+x_2^*)\\
& = & k_2-\gamma_2e_2-k_2=-\gamma_2e_2.
\end{eqnarray*}
Then, the $x$ component of $\dot{z}=F(z)$ leads to 
$$
\dot{e}=\left[
\begin{array}{cccc}
-\gamma_1 &0\\
0& -\gamma_2.
\end{array}\right]e.
$$
Then, since $\gamma_1$ and $\gamma_2$ are positive, 
we have that $z^*_1$ is (exponentially) stable -- this property 
can be easily certified with the Lyapunov function
$V(e)=e^\top e$.

Now we check the vector fields of boundaries of $C_2$ (see Figure \ref{fig:c2forwardinvariant}).

\begin{figure}[h]
\begin{center}
\psfrag{C1}[][][0.9]{$C_4$}
\psfrag{C2}[][][0.9]{$C_4$}
\psfrag{C3}[][][0.9]{$C_4$}
\psfrag{C4}[][][0.9]{$C_4$}
\psfrag{q1}[][][0.9]{$q_1\!\!=\!0$}
\psfrag{q2}[][][0.9]{$q_2\!\!=\!1$}
\psfrag{(a)}[][][0.9]{$(b)$}
\psfrag{1}[][][0.9]{$1$}
\psfrag{2}[][][0.9]{$2$}
\psfrag{3}[][][0.9]{$3$}
\psfrag{4}[][][0.9]{$4$}
\psfrag{5}[][][0.9]{$5$}
\psfrag{x1}[][][0.9]{$x_1$}
\psfrag{x2}[][][0.9]{$x_2$}
\psfrag{t1m}[][][0.9]{$\theta_1\!-\!h_1$}
\psfrag{t1p}[][][0.9]{$\theta_1\!+\!h_1$}
\psfrag{t1max}[][][0.9]{$\theta_1^{\max}$}
\psfrag{t2m}[][][0.9]{$\theta_2\!-\!h_2$}
\psfrag{t2p}[][][0.9]{$\theta_2\!+\!h_2$}
\psfrag{t2max}[][][0.9]{$\theta_2^{\max}$}
\includegraphics[width=0.55\columnwidth]{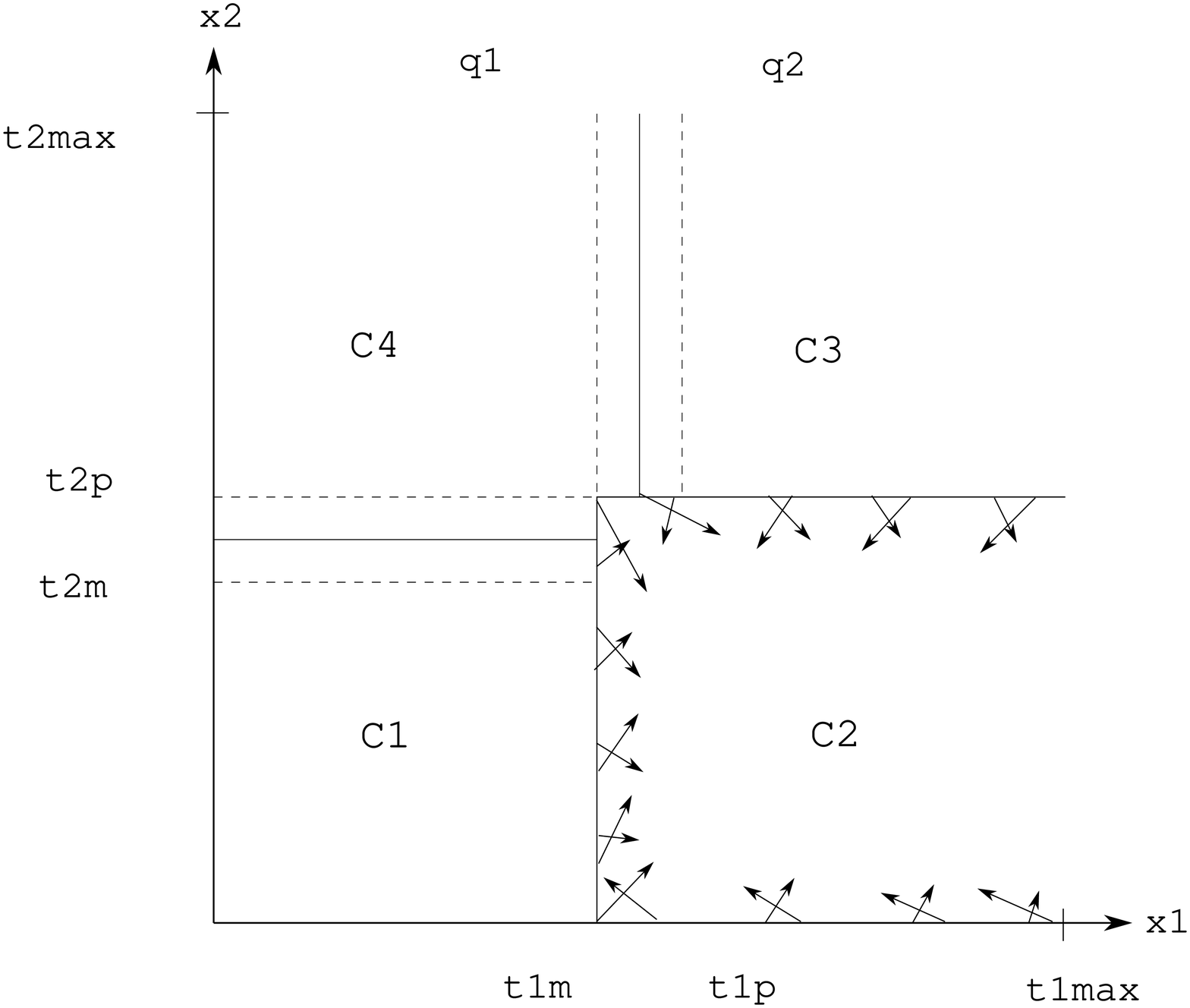}
\end{center}
\caption{\emph{The vector fields on the boundaries of $C_2$ when $\theta_1+h_1<\frac{k_1}{\gamma_1}<\theta_1^{max}$, $0<\frac{k_2}{\gamma_2}<\theta_2+h_2$.}\label{fig:c2forwardinvariant}}
\end{figure}

\begin{itemize}
\item
When $x\in\{x: x_1=\theta_1-h_1, 0\leq x_2\leq\theta_2+h_2\}$, $f_{2}(x)=\left[\begin{array}{cc}k_1-\gamma_1(\theta_1-h_1)\\k_2-\gamma_2x_2\end{array}\right]$. Since $\theta_1+h_1<\frac{k_1}{\gamma_1}<\theta_1^{max}$, $0<\frac{k_2}{\gamma_2}<\theta_2+h_2$, we have that $f_2(x)$ points inside of $C_2$. 
\item
When $x\in\{x: x_1\geq\theta_1-h_1, x_2=\theta_2+h_2\}$, $f_2(x)=\left[\begin{array}{cc}k_1-\gamma_1x_1\\k_2-\gamma_2(\theta_2+h_2)\end{array}\right]$. Since $\theta_1+h_1<\frac{k_1}{\gamma_1}<\theta_1^{max}$, $0<\frac{k_2}{\gamma_2}<\theta_2+h_2$, we have that $f_2(x)$ points inside of $C_2$. 
\item
When $x\in\{x: x_1\geq\theta_1-h_1, x_2=0\}$, $f_2(x)=\left[\begin{array}{cc}k_1-\gamma_1x_1\\k_2\end{array}\right]$. Since $\theta_1+h_1<\frac{k_1}{\gamma_1}<\theta_1^{max}$, $k_2>0$, we have that $f_2(x)$ points inside of $C_2$. 
\end{itemize}
Then, once a trajectory enters or starts in $C_2,$ it will stay and never leave $C_2$, i.e., 
the set $C_2$ is forward invariant.
As a consequence,
since the equilibrium point $z_1^*$ belongs to $C_2$, 
every trajectory starting from or reaching $C_2$ converges to $z_1^*$.

Global asymptotic stability follows since for every initial condition $z(0,0)\in(C\cup D)\backslash C_2$, solutions reach $C_2$ in finite time.  
To establish this property, we check the vector field of the system on the boundary of each set $C_i$, for each $i\in\{1, 3, 4\}.$
\begin{figure}[h!]
\centering
\psfrag{C1}[][][0.9]{$C_1$}
\psfrag{C2}[][][0.9]{$C_2$}
\psfrag{C3}[][][0.9]{$C_3$}
\psfrag{C4}[][][0.9]{$C_4$}
\psfrag{q1}[][][0.9]{$q_1=0$}
\psfrag{q2}[][][0.9]{$q_2=0$}
\psfrag{q3}[][][0.9]{$q_2=1$}
\psfrag{(a)}[][][0.9]{$(a)$}
\psfrag{t1m}[][][0.9]{$\theta_1\!\!-\!h_1$}
\psfrag{t1p}[][][0.9]{$\theta_1\!\!+\!h_1$}
\psfrag{t1max}[][][0.9]{$\theta_1^{\max}$}
\psfrag{x1}[][][0.9]{$x_1$}
\psfrag{x2}[][][0.9]{$x_2$}
\psfrag{(b)}[][][0.9]{$(b)$}
\psfrag{t2m}[][][0.8]{$\theta_2-h_2$}
\psfrag{t2p}[][][0.9]{$\theta_2+h_2$}
\psfrag{t2max}[][][0.9]{$\theta_2^{\max}$}
\includegraphics[width=0.85\columnwidth]{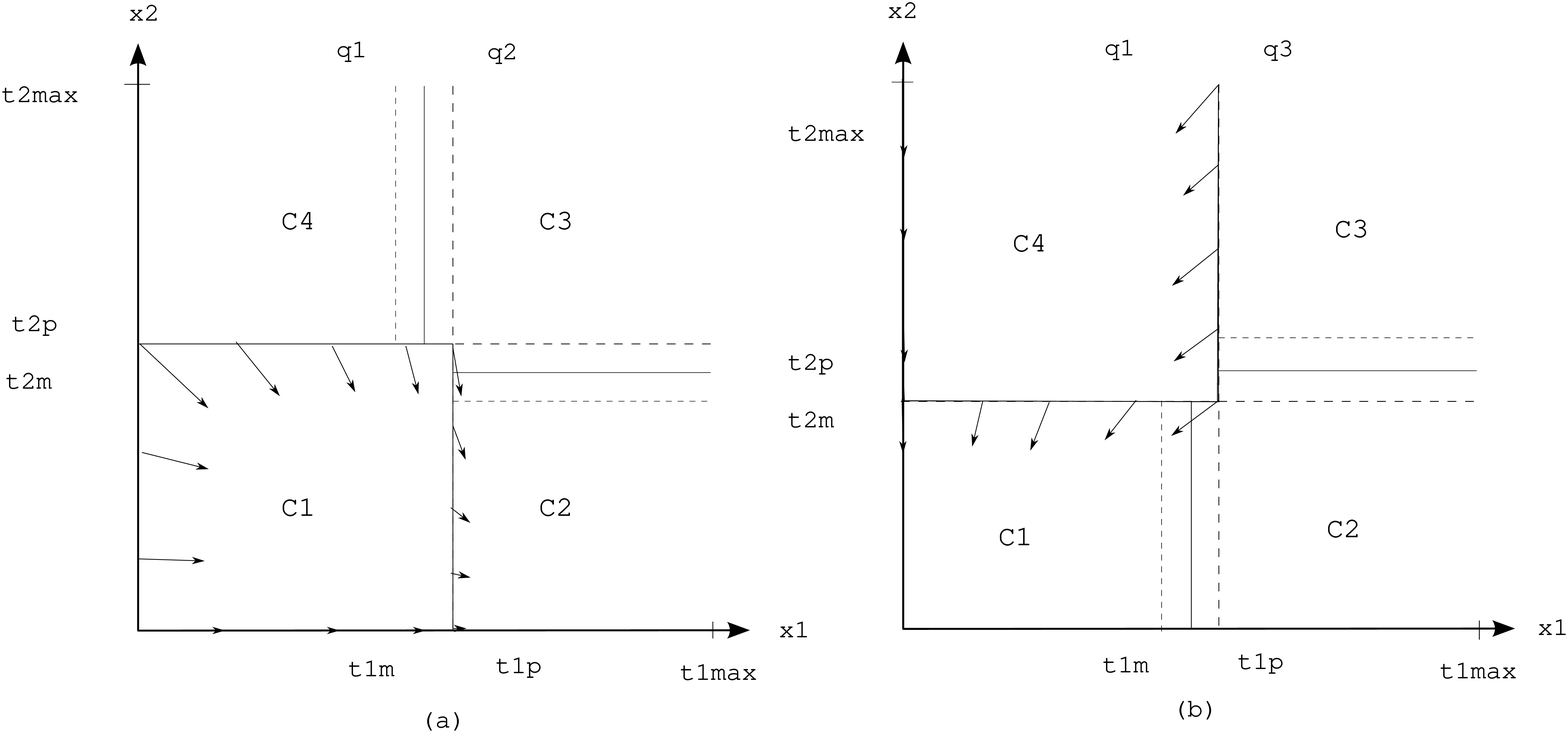}
\caption{\label{fig:V14}\emph{Vector fields at the boundaries of $C_{1}$, $C_4$, when $\theta_1+h_1<\frac{k_1}{\gamma_1}<\theta_1^{max}$,$0<\frac{k_2}{\gamma_2}<\theta_2+h_2.$ (a) The vector field at the boundaries of $C_1,$ (b) The vector field of the boundaries of $C_4.$}}
\end{figure}

\begin{figure}[h]
\centering
\psfrag{C1}[][][0.9]{$C_1$}
\psfrag{C2}[][][0.9]{$C_2$}
\psfrag{C3}[][][0.9]{$C_3$}
\psfrag{C4}[][][0.9]{$C_4$}
\psfrag{q1}[][][0.9]{$q_1=1$}
\psfrag{q2}[][][0.9]{$q_2=1$}
\psfrag{(a)}[][][0.9]{$(a)$}
\psfrag{t1m}[][][0.9]{$\theta_1\!\!-\!h_1$}
\psfrag{t1p}[][][0.9]{$\theta_1\!\!+\!h_1$}
\psfrag{t1max}[][][0.9]{$\theta_1^{\max}$}
\psfrag{x1}[][][0.9]{$x_1$}
\psfrag{x2}[][][0.9]{$x_2$}
\psfrag{(b)}[][][0.9]{$(b)$}
\psfrag{t2m}[][][0.9]{$\theta_2-h_2$}
\psfrag{t2p}[][][0.9]{$\theta_2+h_2$}
\psfrag{t2max}[][][0.9]{$\theta_2^{\max}$}
\includegraphics[width=0.85\columnwidth]{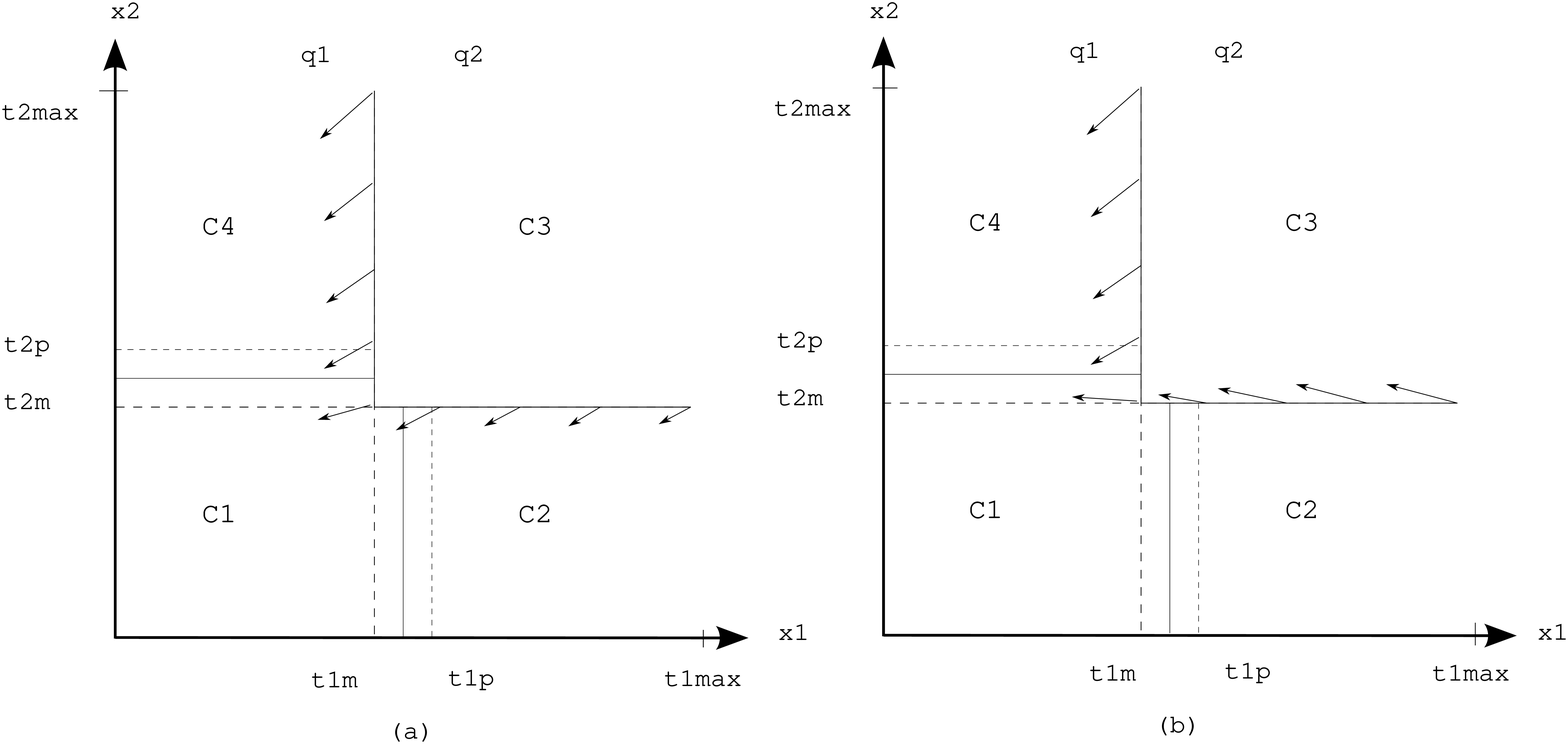}
\caption{\label{fig:V3}\emph{The vector field at the boundaries of $C_{3}$, when $\theta_1+h_1<\frac{k_1}{\gamma_1}<\theta_1^{max}$,$\frac{k_2}{\gamma_2}<\theta_2+h_2.$ (a) Case $\frac{k_2}{\gamma_2}<\theta_2-h_2$, (b) Case $\theta_2-h_2<\frac{k_2}{\gamma_2}<\theta_2+h_2$  }}
\end{figure}

\begin{itemize}
\item
For initial points $z(0, 0)\in C_1$ (see Figure \ref{fig:V14}(a)), when $x\in\{x: 0\leq x_1<\theta_1+h_1, x_2=\theta_2+h_2\}$, $f_1(x)=\left[\begin{array}{cc}k_1-\gamma_1x_1\\-\gamma_2(\theta_2+h_2)\end{array}\right]$. Since $\theta_1+h_1<\frac{k_1}{\gamma_1}<\theta_1^{max}$, $0<\frac{k_2}{\gamma_1}<\theta_2+h_2$ we have that $f_1(x)$ points inside $C_1$.
When $x\in\{x: x_1=\theta_1+h_1, 0\leq x_2\leq\theta_2+h_2\}$, $f_1(x)=\left[\begin{array}{cc}k_1-\gamma_1(\theta_1+h_1)\\-\gamma_2x_2\end{array}\right]$. Since $\frac{k_1}{\gamma_1}>\theta_1+h_1$, we have that $f_1(x)$ points outside $C_1$.
Thus, for every point $z\in C_1$, $x_1$ reaches $\theta_1+h_1$ since there is no equilibrium point in $C_1$ for this range of parameters. Then, a jump occurs. After the jump, the solution belongs to $C_2$.
\item
For each initial point $z(0, 0)\in C_4$ (see Figure \ref{fig:V14}(b)), 
when $x\in\{x: 0\leq x_1<\theta_1+h_1, x_2=\theta_2-h_2\}$, $f_4(x)=\left[\begin{array}{cc}-\gamma_1x_1\\-\gamma_2(\theta_2-h_2)\end{array}\right]$. Then, we have that $f_4(x)$ points outside $C_4$ and every solution leaves $C_4$ by jumping into $C_1$ when $x_1$ reaches $\theta_1-h_1$, from where it will enter $C_2$ in finite time. When $x\in\{x: x_1=\theta_1+h_1, x_2>\theta_2-h_2\}$, $f_4(x)=\left[\begin{array}{cc}-\gamma_1(\theta_1+h_1)\\-\gamma_2x_2\end{array}\right]$. Then, we have that $f_4(x)$ points inside $C_4$. Then, for every initial condition $z(0, 0)\in C_4$, solutions will reach $C_2$ in finite time.
\item
For every initial point $z(0, 0)\in C_3$ (see Figure \ref{fig:V3}(a)),
if $0<\frac{k_2}{\gamma_2}<\theta_2-h_2,$ when $x\in\{x: x_1\geq\theta_1-h_1, x_2=\theta_2-h_2\},$  
$f_3(x)=\left[\begin{array}{cc}-\gamma_1x_1\\k_2-\gamma_2(\theta_2-h_2)\end{array}\right]$. We have that $f_3(x)$ points outside of $C_3$. When $x\in\{x_1=\theta_1-h_1, x_2\geq\theta_2-h_2\},$ we have $f_3(x)=\left[\begin{array}{cc}-\gamma_1(\theta_1-h_1)\\k_2-\gamma_2x_2\end{array}\right],$ which points outside $C_3.$ If $\theta_2-h_2<\frac{k_2}{\gamma_2}<\theta_2+h_2$ (see Figure \ref{fig:V3}(b)), when $x\in\{x: \theta_1-h_1<x_1<\theta_1^{max}, x_2=\theta_2-h_2\},$ $f_3(x)=\left[\begin{array}{cc}-\gamma_1x_1\\k_2-\gamma_2(\theta_2-h_2)\end{array}\right]$ points inside $C_3$. When $x\in\{x: x_1=\theta_1-h_1, x_2\geq\theta_2-h_2\},$ $f_3(x)=\left[\begin{array}{cc}-\gamma_1(\theta_1-h_1)\\k_2-\gamma_2x_2\end{array}\right]$ points outside $C_3$. Then, from $z(0, 0)\in C_3,$ solutions will leave $C_3$ and jump into $C_4$or $C_2$. Using the arguments above, solutions will enter inside of $C_2$ in finite time.
\end{itemize}
From the arguments above, when $\theta_1+h_1<\frac{k_1}{\gamma_1}<\theta_1^{max}$, $0< \frac{k_2}{\gamma_2}<\theta_2+h_2$, the equilibrium point $z^*_1$ is globally asymptotically stable.

\begin{figure}[h]
\centering
\psfrag{C1}[][][0.9]{$C_1$}
\psfrag{C2}[][][0.9]{$C_2$}
\psfrag{C3}[][][0.9]{$C_3$}
\psfrag{C4}[][][0.9]{$C_4$}
\psfrag{q1}[][][0.9]{$q_1=1$}
\psfrag{q2}[][][0.9]{$q_2=1$}
\psfrag{(a)}[][][0.9]{$(a)$}
\psfrag{t1m}[][][0.9]{$\theta_1\!\!-\!h_1$}
\psfrag{t1p}[][][0.9]{$\theta_1\!\!+\!h_1$}
\psfrag{t1max}[][][0.9]{$\theta_1^{\max}$}
\psfrag{x1}[][][0.9]{$x_1$}
\psfrag{x2}[][][0.9]{$x_2$}
\psfrag{(b)}[][][0.9]{$(b)$}
\psfrag{t2m}[][][0.9]{$\theta_2-h_2$}
\psfrag{t2p}[][][0.9]{$\theta_2+h_2$}
\psfrag{t2max}[][][0.9]{$\theta_2^{\max}$}
\includegraphics[width=0.85\columnwidth]{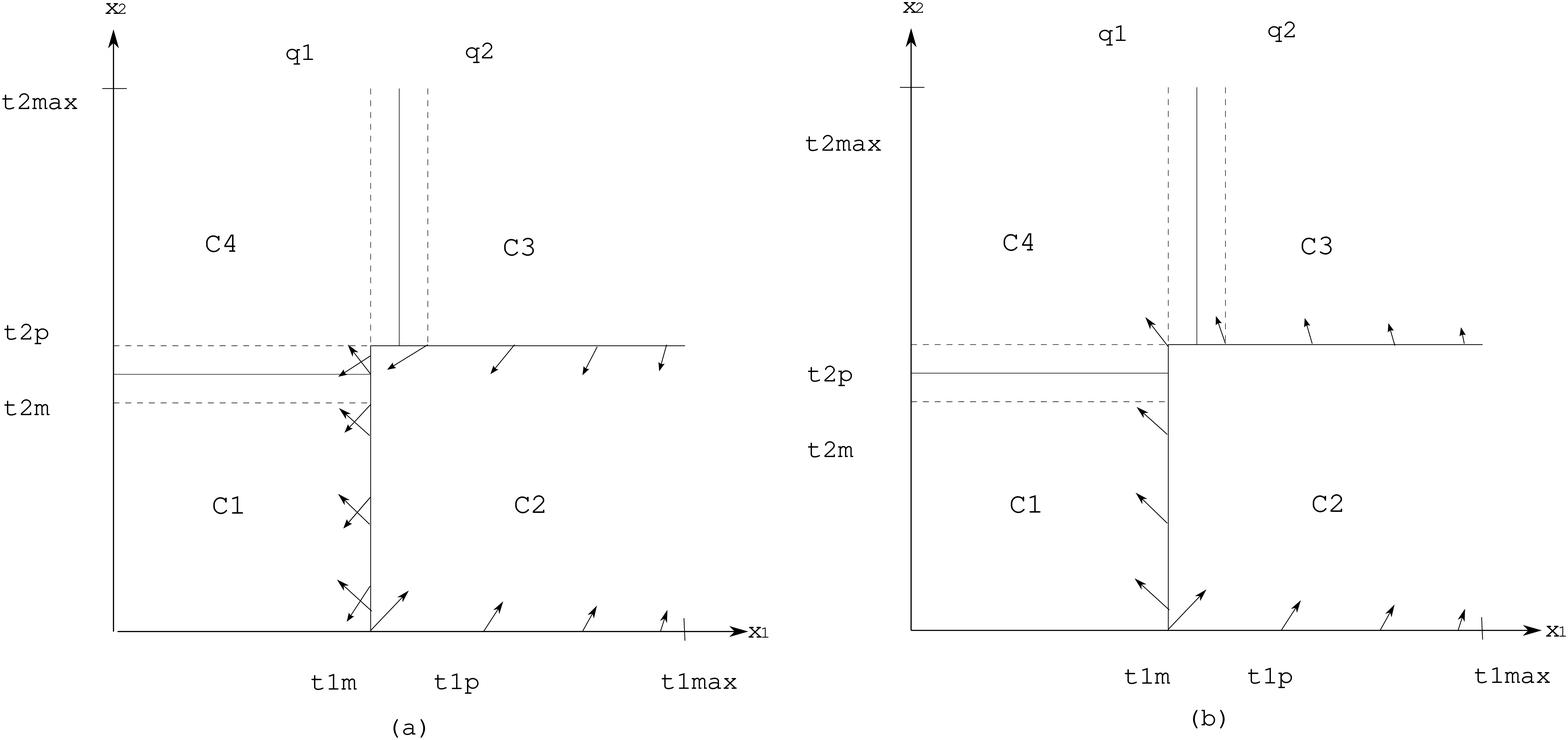}
\caption{\label{fig:C2Z2}\emph{The vector field at the boundaries of $C_{2}$, when $0<\frac{k_1}{\gamma_1}<\theta_1+h_1$. (a) Case $\frac{k_2}{\gamma_2}<\theta_2+h_2$, (b) Case $\frac{k_2}{\gamma_2}>\theta_2+h_2$. } }
\end{figure}

For case 2 in Table \ref{tab:EqPoints}, it can be proven that $z^*_2$ is globally asymptotically stable. 

Note that when $q_1=0$, $q_2=0$, $0<\frac{k_1}{\gamma_1}<\theta_1-h_1$, 
\\
$$\begin{array}{lll}
0=\dot{x}_1 =k_1-\gamma_1x_1\quad\Rightarrow\quad k_1-\gamma_1x_1=0\quad \Rightarrow \quad x_1^*=\frac{k_1}{\gamma_1}\\
0=\dot{x}_2 =\gamma_2x_2\quad\Rightarrow\quad -\gamma_2x_2=0\quad \Rightarrow \quad x_2^*=0,
\end{array}$$
from where we have $z^*_2 = [{x^*}^\top\ 0\ 0]^\top\in C_1.$ Now, change to $e$ coordinates given by\\
$$e_1 = x_1-x_1^*, \quad e_2 =x_2-x_2^*.$$
We have that
\begin{eqnarray*}
\dot{e}_1 & = & \dot{x}_1-\dot{x}_1^*= k_1-\gamma_1x_1-0=k_1-\gamma_1(e_1+x_1^*)\\
& = & k_1-\gamma_1e_1-k_1=-\gamma_1e_1,\\
\dot{e}_2 & = & \dot{x}_2=k_2-\gamma_2x_2=k_2-\gamma_2(e_2+x_2^*)\\
& = & k_2-\gamma_2e_2-k_2=-\gamma_2e_2.
\end{eqnarray*}
Then, the $x$ component of $\dot{z}=F(z)$ becomes 
$$
\dot{e}=\left[
\begin{array}{cccc}
-\gamma_1 &0\\
0& -\gamma_2.
\end{array}\right]e.
$$
Then, since $\gamma_1$ and $\gamma_2$ are positive, we have that $z^*_2$ is stable (this property can be certified with the Lyapunov function $V(e)=e^\top e$).

Now we check the vector fields on the boundaries of $C_1$ (see Figure \ref{fig:c1forwardinvariant}).

\begin{figure}[h]
\centering
\psfrag{C1}[][][0.9]{$C_1$}
\psfrag{C2}[][][0.9]{$C_2$}
\psfrag{C3}[][][0.9]{$C_3$}
\psfrag{C4}[][][0.9]{$C_4$}
\psfrag{q1}[][][0.9]{$q_1=1$}
\psfrag{q2}[][][0.9]{$q_2=1$}
\psfrag{(a)}[][][0.9]{$(a)$}
\psfrag{t1m}[][][0.9]{$\theta_1\!\!-\!h_1$}
\psfrag{t1p}[][][0.9]{$\theta_1\!\!+\!h_1$}
\psfrag{t1max}[][][0.9]{$\theta_1^{\max}$}
\psfrag{x1}[][][0.9]{$x_1$}
\psfrag{x2}[][][0.9]{$x_2$}
\psfrag{(b)}[][][0.9]{$(b)$}
\psfrag{t2m}[][][0.9]{$\theta_2-h_2$}
\psfrag{t2p}[][][0.9]{$\theta_2+h_2$}
\psfrag{t2max}[][][0.9]{$\theta_2^{\max}$}
\includegraphics[width=0.55\columnwidth]{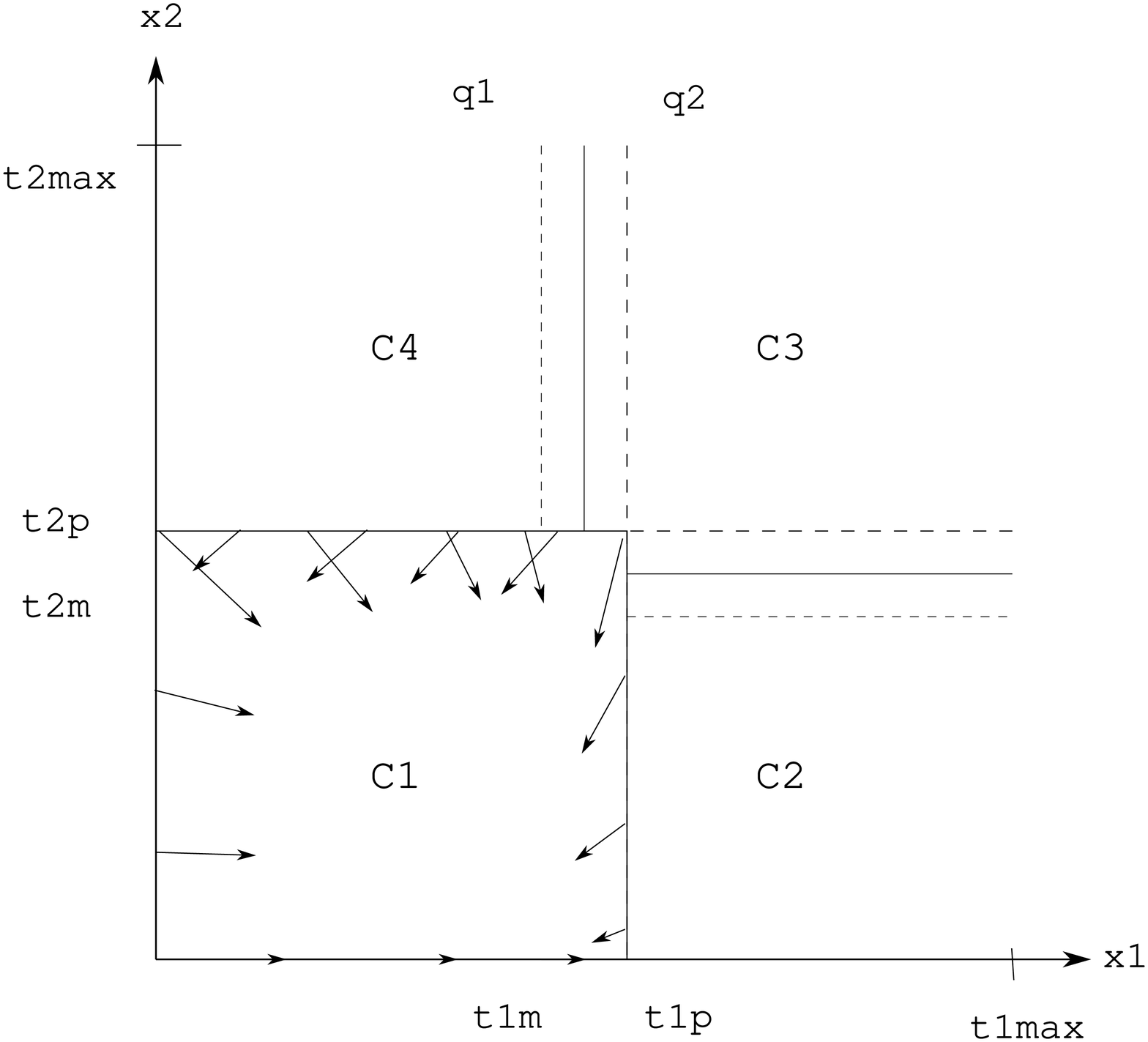}
\caption{\label{fig:c1forwardinvariant}\emph{The vector field at the boundaries of $C_{1}$, when $0<\frac{k_1}{\gamma_1}<\theta_1-h_1$.}}
\end{figure}

\begin{itemize}
\item
When $x\in\{x: x_1=\theta_1+h_1, 0\leq x_2\leq\theta_2+h_2\}$, $f_1(x)=\left[\begin{array}{cc}k_1-\gamma_1(\theta_1+h_1)\\-\gamma_2x_2\end{array}\right]$. Since $0<\frac{k_1}{\gamma_1}<\theta_1-h_1$, we have that $f_1(x)$ points inside $C_1$.
\item
When $x\in\{x: 0\leq x_1\leq\theta_1+h_1, x_2=\theta_2+h_2\}$, $f_1(x)=\left[\begin{array}{cc}k_1-\gamma_1x_1\\-\gamma_2(\theta_2+h_2)\end{array}\right]$. Since $0<\frac{k_1}{\gamma_1}<\theta_1-h_1$, we have that $f_1(x)$ points inside $C_1$.
\item
When $x\in\{x: 0\leq x_1\leq\theta_1+h_1, x_2=0\}$, $f_1(x)=\left[\begin{array}{cc} k_1-\gamma_1x_1\\0\end{array}\right]$. $f_1(x)$ is tangent to the boundary of $C_1$.
\item
When $x\in\{x: x_1=0, 0\leq x_2\leq\theta_2+h_2\}$, $f_1(x)=\left[\begin{array}{cc}k_1\\-\gamma_2x_2\end{array}\right]$. Since $k_1>0$, we have that $f_1(x)$ points inside $C_1$.
\end{itemize}
Then, every trajectory that enters or starts from $C_1$  will stay or never leave $C_1$.
Then, the set $C_1$ is forward invariant. 

Since the equilibrium point $z_2^*$ belongs to $C_1$, 
every trajectory that reaches or starts from $C_1$ converges to $z_2^*$.

Global asymptotic stability follows since for every initial condition $z(0,0)\in(C\cup D)\backslash C_1$, solutions reach $C_1$ in finite time.  
To establish this property, we check the vector field of the system on the boundary of each set $C_i$, for each $i\in\{2, 3, 4\}.$
\begin{itemize}
\item
For initial points $z(0, 0)\in C_2$:
\begin{itemize}
\item if $\frac{k_2}{\gamma_2}<\theta_2+h_2$ (see Figure \ref{fig:C2Z2}(a)), when $x\in\{x: x_1=\theta_1-h_1, 0\leq x_2\leq\theta_2+h_2\}$, $f_ 2(x)=\left[\begin{array}{cc}k_1-\gamma_1(\theta_1-h_1)\\k_2-\gamma_2x_2\end{array}\right]$. Since $0<\frac{k_1}{\gamma_1}<\theta_1-h_1$, we have that $f_2(x)$ points outside of $C_2$. When $x\in\{x: x_1>\theta_1-h_1, x_2=\theta_2+h_2\}$, $f_2(x)=\left[\begin{array}{cc}k_1-\gamma_1x_1\\k_2-\gamma_2(\theta_2+h_2)\end{array}\right]$. 
Since $0<\frac{k_1}{\gamma_1}<\theta_1-h_1$, we have that $f_2(x)$ points inside of $C_2$.

Since there is no equilibrium point in $C_2$ for this range of parameters
and the dynamics of $x$ are linear,
$x_1$ reaches $\theta_1-h_1$ for every point $z\in C_2$. Then, a jump occurs. After the jump, the solution belongs to $C_1$. 
\item If $\frac{k_2}{\gamma_2}>\theta_2+h_2$ (see Figure \ref{fig:C2Z2}(b)), when $x\in\{x: x_1=\theta_1-h_1, 0\leq x_2\leq\theta_2+h_2\}$, $f_{2}(x)=\left[\begin{array}{cc}k_1-\gamma_1(\theta_1-h_1)\\k_2-\gamma_2x_2\end{array}\right]$. Since $0<\frac{k_1}{\gamma_1}<\theta_1-h_1$, we have that $f_{2}(x)$ points outside of $C_2$. When $x\in\{x: x_1\geq\theta_1-h_1, x_2=\theta_2+h_2\}$, $f_1(x)=\left[\begin{array}{cc}k_1-\gamma_1x_1\\k_2-\gamma_2(\theta_2+h_2)\end{array}\right]$. Since $0<\frac{k_1}{\gamma_1}<\theta_1-h_1$, we have that $f_2(x)$ points outside $C_2$. 
\end{itemize}
Then, from $z(0, 0)\in C_2,$ solutions will leave $C_2$ and jump into $C_1$ or $C_3$.
\item
For every initial point $z(0, 0)\in C_3$:
\begin{itemize}
\item
if $\frac{k_2}{\gamma_2}<\theta_2-h_2,$ when $x\in\{x: x_1\geq\theta_1-h_1, x_2=\theta_2-h_2\}$  (see similar case shown in Figure \ref{fig:V3}(a)), 
$f_3(x)=\left[\begin{array}{cc}-\gamma_1x_1\\k_2-\gamma_2(\theta_2-h_2)\end{array}\right]$. We have that $f_3(x)$ points outside of $C_3$. When $x\in\{x: x_1=\theta_1-h_1, x_2\geq\theta_2-h_2\},$ we have $f_3(x)=\left[\begin{array}{cc}-\gamma_1(\theta_1-h_1)\\k_2-\gamma_2x_2\end{array}\right],$ which points outside of $C_3.$ 
\item
If $\theta_2-h_2<\frac{k_2}{\gamma_2}<\theta_2+h_2$ (see similar case shown in  Figure \ref{fig:V3}(b)), when $x\in\{x: x_1=\theta_1-h_1, x_2\geq\theta_2-h_2\},$ $f_3(x)=\left[\begin{array}{cc}-\gamma_1(\theta_1-h_1)\\k_2-\gamma_2x_2\end{array}\right]$ points outside of $C_3$.
When $x\in\{x: x_1>\theta_1-h_1, x_2=\theta_2-h_2\}$, 
$f_3(x)=\left[\begin{array}{cc}-\gamma_1x_1\\k_2-\gamma_2(\theta_2-h_2)\end{array}\right]$. We have that $f_3(x)$ points inside of $C_3$. 
\end{itemize}
Then, from $z(0, 0)\in C_3,$ solutions will leave $C_3$ and jump into $C_4$ or $C_2$. 
\item
For each initial point $z(0, 0)\in C_4$ (see similar case shown in Figure \ref{fig:V14}(b)), 
when $x\in\{x: 0\leq x_1<\theta_1+h_1, x_2=\theta_2-h_2\}$, $f_4(x)=\left[\begin{array}{cc}-\gamma_1x_1\\-\gamma_2(\theta_2-h_2)\end{array}\right]$. Then, we have that $f_4(x)$ points outside $C_4$ and every solution leaves $C_4$ by jumping into $C_1$. When $x\in\{x: x_1=\theta_1+h_1, x_2>\theta_2-h_2\}$, $f_4(x)=\left[\begin{array}{cc}-\gamma_1(\theta_1+h_1)\\-\gamma_2x_2\end{array}\right]$. Then, we have that $f_4(x)$ points inside $C_4$.
Then, for every initial condition $z(0, 0)\in C_4$, solutions will reach $C_1$ in finite time.
\end{itemize}

From the above analysis, we have that: 
1) from $C_2$ trajectories go to either $C_1$ or $C_3$;
2) from $C_3$ trajectories go to either $C_1$ or $C_4$;
3) from $C_4$ trajectories go to $C_1$.

Then, trajectories eventually enter $C_1$, which, using the arguments above, implies that the equilibrium point $z^*_2$ is globally asymptotically stable.

Similarly, for case 4 in Table \ref{tab:EqPoints}, it can be proven that $z^*_2$ is globally asymptotically stable. Since stability of $z^*_2$ was proven in the proof of case 2 in Table \ref{tab:EqPoints}, now we check the vector fields on the boundaries of $C_1$  when $\theta_1-h_1<\frac{k_1}{\gamma_1}<\theta_1+h_1$, $\theta_2+h_2<\frac{k_2}{\gamma_2}<\theta_2^{max}$ (see Figure \ref{fig:c1forwardinvariant1}).

\begin{figure}[h]
\centering
\psfrag{C1}[][][0.9]{$C_1$}
\psfrag{C2}[][][0.9]{$C_2$}
\psfrag{C3}[][][0.9]{$C_3$}
\psfrag{C4}[][][0.9]{$C_4$}
\psfrag{q10}[][][0.9]{$q_1=0$}
\psfrag{q20}[][][0.9]{$q_2=0$}
\psfrag{t1-h1}[][][0.9]{$\theta_1\!\!-\!h_1$}
\psfrag{t1h1}[][][0.9]{$\theta_1\!\!+\!h_1$}
\psfrag{t1max}[][][0.9]{$\theta_1^{\max}$}
\psfrag{x1}[][][0.9]{$x_1$}
\psfrag{x2}[][][0.9]{$x_2$}
\psfrag{t2-h2}[][][0.9]{$\!\!\!\!\theta_2-h_2$}
\psfrag{t2h2}[][][0.9]{$\!\!\theta_2+h_2$}
\psfrag{t2max}[][][0.9]{$\theta_2^{\max}$}
\includegraphics[width=0.55\columnwidth]{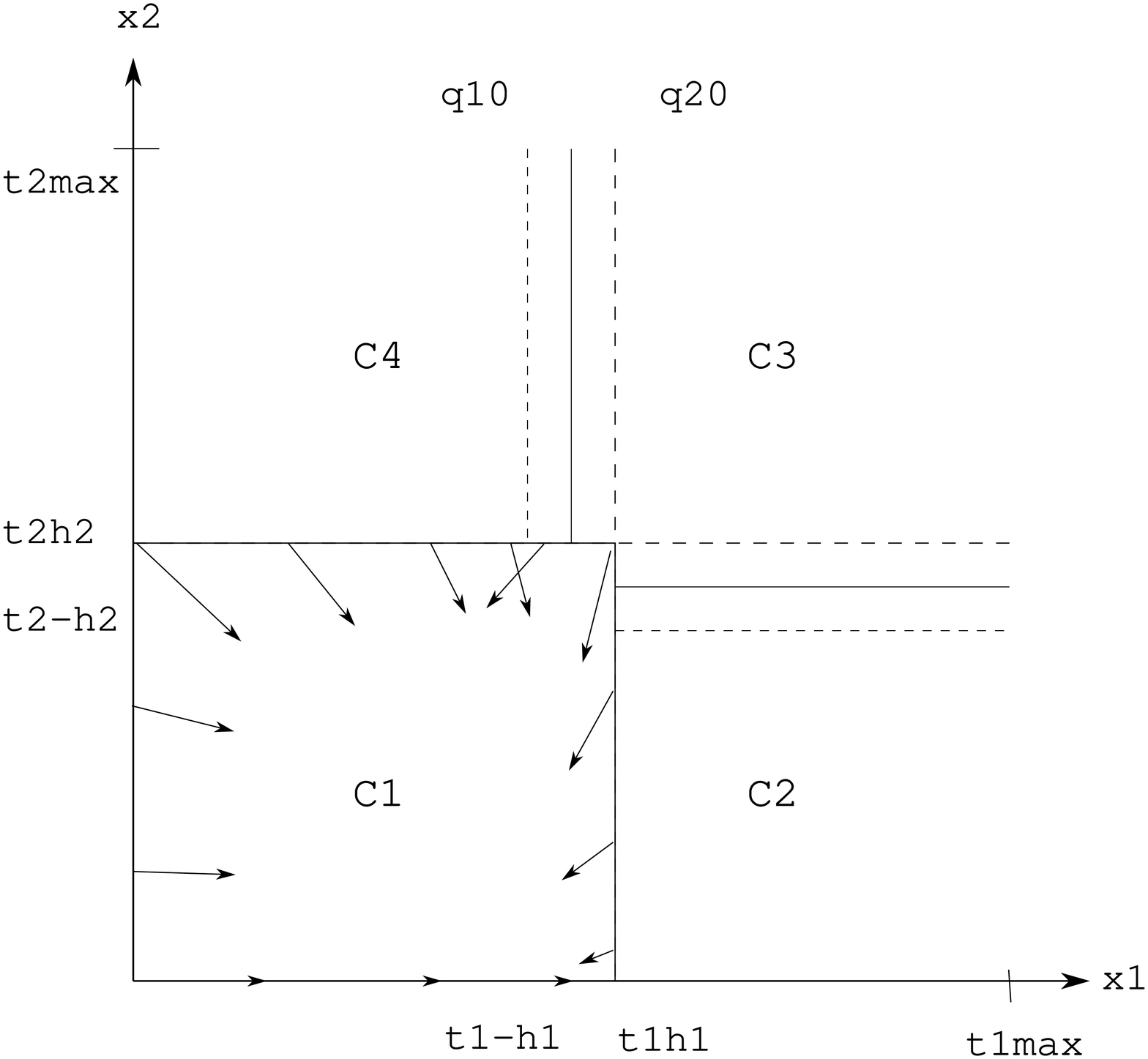}
\caption{\label{fig:c1forwardinvariant1}\emph{The vector field at the boundaries of $C_{1}$, when $\theta_1-h_1<\frac{k_1}{\gamma_1}<\theta_1+h_1$, $\theta_2+h_2<\frac{k_2}{\gamma_2}<\theta_2^{max}$.}}
\end{figure}

\begin{itemize}
\item
When $x\in\{x: x_1=\theta_1+h_1, 0\leq x_2\leq\theta_2+h_2\}$, $f_1(x)=\left[\begin{array}{cc}k_1-\gamma_1(\theta_1+h_1)\\-\gamma_2x_2\end{array}\right]$. Since $\theta_1-h_1<\frac{k_1}{\gamma_1}<\theta_1+h_1$, we have that $f_1(x)$ points inside $C_1$.
\item
When $x\in\{x: 0\leq x_1\leq\theta_1+h_1, x_2=\theta_2+h_2\}$, $f_1(x)=\left[\begin{array}{cc}k_1-\gamma_1x_1\\-\gamma_2(\theta_2+h_2)\end{array}\right]$. Since $\theta_1-h_1<\frac{k_1}{\gamma_1}<\theta_1+h_1$, $\theta_2+h_2<\frac{k_2}{\gamma_2}<\theta_2^{max}$, we have that $f_1(x)$ points inside $C_1$.
\item
When $x\in\{x: 0\leq x_1\leq\theta_1+h_1, x_2=0\}$, $f_1(x)=\left[\begin{array}{cc} k_1-\gamma_1x_1\\0\end{array}\right]$. $f_1(x)$ is tangent to the boundary of $C_1$.
\item
When $x\in\{x: x_1=0, 0\leq x_2\leq\theta_2+h_2\}$, $f_1(x)=\left[\begin{array}{cc}k_1\\-\gamma_2x_2\end{array}\right]$. Since $k_1>0$, we have that $f_1(x)$ points inside $C_1$.
\end{itemize}

Then, every trajectory that enters or starts from $C_1$  will stay or never leave $C_1$.
Then, the set $C_1$ is forward invariant. 

Since the equilibrium point $z_2^*$ belongs to $C_1$, 
every trajectory that reaches or starts from $C_1$ converges to $z_2^*$.

\begin{figure}[h]
\centering
\psfrag{C1}[][][0.9]{$C_1$}
\psfrag{C2}[][][0.9]{$C_2$}
\psfrag{C3}[][][0.9]{$C_3$}
\psfrag{C4}[][][0.9]{$C_4$}
\psfrag{q11}[][][0.9]{$q_1=1$}
\psfrag{q20}[][][0.9]{$q_2=0$}
\psfrag{t1-h1}[][][0.9]{$\theta_1\!\!-\!h_1$}
\psfrag{t1h1}[][][0.9]{$\theta_1\!\!+\!h_1$}
\psfrag{t1m}[][][0.9]{$\theta_1^{\max}$}
\psfrag{x1}[][][0.9]{$x_1$}
\psfrag{x2}[][][0.9]{$x_2$}
\psfrag{t2-h2}[][][0.9]{$\!\!\!\!\!\theta_2-h_2$}
\psfrag{t2h2}[][][0.9]{$\!\!\!\!\theta_2+h_2$}
\psfrag{t2max}[][][0.9]{$\theta_2^{\max}$}
\includegraphics[width=0.55\columnwidth]{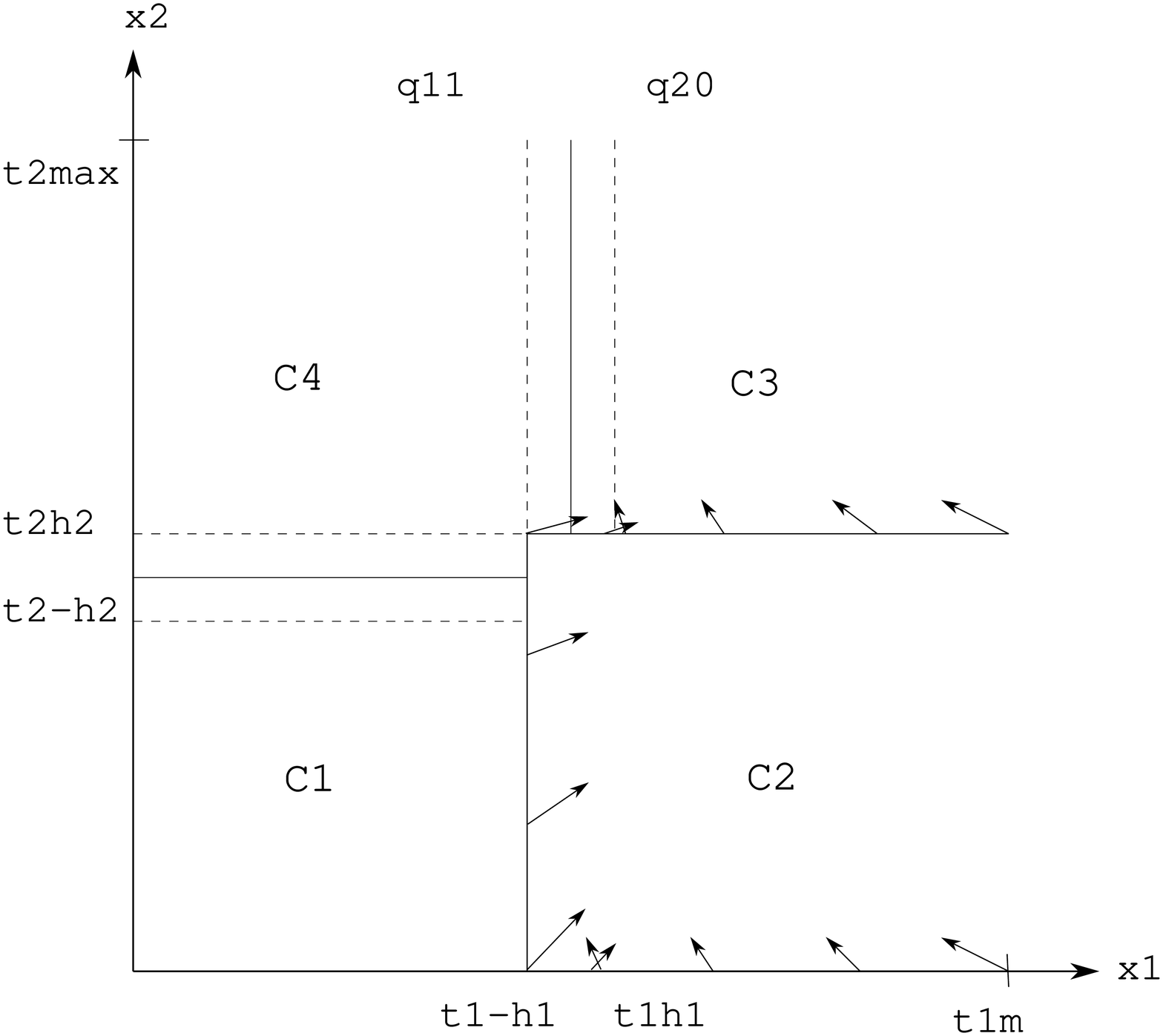}
\caption{\label{fig:c2case4}\emph{The vector field at the boundaries of $C_{2}$, when $\theta_1-h_1<\frac{k_1}{\gamma_1}<\theta_1+h_1$, $\theta_2+h_2<\frac{k_2}{\gamma_2}<\theta_2^{max}$.}}
\end{figure}

\begin{figure}[h]
\centering
\psfrag{C1}[][][0.9]{$C_1$}
\psfrag{C2}[][][0.9]{$C_2$}
\psfrag{C3}[][][0.9]{$C_3$}
\psfrag{C4}[][][0.9]{$C_4$}
\psfrag{q11}[][][0.9]{$q_1=1$}
\psfrag{q21}[][][0.9]{$q_2=1$}
\psfrag{t1-h1}[][][0.9]{$\theta_1\!\!-\!h_1$}
\psfrag{t1h1}[][][0.9]{$\theta_1\!\!+\!h_1$}
\psfrag{t1max}[][][0.9]{$\theta_1^{\max}$}
\psfrag{x1}[][][0.9]{$x_1$}
\psfrag{x2}[][][0.9]{$x_2$}
\psfrag{t2-h2}[][][0.9]{$\!\!\!\!\!\theta_2-h_2$}
\psfrag{t2h2}[][][0.9]{$\!\!\!\theta_2+h_2$}
\psfrag{t2max}[][][0.9]{$\theta_2^{\max}$}
\includegraphics[width=0.55\columnwidth]{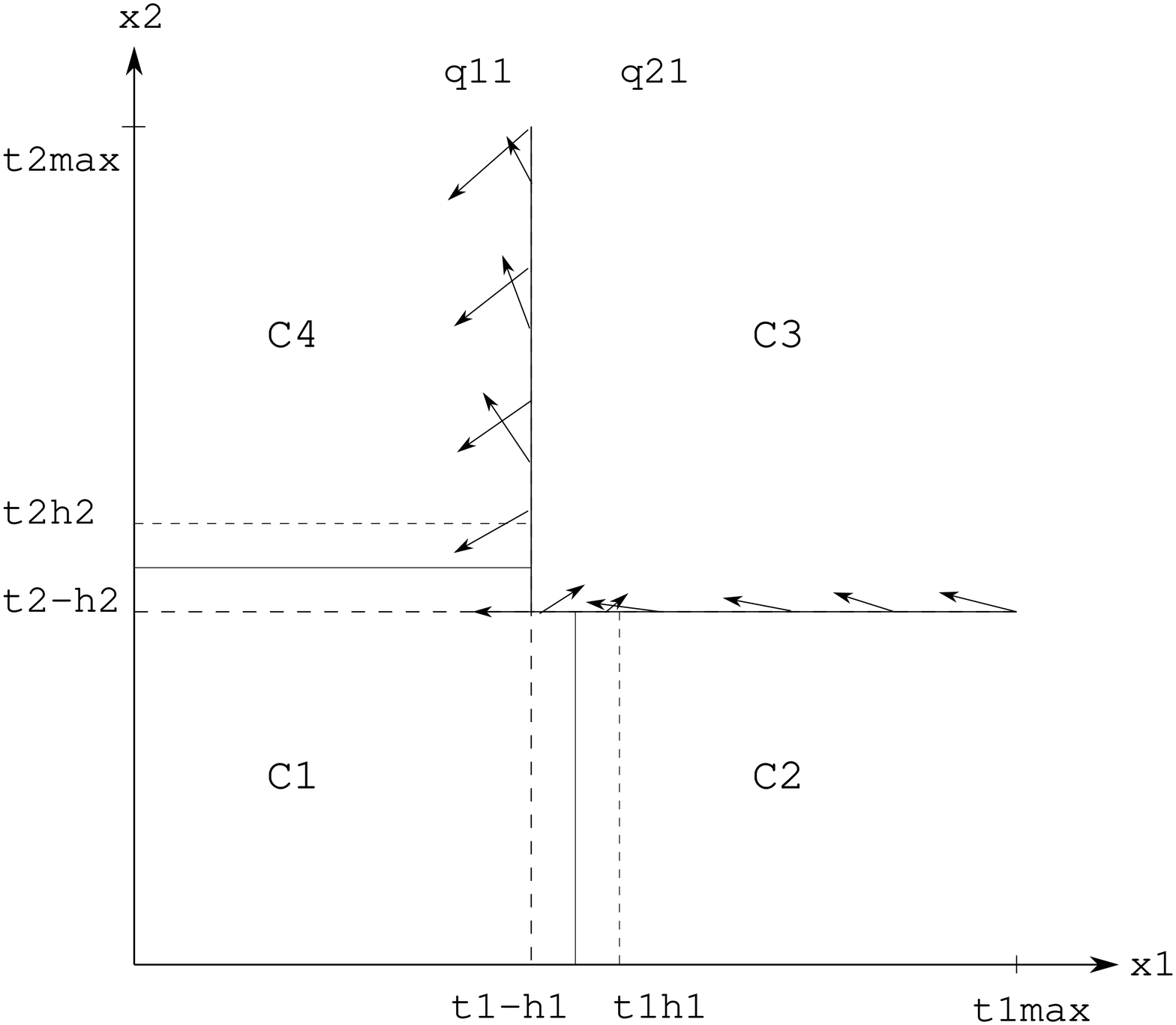}
\caption{\label{fig:c3case4}\emph{The vector field at the boundaries of $C_{3}$, when $\theta_1-h_1<\frac{k_1}{\gamma_1}<\theta_1+h_1$, $\theta_2+h_2<\frac{k_2}{\gamma_2}<\theta_2^{max}$.}}
\end{figure}

Global asymptotic stability follows since for every initial condition $z(0,0)\in(C\cup D)\backslash C_1$, solutions reach $C_1$ in finite time.  
To establish this property, we check the vector field of the system on the boundary of each set $C_i$, for each $i\in\{2, 3, 4\}.$
\begin{itemize}
\item
For initial points $z(0, 0)\in C_2$:
When $x\in\{x: x_1=\theta_1-h_1, 0\leq x_2<\theta_2+h_2\}$(see Figure \ref{fig:c2case4}), $f_ 2(x)=\left[\begin{array}{cc}k_1-\gamma_1(\theta_1-h_1)\\k_2-\gamma_2x_2\end{array}\right]$. Since $\theta_1-h_1<\frac{k_1}{\gamma_1}<\theta_1+h_1$, we have that $f_2(x)$ points inside of $C_2$. When $x\in\{x: x_1\geq\theta_1-h_1, x_2=\theta_2+h_2\}$, $f_2(x)=\left[\begin{array}{cc}k_1-\gamma_1x_1\\k_2-\gamma_2(\theta_2+h_2)\end{array}\right]$. 
Since $\theta_1-h_1<\frac{k_1}{\gamma_1}<\theta_1+h_1$, $\theta_2+h_2<\frac{k_2}{\gamma_2}<\theta_2^{max}$, we have that $f_2(x)$ points outside of $C_2$.

Since there is no equilibrium point in $C_2$ for this range of parameters
and the dynamics of $x$ are linear,
$x_2$ reaches $\theta_2+h_2$ for every point $z\in C_2$.

Then, from $z(0, 0)\in C_2,$ solutions will leave $C_2$ and jump into  $C_3$.
\item
For every initial point $z(0, 0)\in C_3:$ (see Figure \ref{fig:c3case4}) when $x\in\{x: x_1=\theta_1-h_1, x_2\geq\theta_2-h_2\},$ since $\theta_1-h_1<\frac{k_1}{\gamma_1}<\theta_1+h_1$, $f_3(x)=\left[\begin{array}{cc}-\gamma_1(\theta_1-h_1)\\k_2-\gamma_2x_2\end{array}\right]$ points outside of $C_3$. 
When $x\in\{x: x_1>\theta_1-h_1, x_2=\theta_2-h_2\}$, 
$f_3(x)=\left[\begin{array}{cc}-\gamma_1x_1\\k_2-\gamma_2(\theta_2-h_2)\end{array}\right]$. Since $\theta_1-h_1<\frac{k_1}{\gamma_1}<\theta_1+h_1$, we have that $f_3(x)$ points inside of $C_3$.
Then, from $z(0, 0)\in C_3,$ solutions will leave $C_3$ and jump into $C_4$. 
\item
For each initial point $z(0, 0)\in C_4:$ (see similar case shown in Figure \ref{fig:V14}(b))
when $x\in\{x: 0\leq x_1<\theta_1+h_1, x_2=\theta_2-h_2\}$, since $\theta_1-h_1<\frac{k_1}{\gamma_1}<\theta_1+h_1$, $f_4(x)=\left[\begin{array}{cc}-\gamma_1x_1\\-\gamma_2(\theta_2-h_2)\end{array}\right]$. Then, we have that $f_4(x)$ points outside $C_4$ and every solution leaves $C_4$ by jumping into $C_1$. When $x\in\{x: x_1=\theta_1+h_1, x_2>\theta_2-h_2\}$, $f_4(x)=\left[\begin{array}{cc}-\gamma_1(\theta_1+h_1)\\-\gamma_2x_2\end{array}\right]$. Then, we have that $f_4(x)$ points inside $C_4$.
Then, for every initial condition $z(0, 0)\in C_4$, solutions will reach $C_1$ in finite time.
\end{itemize}

From the above analysis, we have that: 
1) from $C_2$ trajectories go to $C_3$;
2) from $C_3$ trajectories go to $C_4$;
3) from $C_4$ trajectories go to $C_1$.

Then, trajectories eventually enter $C_1$, which, using the arguments above, implies that the equilibrium point $z^*_2$ is globally asymptotically stable.

\subsection{Proof of Proposition~\ref{eqn:Stable2}}
\label{app:Stable2Proof}

For case 3, when $\theta_1-h_1<\frac{k_1}{\gamma_1}<\theta_1+h_1$, $\frac{k_2}{\gamma_2}<\theta_2+h_2$, $z^*_1$ and $z^*_2$ are located in the region noted as $C'$ in Figure \ref{fig:C'}.

\begin{figure}[h!]
\begin{center}
\psfrag{C1}[][][0.9]{$C_1$}
\psfrag{C2}[][][0.9]{$C_2$}
\psfrag{C3}[][][0.9]{$C_3$}
\psfrag{C4}[][][0.9]{$C_4$}
\psfrag{1C}[][][0.9]{$C'$}
\psfrag{t1m}[][][0.9]{$\theta_1\!\!-\!h_1$}
\psfrag{t1p}[][][0.9]{$\theta_1\!\!+\!h_1$}
\psfrag{t1max}[][][0.9]{$\theta_1^{\max}$}
\psfrag{t2m}[][][0.9]{$\theta_2-h_2$}
\psfrag{t2p}[][][0.9]{$\theta_2+h_2$}
\psfrag{t2max}[][][0.9]{$\theta_2^{\max}$}
\psfrag{x1}[][][0.9]{$x_1$}
\psfrag{x2}[][][0.9]{$x_2$}
\includegraphics[width=0.5\columnwidth]{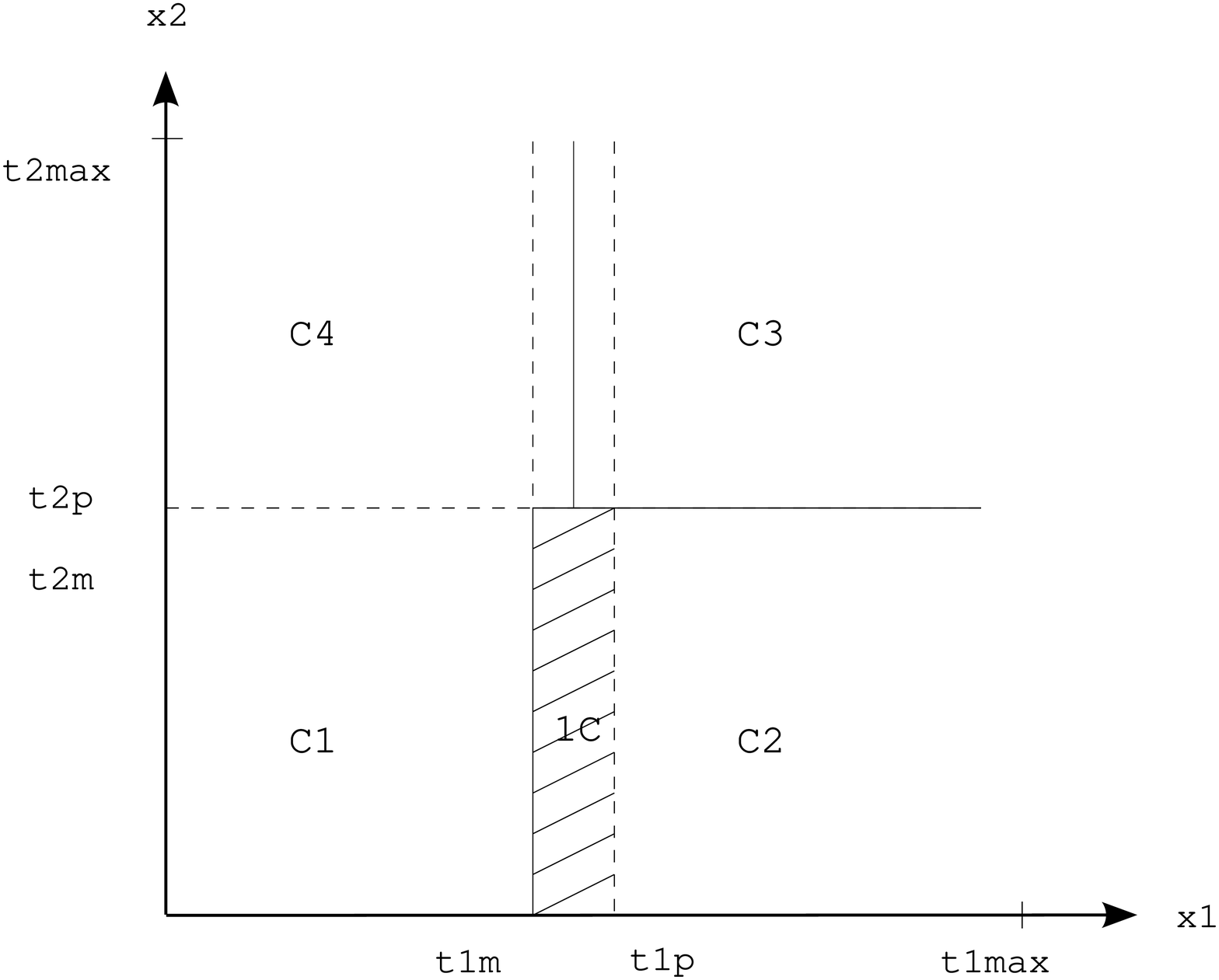}
\end{center}
\caption{\label{fig:C'}\emph{Region $C'$ given by the overlap between $C_1$ and $C_2$ when projected onto the $x$ component.}}
\end{figure}

\begin{figure}[h]
\centering
\psfrag{C1}[][][0.9]{$C_1$}
\psfrag{C2}[][][0.9]{$C_2$}
\psfrag{C3}[][][0.9]{$C_3$}
\psfrag{C4}[][][0.9]{$C_4$}
\psfrag{q1}[][][0.9]{$q_1=1$}
\psfrag{q2}[][][0.9]{$q_2=0$}
\psfrag{q3}[][][0.9]{$q_1=0$}
\psfrag{(a)}[][][0.9]{$(a)$}
\psfrag{t1m}[][][0.9]{$\theta_1\!\!-\!h_1$}
\psfrag{t1p}[][][0.9]{$\theta_1\!\!+\!h_1$}
\psfrag{t1max}[][][0.9]{$\theta_1^{\max}$}
\psfrag{x1}[][][0.9]{$x_1$}
\psfrag{x2}[][][0.9]{$x_2$}
\psfrag{(b)}[][][0.9]{$(b)$}
\psfrag{t2m}[][][0.9]{$\theta_2-h_2$}
\psfrag{t2p}[][][0.9]{$\theta_2+h_2$}
\psfrag{t2max}[][][0.9]{$\theta_2^{\max}$}
\includegraphics[width=0.95\columnwidth]{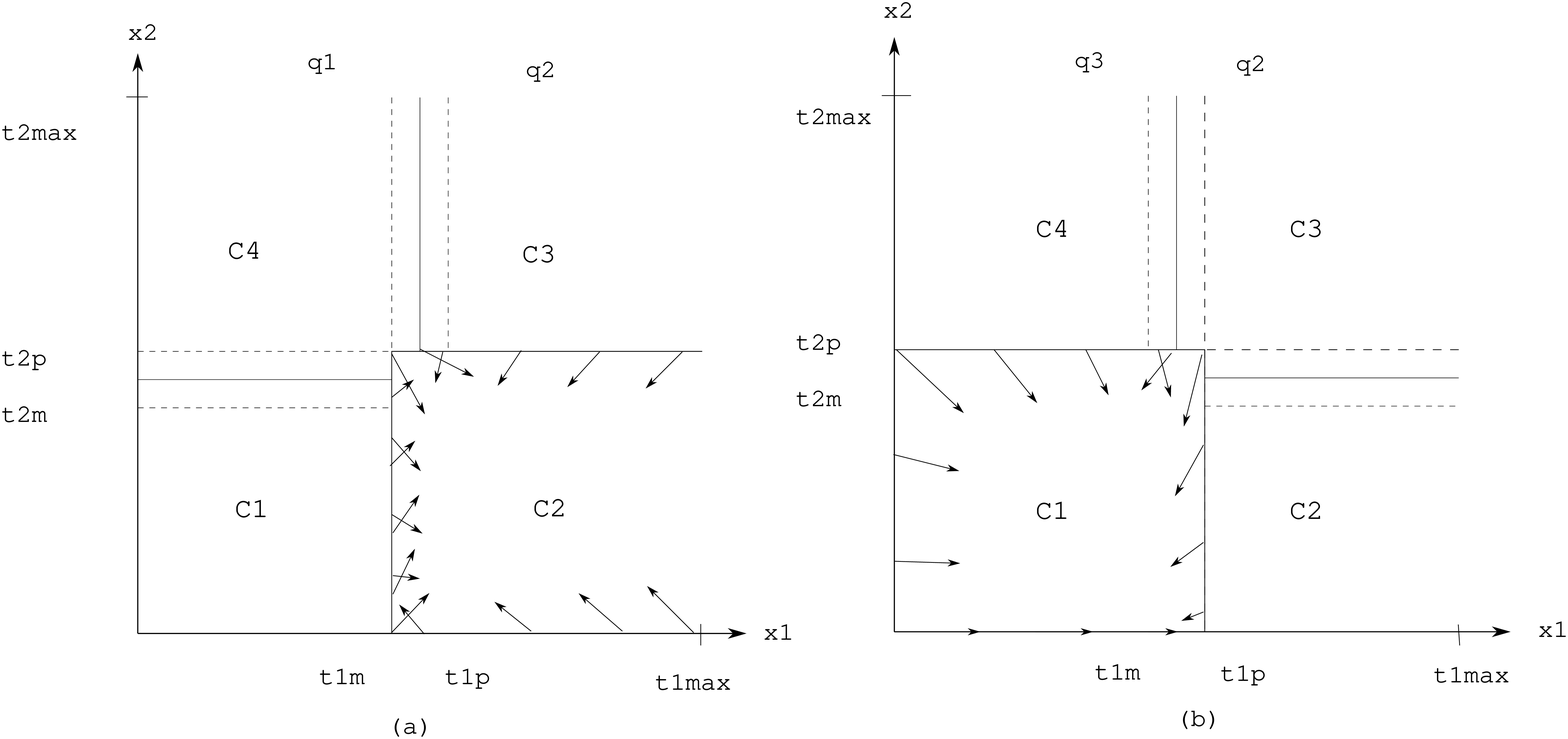}
\caption{\label{fig:c1c2case3}\emph{when $\theta_1-h_1<\frac{k_1}{\gamma_1}<\theta_1+h_1$,$0<\frac{k_2}{\gamma_2}<\theta_2+h_2$: (a) the vector field at the boundaries of $C_{2}$, (b) the vector field at the boundaries of $C_{1}$.}}
\end{figure}

For an initial point $z(0,0)\in C_2$,   we check the vector fields on the boundaries of $C_2$ (see Figure \ref{fig:c1c2case3}(a)).
\begin{itemize}
\item
 When $x\in\{x: x_1=\theta_1-h_1, 0\leq x_2\leq\theta_2+h_2\}$, $f_{2}(x)=\left[\begin{array}{cc}k_1-\gamma_1(\theta_1-h_1)\\k_2-\gamma_2x_2\end{array}\right]$. Since $\theta_1-h_1<\frac{k_1}{\gamma_1}<\theta_1+h_1$,$0<\frac{k_2}{\gamma_2}<\theta_2+h_2$, we have that $f_2(x)$ points inside of $C_2$. 
\item
When $x\in\{x: x_1\geq\theta_1-h_1, x_2=\theta_2+h_2\}$, $f_2(x)=\left[\begin{array}{cc}k_1-\gamma_1x_1\\k_2-\gamma_2(\theta_2+h_2)\end{array}\right]$. Since $\theta_1-h_1<\frac{k_1}{\gamma_1}<\theta_1+h_1$, $0<\frac{k_2}{\gamma_2}<\theta_2+h_2$, we have that $f_2(x)$ points inside of $C_2$. 
\item
When $x\in\{x: x_1\geq\theta_1-h_1, x_2=0\}$, $f_2(x)=\left[\begin{array}{cc}k_1-\gamma_1x_1\\k_2\end{array}\right]$. Since $\theta_1-h_1<\frac{k_1}{\gamma_1}<\theta_1+h_1$,$k_2>0$, we have that $f_2(x)$ points inside of $C_2$. 
\end{itemize}
Then, once a trajectory enters or starts from $C_2,$ it will stay or never leave $C_2$
Then, the set $C_2$ is forward invariant. Since the equilibrium point $z_1^*$ belongs to $C_2$, every trajectory reaching or starting from $C_2$ converges to $z_1^*$.

For $z(0, 0)\in C_1$, we check the vector fields on the boundaries of $C_1$ (see Figure \ref{fig:c1c2case3}(b)).
\begin{itemize}
\item
When $x\in\{x: x_1=\theta_1+h_1, 0\leq x_2\leq\theta_2+h_2\}$, $f_1(x)=\left[\begin{array}{cc}k_1-\gamma_1(\theta_1+h_1)\\-\gamma_2x_2\end{array}\right]$. Since $\theta_1-h_1<\frac{k_1}{\gamma_1}<\theta_1+h_1$, we have that $f_1(x)$ points inside $C_1$.
\item
When $x\in\{x: 0\leq x_1\leq\theta_1+h_1, x_2=\theta_2+h_2\}$, $f_1(x)=\left[\begin{array}{cc}k_1-\gamma_1x_1\\-\gamma_2(\theta_2+h_2)\end{array}\right]$. Since $\theta_1-h_1<\frac{k_1}{\gamma_1}<\theta_1+h_1$, we have that $f_1(x)$ points inside $C_1$.
\item
When $x\in\{x: 0\leq x_1\leq\theta_1+h_1, x_2=0\}$, $f_1(x)=\left[\begin{array}{cc}k_1-\gamma_1(\theta_1+h_1)\\0\end{array}\right]$. $f_1(x)$ is tangent to the boundary of $C_1$.
\item
When $x\in\{x: x_1=0, 0\leq x_2\leq\theta_2+h_2\}$, $f_1(x)=\left[\begin{array}{cc}k_1\\-\gamma_2x_2\end{array}\right]$. Since $k_1>0$, we have that $f_1(x)$ points inside $C_1$.
\end{itemize}
Then, once a trajectory enters or starts from $C_1,$ it will stay or never leave $C_1.$
Then, the set $C_1$ is forward invariant. Since the equilibrium point $z_2^*$ belongs to $C_1$, every trajectory reaching or starting from $C_1$ converges to $z_2^*$.

For initial points in $C_4$, the vector field of the boundary at $C_4$ is shown in Figure \ref{fig:V14} (b). When $x\in\{x: 0<x_1<\theta_1+h_1, x_2=\theta_2-h_2\}$, $f_4(x)=\left[\begin{array}{cc}-\gamma_1x_1\\-\gamma_2(\theta_2-h_2)\end{array}\right]$. Then, we have that $f_4(x)$ points outside $C_4$ and every solution leaves $C_4$ by jumping into $C_1$. Then, for every initial condition $z(0, 0)\in C_4$, solutions will reach $C_1$ in finite time. As $z^*_2\in C_1$, so the trajectory will stay in $C_1$ and converge to $z_2^*$.

If $z(0, 0)\in C_3$, the vector field at the boundary of $C_3$ is similar as that shown in Figure~\ref{fig:V3}.
\begin{enumerate}
\item
If the parameters are in the range of $\theta_1-h_1<\frac{k_1}{\gamma_1}<\theta_1+h_1, 0<\frac{k_2}{\gamma_2}<\theta_2-h_2$ (see similar case shown in Figure \ref{fig:V3}(a)), when $x\in\{x : \theta_1-h_1<x_1<\theta_1^{max}, x_2=\theta_2-h_2\},$ $f_3(x)=\left[\begin{array}{cc}-\gamma_1x_1\\k_2-\gamma_2(\theta_2-h_2)\end{array}\right]$. We have that $f_3(x)$ points outside of $C_3$. When $x\in\{x : x_1=\theta_1-h_1, x_2\geq\theta_2-h_2\},$ we have $f_3(x)=\left[\begin{array}{cc}-\gamma_1(\theta_1-h_1)\\k_2-\gamma_2x_2\end{array}\right],$ which points outside $C_3.$ Depending on which jump set the solution hits, two possibilities of the equilibrium points exist. 
\begin{itemize}
\item 
If the trajectory hits the set $x\in\{x: x_1=\theta_1-h_1, x_2>\theta_2-h_2\}$, which leads to an update of $q_1,$ the solution will jump into $C_4.$ For this case, the trajectory will reach $C_1$ in finite time and converge to $z^*_2.$
\item
If the trajectory hits the set $x\in\{x: x_1>\theta_1-h_1, x_2=\theta_2-h_2\},$ $q_2$ is updated to 0 from 1, the trajectory will enter $C_2,$ and converge to the equilibrium point $z^*_1$, which is inside of $C_2.$ 
\end{itemize}
\item
If the parameters are in the range of $\theta_1-h_1<\frac{k_1}{\gamma_1}<\theta_1+h_1, \theta_2-h_2<\frac{k_2}{\gamma_2}<\theta_2+h_2$ (see similar case shown in Figure \ref{fig:V3}(b)), when $x\in\{x: x_1>\theta_1+h_1, x_2=\theta_2-h_2\},$ $f_3(x)=\left[\begin{array}{cc}-\gamma_1x_1\\k_2-\gamma_2(\theta_2-h_2)\end{array}\right]$. We have that $f_3(x)$ points inside of $C_3$. When $x\in\{x: x_1=\theta_1-h_1, x_2\geq\theta_2-h_2\},$ we have $f_3(x)=\left[\begin{array}{cc}-\gamma_1(\theta_1-h_1)\\k_2-\gamma_2x_2\end{array}\right],$ which points outside of $C_3.$ Thus, when the trajectory hits the set $\{x: x_1=\theta_1-h_1, x_2>\theta_2-h_2\}$, a jump will occur, the solution will jump into $C_4$ and finally enter $C_1$ and converge to $z^*_2.$
\end{enumerate}

}
}

\end{document}